\documentclass[11pt,english]{article}
\usepackage{mathptmx}
\usepackage[T1]{fontenc}
\usepackage[latin9]{inputenc}
\usepackage{array}
\usepackage{float}
\usepackage{booktabs}
\usepackage{mathtools}
\usepackage{enumitem}
\usepackage{amsmath}
\usepackage{amsthm}
\usepackage{amssymb}
\usepackage{stackrel}
\usepackage{graphicx}
\usepackage[a4paper]{geometry}
\usepackage{setspace}
\usepackage[authoryear]{natbib}
\makeatletter

\providecommand{\tabularnewline}{\\}

\theoremstyle{definition}
\newtheorem{defn}{\protect\definitionname}
\theoremstyle{definition}
 \newtheorem{example}{\protect\examplename}
\theoremstyle{plain}
\newtheorem{prop}{\protect\propositionname}
\theoremstyle{remark}
\newtheorem*{rem*}{\protect\remarkname}
\theoremstyle{remark}
\newtheorem{claim}{\protect\claimname}

\usepackage{algorithm,algpseudocode}

\usepackage{caption}

\numberwithin{equation}{section}

\newcommand*\dif{\mathop{}\!\mathrm{d}}

\usepackage{enumitem}  
\setlist[itemize]{noitemsep,topsep=2pt}
\setlist{noitemsep,topsep=2pt} 

\usepackage[font=small,labelfont=bf]{caption}

\makeatother

\usepackage{babel}
\providecommand{\claimname}{Claim}
\providecommand{\definitionname}{Definition}
\providecommand{\examplename}{Example}
\providecommand{\propositionname}{Proposition}
\providecommand{\remarkname}{Remark}

\begin{document}
\title{Ensembling Portfolio Strategies for Long-Term Investments: A Distribution-Free
Preference Framework for Decision-Making and Algorithms}
\author{Duy Khanh Lam\thanks{I am thankful to my advisor, Giulio Bottazzi, and Daniele Giachini
at Scuola Superiore Sant\textquoteright Anna for providing the dataset
for the experiments and their helpful comments on this paper.}\\Scuola Normale Superiore}
\maketitle
\begin{abstract}
This paper investigates the problem of ensembling multiple strategies
for sequential portfolios to outperform individual strategies in terms
of long-term wealth. Due to the uncertainty of strategies' performances
in the future market, which are often based on specific models and
statistical assumptions, investors often mitigate risk and enhance
robustness by combining multiple strategies, akin to common approaches
in collective learning prediction. However, the absence of a distribution-free
and consistent preference framework complicates decisions of combination
due to the ambiguous objective. To address this gap, we introduce
a novel framework for decision-making in combining strategies, irrespective
of market conditions, by establishing the investor's preference between
decisions and then forming a clear objective. Through this framework,
we propose a combinatorial strategy construction, free from statistical
assumptions, for any scale of component strategies, even infinite,
such that it meets the determined criterion. Finally, we test the
proposed strategy along with its accelerated variant and some other
multi-strategies. The numerical experiments show results in favor
of the proposed strategies, albeit with small tradeoffs in their Sharpe
ratios, in which their cumulative wealths eventually exceed those
of the best component strategies while the accelerated strategy significantly
improves performance.\bigskip
\end{abstract}
\begin{center}
\par\end{center}

\setlength{\abovedisplayskip}{2.5pt} 
\setlength{\abovedisplayshortskip}{2.5pt}
\setlength{\belowdisplayskip}{2.5pt} 
\setlength{\belowdisplayshortskip}{2.5pt} 

\section{Introduction}

In this paper, we investigate the problem of ensembling multiple sequential
portfolio strategies such that the combinatorial strategy could eventually
outperform all the component ones during a long-term investment. The
component strategies may include a broad range of prediction and analysis-based
strategies or even naive ones, such as the equal allocation strategy,
that an investor can observe over time. A natural problem arises in
that, due to the uncertainty of the future, the performances of the
observed strategies are not guaranteed and seem uncertain. Moreover,
many sophisticated strategies depend on a random market with specific
statistical assumptions, which are often invalid in reality, as the
future market could evolve much differently from the hypothetical
models. Hence, an investor could tackle the problem by ensembling
several strategies rather than believing in a single one, which is
probably a well-known common approach for noise reduction in machine
learning prediction, like the fusion models. Unfortunately, since
there has not been a concept to measure the preference of choice between
component strategies that do not rely on known distributions, which
is required for the traditional expected utility framework, the distribution-free
preference framework between the ensembles of collective strategies
also does not exist. In order to deal with this challenge, we primarily
put attention on well-forming a particular objective of ensembling
strategies through establishing a preference framework for making
decisions without relying on statistical assumptions on market data.\smallskip

Given the significant interest and importance of collective learning
methods and ensembles of trading or investment strategies in the literature,
particularly within practical investment contexts, there has been
a substantial volume of research, mostly experimental. To underscore
our later arguments, we briefly reference several recent articles
that follow the increasing popularity of algorithmic trading under
the umbrella of fintech and machine learning such as \citet{Nti2020},
\citet{Padhi2022}, \citet{Yu2023}, \citet{Sun2023} and \citet{Deng2024}.
The similarity in contemporary research is they only focus on prediction
and lack a uniquely theoretical framework for the formulation of an
economic agent's objective and preference. In general, the problem
of strategies ensemble can be equivalently cast as the problem of
making a portfolio of assets' allocation, which results in the inherent
uncertainty in the decision of choice. Consequently, their performances
in out-of-sample long-term investments are not secured. Moreover,
as a particular framework for making an ensemble has not been established,
the problem of an investor learning in a market with an infinitely
large amount of agents' strategies is difficult to be investigated.
These issues motivate us to propose a framework for making decisions
on ensembling various component strategies, with a clearly defined
and consistent preference over the lifetime of a market.\smallskip

\textbf{The main contributions and organization of the paper}. By
convention, the primary objective of an investor in this paper is
to create an ensemble of component strategies that eventually exceed
all of them in terms of wealth at some future periods. In Section
2, we firstly establish a model of an investor in a market who sequentially
decides combinations of various observable component strategies of
no-short portfolios over time. Then, we introduce a general framework
of distribution-free and time-consistent preference based on the cumulative
wealth of a strategy, which constitutes the investor's goal. Based
on this theoretical framework, the criterion for creating an ensemble
of strategies is specified. In Section 3, since the preferences between
strategies are defined on each sequence of unknown market data, we
propose a particular online learning strategy to determine combinations
of component strategies, ensuring the investor's preference is universally
guaranteed for any possibilities of market data. In the case of a
market with an infinitely large number of observable strategies, we
propose a modified algorithm to reduce dimensionality by learning
on a smaller set of representative strategies, thus ensuring universal
preference.\smallskip

In Section 4, we conduct experiments to numerically approximate the
proposed strategy and explore further accelerated variants by leveraging
the flexibility of the original strategy's mechanism. The proposed
strategies are tested and compared with other online learning strategies,
serving as alternative methods for ensemble multi-strategies, in terms
of empirical performances such as the Sharpe ratio, average growth
rate, and final cumulative wealth. The dataset utilized for these
experiments is sourced from The Center for Research in Security Prices,
LLC, spanning a comprehensive range of 27 years, which includes many
historically significant market events and comprises six blue-chip
stock symbols. Various experimental scenarios are conducted for both
small and large scales of component strategies, involving sequential
Mean-Variance portfolios with different risk profiles. In all experiments,
while the other multi-strategies often performed poorly, the wealth
generated by the proposed strategy eventually exceeds that of the
best component ones, albeit with slight tradeoffs in the Sharpe ratios.
Moreover, the accelerated strategy also significantly increases cumulative
wealth compared with the original one, though it is not theoretically
guaranteed but rather heuristic.

\section{Model settings and formalizations}

Let us consider a stock market of $m\geq2$ risky assets over discrete
time periods $n\in\mathbb{N}_{+}$, where $\ensuremath{\left[m\right]\coloneqq\big\{1,\ldots,m\big\}}$.
The vector $x_{n}\coloneqq\big(x_{n,1},\ldots,x_{n,m}\big)\in\mathbb{R}_{++}^{m}$
represents the assets' returns at time $n$, which is simply defined
as the ratio of asset prices at time $n$ to those at time $n-1$.
Let $x_{1}^{n}\coloneqq\big\{ x_{i}\big\}_{i=1}^{n}$ denote the finite
sequence of assets' returns over $n$ time periods, while $x_{1}^{\infty}$
represents an infinite sequence of returns over the lifetime of the
regarded market. The no-short portfolios are constrained in the simplex
$\mathcal{B}^{m}\coloneqq\big\{\beta\in\mathbb{R}^{m}:\,\sum_{j=1}^{m}\beta_{j}=1,\,\beta_{j}\geq0\big\}$
as the of portfolio choices, and the return of a portfolio $b\in\mathcal{B}^{m}$
with respect to $x_{n}$ is denoted by $\langle b,x_{n}\rangle$,
where $\langle\cdot,\cdot\rangle$ is the scalar product of two vectors.
For simplification and generality, the infinite sequence $x_{1}^{\infty}$
is treated as deterministic but unknown data, rather than as random,
throughout this paper.\smallskip

At each $n$, let $b_{n}:\mathbb{R}_{++}^{m\times(n-1)}\to\mathcal{B}^{m}$
denote the choice of portfolio based on $x_{1}^{n-1}$, and let the
\emph{strategy} associated with these portfolios be denoted by the
infinite sequence $\big(b_{n}\big)\coloneqq\left\{ b_{n}\right\} _{n=1}^{\infty}\in\mathcal{B}^{m\times\infty}$;
accordingly, a \textit{constant strategy} of a fixed portfolio $b$
is denoted as solely $\big(b\big)$, without a timing index. Given
the initial capital $S_{0}\big(b_{0}\big)\eqqcolon S_{0}=1$, assuming
there are no commission fees, the \textit{cumulative wealth} and its
corresponding \textit{\emph{exponential average }}\textit{growth rate}
after $n$ periods of investment using a generic strategy $\left(b_{n}\right)$
are defined respectively as:
\[
S_{n}\left(b_{n}\right)\coloneqq{\displaystyle {\displaystyle \prod_{i=1}^{n}}\left\langle b_{i},x_{i}\right\rangle }\text{ and }W_{n}\left(b_{n}\right)\coloneqq\dfrac{1}{n}\log S_{n}\left(b_{n}\right).
\]
By employing the definitions of cumulative wealth and growth rate,
we propose a conceptual framework that generalizes the evolution of
a portfolio strategy in long-term investment, as defined in Definition
\ref{Evolution of strategies} below, and also provide an illustration
in Example \ref{exp evolution of strategies}.
\begin{defn}
A strategy $\big(\hat{b}_{n}\big)$ is said to \textit{almost always
exceed} another strategy $\big(\bar{b}_{n}\big)$ if there exists
$M>1$ such that $S_{n}\big(\hat{b}_{n}\big)>S_{n}\big(\bar{b}_{n}\big)$
for all $n\geq M$ (the term \textquotedblleft almost\textquotedblright{}
can be omitted for $M=1$).\label{Evolution of strategies}
\end{defn}
As the negation of the Definition \ref{Evolution of strategies},
the strategy $\big(\hat{b}_{n}\big)$ does not always exceed the strategy
$\big(\bar{b}_{n}\big)$ if for any $M\geq1$, there exists $n_{M}\geq M$
such that $S_{n_{M}}\big(\hat{b}_{n_{M}}\big)\leq S_{n_{M}}\big(\bar{b}_{n_{M}}\big)$;
moreover, if there are infinitely many $n_{M}$ such that the strict
inequality holds, the strategy $\big(\hat{b}_{n}\big)$ can be said
to be \textit{infinitely often exceeded} by the strategy $\big(\bar{b}_{n}\big)$.
This way of defining the relative performance in terms of the growth
of strategies is general in the sense that it does not rely on any
specific statistical properties of the assets' returns. Hence, the
cumulative wealth of a strategy will either be almost always exceeded
or not by that of another, regardless of the behavior of the infinite
sequence $x_{1}^{\infty}$.
\begin{example}
Figure \ref{figure 1} illustrates an example of the evolutions of
some strategies taken from the experiment section. The left graphic
shows two strategies often exceeding each other, while the right graphic
shows a strategy almost always exceeding the remainder, except for
a few observations around time $n=3400$, where they are temporarily
equivalent.\label{exp evolution of strategies}
\begin{figure}[H]
\begin{centering}
\includegraphics[scale=0.33]{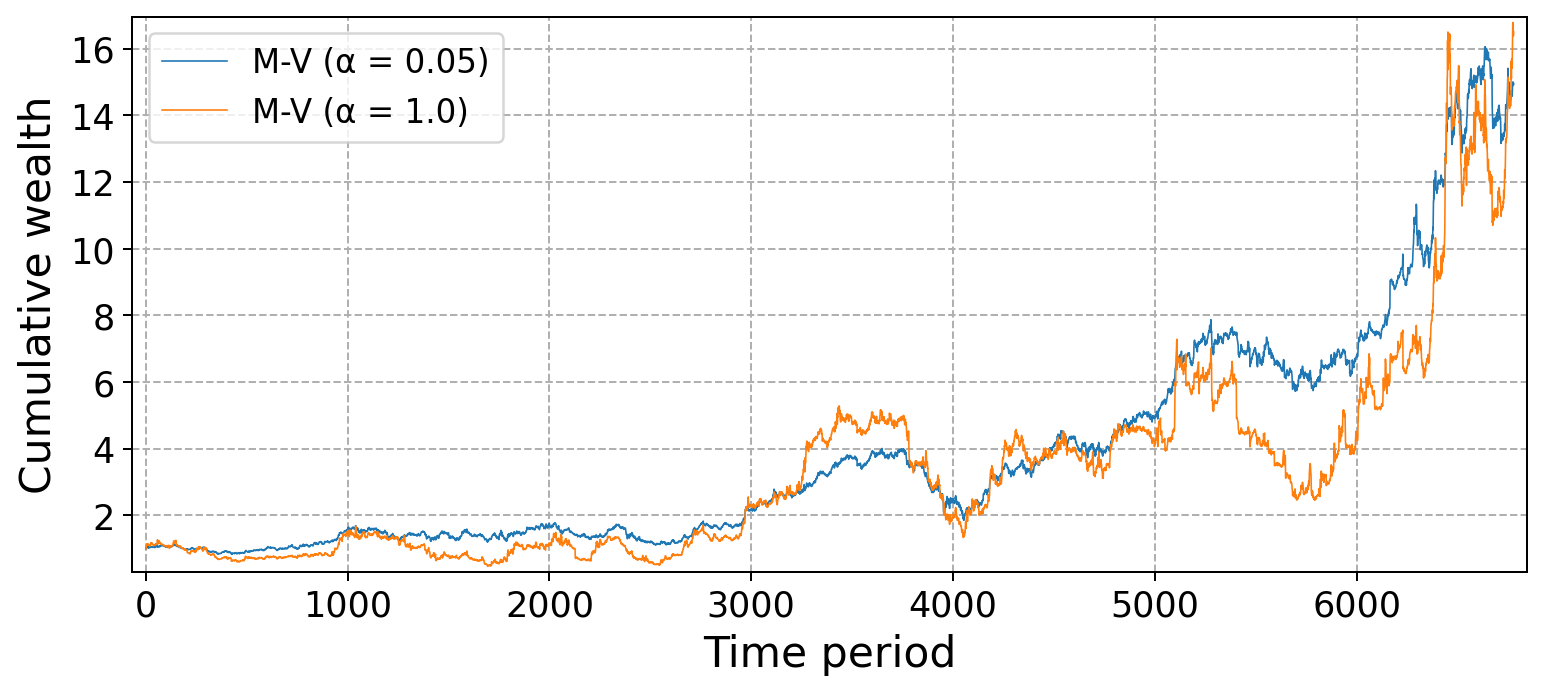}\includegraphics[scale=0.33]{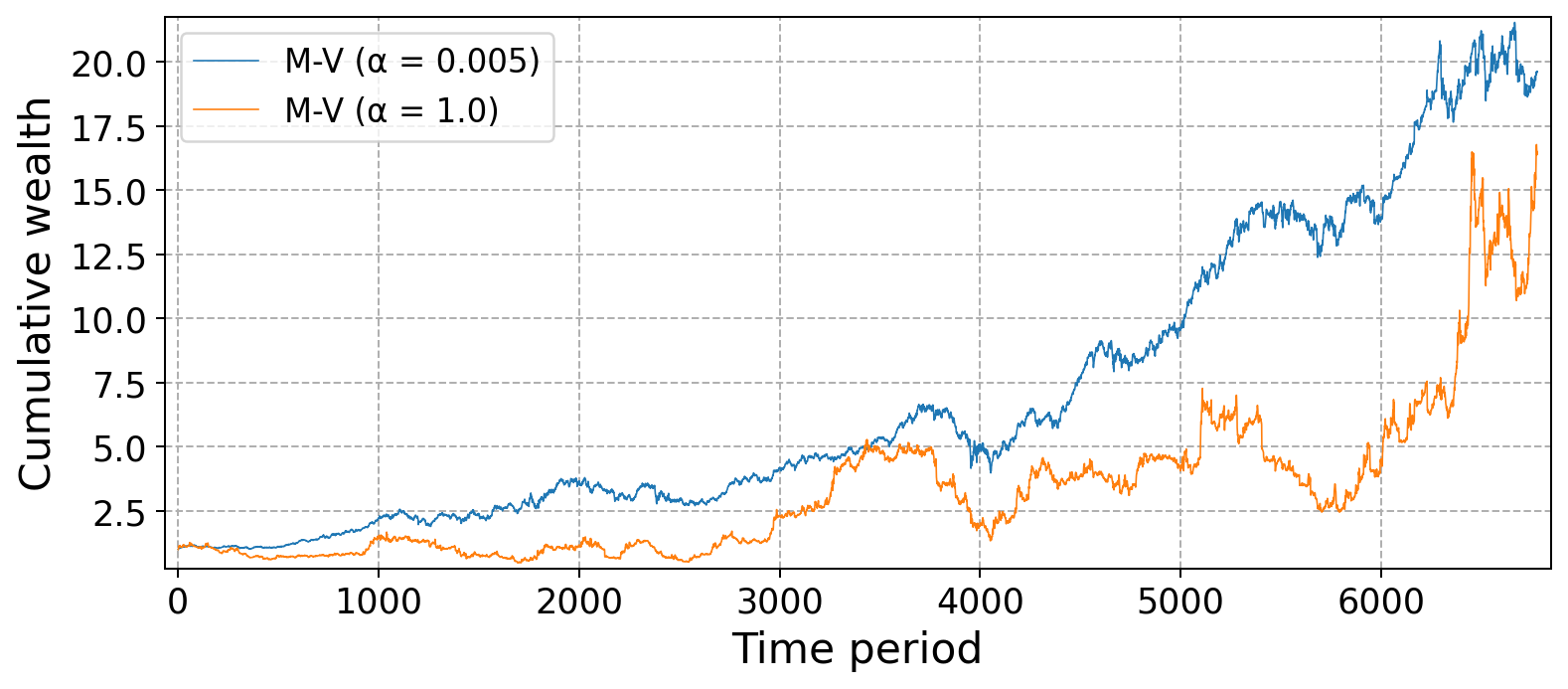}
\par\end{centering}
\centering{}{\footnotesize\caption{Evolutions of component M-V strategies over time with different levels
of risk aversion.\label{figure 1}}
}{\footnotesize\par}
\end{figure}
\end{example}
The main point of introducing the evolutionary framework is that a
strategy which is not always exceeded by other strategies will eventually
exceed them after a finite number of time periods from any observation
point. However, although using Definition \ref{Evolution of strategies}
alone is able to distinguish any pairs of strategies, it seems not
convenient for formalizing a preference ranking essential for constructing
an ensemble. In the next sections, we propose an alternative but closely
related preference relation framework which constitutes a criterion
for strategy combination corresponding to the investor's objective.
Then, we demonstrate its link to the evolution of cumulative wealth
of strategy, wherein a strategy cannot always exceed another one that
is strictly preferred.

\subsection{Preference relation and combination criterion}

Motivated by the fact that the amount of market data is infinitely
growing rather than fixed and wholly available, we need a \textit{preference
relation} that demonstrates a stable ranking over time to compare
any two strategies, which implies that the relation does not change
as the assets' returns are gradually observed as time evolves. For
this reason, such a preference should be defined over all the strategies
in a single infinite-dimensional space $\mathcal{B}^{m\times\infty}$.
The following Definition \ref{Definition of preference} introduces
an appropriate preference in an asymptotic sense, posing consistency
for the ranking.
\begin{defn}
\begin{singlespace}
\noindent Given an infinite sequence of assets' returns $x_{1}^{\infty}$
and two generic strategies $\big(\hat{b}_{n}\big)$, $\big(\bar{b}_{n}\big)$,
the preference relations along with their corresponding notations
are defined as follows:\label{Definition of preference}
\end{singlespace}
\begin{enumerate}
\item $\big(\bar{b}_{n}\big)\sim\big(\hat{b}_{n}\big)$: Neither one of
them is preferred to the other, or they are indifferent, as
\[
\big(\bar{b}_{n}\big)\sim\big(\hat{b}_{n}\big)\Leftrightarrow\lim_{n\to\infty}\big(W_{n}\big(\hat{b}_{n}\big)-W_{n}\big(\bar{b}_{n}\big)\big)=0.
\]
\item $\big(\bar{b}_{n}\big)\succsim\big(\hat{b}_{n}\big)$: $\big(\bar{b}_{n}\big)$
is weakly preferred to, or at least as good as, $\big(\hat{b}_{n}\big)$,
as 
\[
\big(\bar{b}_{n}\big)\succsim\big(\hat{b}_{n}\big)\Leftrightarrow\limsup_{n\to\infty}\big(W_{n}\big(\hat{b}_{n}\big)-W_{n}\big(\bar{b}_{n}\big)\big)\leq0.
\]
\item $\big(\bar{b}_{n}\big)\succ\big(\hat{b}_{n}\big)$: $\big(\bar{b}_{n}\big)$
is strictly preferred to, or better than, $\big(\hat{b}_{n}\big)$,
as
\[
\big(\bar{b}_{n}\big)\succ\big(\hat{b}_{n}\big)\Leftrightarrow\big(\bar{b}_{n}\big)\succsim\big(\hat{b}_{n}\big)\land\big(\bar{b}_{n}\big)\nsim\big(\hat{b}_{n}\big).
\]
\end{enumerate}
\noindent Hence, the subsequent logic deduction follows the definition
of relations as follows:
\[
\big(\bar{b}_{n}\big)\succ\big(\hat{b}_{n}\big)\veebar\big(\bar{b}_{n}\big)\sim\big(\hat{b}_{n}\big)\Leftrightarrow\big(\bar{b}_{n}\big)\succsim\big(\hat{b}_{n}\big)\text{ and }\big(\bar{b}_{n}\big)\succsim\big(\hat{b}_{n}\big)\land\big(\hat{b}_{n}\big)\succsim\big(\bar{b}_{n}\big)\Leftrightarrow\big(\bar{b}_{n}\big)\sim\big(\hat{b}_{n}\big).
\]

\end{defn}
The preference relations defined above provide a consistent ranking
for the investor, who prefers cumulative wealth, to evaluate utility
for any strategy through its equivalent performance metric, the average
growth rate. In an asymptotic sense with a timing factor, two similarly
preferred strategies might not have the same growth rate all the time,
but their difference is negligible as time evolves. This is more general
than the usual sense of indifference in preference without a timing
factor. Akin to a conventional preference relation, the proposed preference
relation also possesses the following properties of a rational ranking,
according to the fundamental axioms.
\begin{prop}
For any infinite sequence of assets' returns $x_{1}^{\infty}$, the
preference relations established in Definition \ref{Definition of preference}
satisfy the following properties:\label{Axioms}
\begin{enumerate}
\item Completeness: $\big(\hat{b}_{n}\big)\succsim\big(\bar{b}_{n}\big)\lor\big(\bar{b}_{n}\succsim\big(\hat{b}_{n}\big)$,
$\forall\big(\hat{b}_{n}\big),\big(\bar{b}_{n}\big)\in\mathcal{B}^{m\times\infty}.$ 
\item Reflexivity: $\big(\hat{b}_{n}\big)\succsim\big(\hat{b}_{n}\big)$,
$\forall\big(\hat{b}_{n}\big)\in\mathcal{B}^{m\times\infty}$.
\item Transitivity: $\big(\hat{b}_{n}\big)\succsim\big(\bar{b}_{n}\big)\land\big(\bar{b}_{n}\big)\succsim\big(\tilde{b}_{n}\big)\Rightarrow\big(\hat{b}_{n}\big)\succsim\big(\tilde{b}_{n}\big)$,
$\forall\big(\hat{b}_{n}\big),\big(\bar{b}_{n}\big),\big(\tilde{b}_{n}\big)\in\mathcal{B}^{m\times\infty}$.
\item Asymmetry: $\big(\hat{b}_{n}\big)\succ\big(\bar{b}_{n}\big)\Rightarrow\big(\bar{b}_{n}\big)\not\succsim\big(\hat{b}_{n}\big)$,
$\forall\big(\hat{b}_{n}\big),\big(\bar{b}_{n}\big)\in\mathcal{B}^{m\times\infty}$.
\end{enumerate}
\end{prop}
\begin{proof}
All properties could be checked straightforwardly using Definition
\ref{Definition of preference}.
\end{proof}
Given a set of strategies, since the infinite sequence $x_{1}^{\infty}$
is unknown, their evolutions of cumulative wealths are not definitively
determined at any time. Thus, any competitive strategy can only guarantee
to be at least as good as the best strategy at some future particular
observation periods. In order to create such a competitive strategy
conveniently and comprehensively, it is necessary to define a single
\textit{baseline strategy} against which it will compete. This baseline
strategy, representing the best possible cumulative wealth among the
observed strategies, is formulated based on the realizations of assets'
returns during the respective time periods.
\begin{defn}
(Baseline strategy). For an infinite sequence of assets' returns $x_{1}^{\infty}$
and a set of strategies $\big\{\big(b_{n}^{\alpha}\big):\alpha\in\left[k\right]\coloneqq\big\{1,\ldots,k\big\}\big\}$,
the baseline strategy, denoted as $\big(\ddot{b}_{n}\big)$ hereafter,
is defined as the best merged strategy of all the component ones,
where $S_{n}\big(\ddot{b}_{n}\big)\coloneqq\max_{\alpha\in[k]}S_{n}\big(b_{n}^{\alpha}\big)$
for all $n$.\label{Baseline strategy} 
\end{defn}
By definition, the baseline strategy $\big(\ddot{b}_{n}\big)$ is
the envelope of all cumulative wealth of observed strategies over
time, so $\big(\ddot{b}_{n}\big)\succsim\big(b_{n}^{\alpha}\big)$
for all $\alpha\in\left[k\right]$. Since the construction of the
baseline strategy relies on future data of assets' returns at every
time period, it is inherently inaccessible to the investor. However,
the investor could instead create a multi-strategy (or an ensemble),
denoted as $\big(b_{n}\big)$ and termed the \textit{combinatorial
strategy} henceforth, to compete with the baseline strategy. Despite
the inaccessibility of the baseline strategy $\big(\ddot{b}_{n}\big)$,
it facilitates us to evaluate the preference between the combinatorial
strategy and a single consistent strategy in the infinite-dimensional
space $\mathcal{B}^{m\times\infty}$, thereby ensuring that the investor's
objective is met. To clarify this point, it is simple but worth noting
that if $\big(b_{n}\big)\succ\big(\ddot{b}_{n}\big)$, then $\big(b_{n}\big)\succ\big(b_{n}^{\alpha}\big)$
for all $\alpha\in\left[k\right]$, due to the transitivity and asymmetry
of preference, but the reverse is not necessarily true. Since strict
preference is favored over a weak one, let us establish a specific
criterion for combining the strategies.
\begin{prop}
\textup{(Combination criterion)}. Given an infinite sequence of assets'
returns $x_{1}^{\infty}$, the baseline strategy $\big(\ddot{b}_{n}\big)$
cannot always exceed a combinatorial strategy $\big(b_{n}\big)$ if
$\big(b_{n}\big)\succ\big(\ddot{b}_{n}\big)$.\label{Combination-criterion}
\end{prop}
\begin{proof}
See Appendix.
\end{proof}
\begin{rem*}
The reverse reasoning of the Combination criterion is not correct
since the baseline strategy could always exceed the combinatorial
strategy over time despite $\big(b_{n}\big)\succsim\big(\ddot{b}_{n}\big)$.
For instance:
\[
1<S_{n}\big(\ddot{b}_{n}\big)\big/S_{n}\big(b_{n}\big)<2,\forall n\Rightarrow\big(b_{n}\big)\sim\big(\ddot{b}_{n}\big)\Rightarrow\big(b_{n}\big)\nsucc\big(\ddot{b}_{n}\big),
\]
due to the asymmetry of preference, so the weak preference is not
sufficient for a combination.
\end{rem*}
According to the Combination criterion, the strict preference $\big(b_{n}\big)\succ\big(\ddot{b}_{n}\big)$
implies that the cumulative wealth of the combinatorial strategy will
eventually match or even surpass all those of the component strategies
simultaneously after a finite number of time periods. Notably, it
also includes the possibility of the combinatorial strategy almost
always exceeding the baseline strategy. Furthermore, this guarantee
does not rely on any statistical assumptions and is also not biased
toward the selection of the starting point of observations\footnote{In the literature, investment strategies and algorithms for ensembling
them are commonly tested on a specific dataset, and the strategy with
the highest final wealth is considered the best one. However, the
choice of a dataset for out-of-sample testing can be biased towards
certain strategies and does not account for longer horizons. Therefore,
the Combination criterion generalizes a target for ensembling the
component strategies. By following a combinatorial strategy that is
not always exceeded by all other strategies concurrently, the investor
can be confident of eventually accumulating the highest final wealth
at certain points in the future.}. Additionally, as an example to reiterate the significance and convenience
of the baseline strategy as mentioned above, it is easy to see that
creating a combinatorial strategy such that $\big(b_{n}\big)\succ\big(b_{n}^{\alpha}\big)$
for all $\alpha\in\left[k\right]$ does not guarantee that it will
exceed all component strategies simultaneously.\smallskip

\textbf{The universality of a preference under unknown future assets'
returns}. Since the preferences of strategies depend on an unknown
specific sequence $x_{1}^{\infty}$, the following Claim \ref{Consistent-universal-preference}
formalizes the universality of a preference over all possibilities
of $x_{1}^{\infty}$ in order to decide a choice between two strategies.
Generally, a weak preference is \textit{more universal} than the reverse
one if the strict relation holds for some possibilities of $x_{1}^{\infty}$
while the indifferent relation holds for the remaining others.
\begin{claim}
Consider two generic strategies $\big(\hat{b}_{n}\big)$ and $\big(\bar{b}_{n}\big)$
and the infinite-dimensional distribution space $\mathbb{R}_{++}^{m\times\infty}$
of all instances of the sequence $x_{1}^{\infty}$. By the completeness
and asymmetry of preference, if $\big(\bar{b}_{n}\big)\succsim\big(\hat{b}_{n}\big)$
everywhere in $\mathbb{R}_{++}^{m\times\infty}$ while $\big(\hat{b}_{n}\big)\succsim\big(\bar{b}_{n}\big)$
only in a non-empty set $D\subset\mathbb{R}_{++}^{m\times\infty}$,
then $\big(\bar{b}_{n}\big)\succ\big(\hat{b}_{n}\big)$ in $\mathbb{R}_{++}^{m\times\infty}\backslash D$.
In this context, an economically rational investor in the market universally
prefers the strategy $\big(\bar{b}_{n}\big)$ over the strategy $\big(\hat{b}_{n}\big)$,
or the preference $\big(\bar{b}_{n}\big)\succsim\big(\hat{b}_{n}\big)$
is said to be more universal than the preference $\big(\hat{b}_{n}\big)\succsim\big(\bar{b}_{n}\big)$.
This is reasonable by the fact that the possibility of the strategy
$\big(\hat{b}_{n}\big)$ being strictly preferred does not exist.
In the case $\big(\bar{b}_{n}\big)\sim\big(\hat{b}_{n}\big)$ everywhere
in $\mathbb{R}_{++}^{m\times\infty}$, they are said to be universally
indifferent.\label{Consistent-universal-preference}
\end{claim}

\section{A construction for combinatorial strategy}

\subsection{The benchmark strategy}

In this section, a \textit{benchmark strategy}, denoted as $\big(b_{n}^{*}\big)$,
is formed using the concept in Definition \ref{def 4} such that the
preference $\big(b_{n}^{*}\big)\succsim\big(\ddot{b}_{n}\big)$ is
more universal than the reverse one $\big(\ddot{b}_{n}\big)\succsim\big(b_{n}^{*}\big)$
in the sense of Claim \ref{Consistent-universal-preference}. Although
the strategy $\big(b_{n}^{*}\big)$ is unattainable as it relies on
future data like the baseline $\big(\ddot{b}_{n}\big)$, we could
create a combinatorial strategy $\big(b_{n}\big)$ using only past
data such that the preference $\big(b_{n}\big)\sim\big(b_{n}^{*}\big)$
holds everywhere in $\mathbb{R}_{++}^{m\times\infty}$, making $\big(b_{n}\big)\succsim\big(\ddot{b}_{n}\big)$
more universal than the reverse preference by the transitivity of
the preference relation.
\begin{defn}
Given a set of component strategies $\big\{\big(b_{n}^{\alpha}\big):\alpha\in\left[k\right]\big\}$,
for any \textit{constant combination} $\lambda\coloneqq\left(\lambda_{1},...,\lambda_{k}\right)\in\mathcal{B}^{k}$\textit{,
}a corresponding\textit{ constant combinatorial strategy} $\big(b_{n}^{\lambda}\big)$
is formed. Subsequently, the best constant combination $\lambda_{n}\coloneqq\left(\lambda_{n,1},...,\lambda_{n,k}\right)$
at time $n$ is determined as follows:\label{def 4}
\[
\lambda_{n}\coloneqq\operatorname*{argmax}_{\lambda\in\mathcal{B}^{k}}S_{n}\big(b_{n}^{\lambda}\big),\text{ where }b_{n}^{\lambda}\coloneqq{\displaystyle \sum_{\alpha=1}^{k}}\lambda_{\alpha}b_{n}^{\alpha},\,\forall n.
\]
\end{defn}
Each constant combinatorial strategy allocates fixed proportions of
capital into the corresponding component strategies. By utilizing
this concept, we establish a benchmark strategy as the best merger
in hindsight of all component strategies $\big(b_{n}^{\lambda}\big)$
similarly to Definition \ref{Baseline strategy}.
\begin{defn}
(Benchmark strategy). The benchmark strategy, represented by $\big(b_{n}^{*}\big)$
henceforth, is defined as one satisfying $S_{n}\big(b_{n}^{*}\big)\coloneqq S_{n}\big(b_{n}^{\lambda_{n}}\big)$
for all $n$, which is simply the best-merged strategy of all component
constant combinatorial strategies $\big(b_{n}^{\lambda}\big)$.\label{(Benchmark-strategy)}
\end{defn}
Definition \ref{(Benchmark-strategy)} immediately results in $W_{n}\big(b_{n}^{*}\big)\geq W_{n}\big(\ddot{b}_{n}\big)$,
so the preference $\big(b_{n}^{*}\big)\succsim\big(\ddot{b}_{n}\big)$
holds for any sequence $x_{1}^{\infty}$. If $\big(b_{n}^{*}\big)\succ\big(\ddot{b}_{n}\big)$
for some sequences in $\mathbb{R}_{++}^{m\times\infty}$, any combinatorial
strategy satisfying $\big(b_{n}\big)\sim\big(b_{n}^{*}\big)$ implies
$\big(b_{n}\big)\succ\big(\ddot{b}_{n}\big)$ due to the transitivity
and asymmetry of the preference relation, thus meeting the Combination
criterion. Therefore, it is sufficient to secure the preference $\big(b_{n}^{*}\big)\succ\big(\ddot{b}_{n}\big)$
for a sequence $x_{1}^{\infty}$ by verifying if $\big(W_{n}\big(b_{n}^{*}\big)-W_{n}\big(\ddot{b}_{n}\big)\big)\nrightarrow0$
as $n\to\infty$ due to $\big(b_{n}\big)\succsim\big(\ddot{b}_{n}\big)\land\big(b_{n}\big)\nsim\big(\ddot{b}_{n}\big)\Leftrightarrow\big(b_{n}\big)\succ\big(\ddot{b}_{n}\big)$.\smallskip

By definition, the constant combination mechanism exploits the inherent
extent of difference in returns between the component strategies.
If there is a significant discrepancy in the behaviors of capital
growth among the component strategies, the maximizers $\lambda_{n}$
will not consistently remain close to the vertices as time develops.
As a result, the benchmark strategy will infinitely often exceed the
baseline strategy, and so $W_{n}\big(b_{n}^{*}\big)>W_{n}\big(\ddot{b}_{n}\big)$
at infinitely many $n$. The following proposition provides a basis
for explaining the behavior of the benchmark strategy.
\begin{prop}
Given a set of strategies $\big\{\big(b_{n}^{\alpha}\big):\alpha\in\left[k\right]\big\}$,
suppose there exist $h$ disjoint and non-empty sets indexed as $\mathcal{A}^{j}$,
$\bigcup_{j\in\left[h\right]}\mathcal{A}^{j}=\left[k\right]$, and
a set $\mathcal{H}\subseteq\left[k\right]$ such that $\left|\mathcal{H}\cap\mathcal{A}^{j}\right|=1$
for all $j\in\left[h\right]\coloneqq\big\{1,...,h\big\}$, which satisfy
the following:
\[
\big<b_{n}^{\alpha^{j}},x_{n}\big>\geq\big<b_{n}^{\alpha},x_{n}\big>,\text{ \ensuremath{\forall\alpha\in\mathcal{A}^{j}},\,\ensuremath{\forall n}}\text{, where }\left\{ \alpha^{j}\right\} \coloneqq\mathcal{H}\cap\mathcal{A}^{j},j\in\left[h\right].
\]
Let $\lambda\in\mathcal{B}^{k}$ and $\gamma\in\mathcal{B}^{h}$ denote
the constant combinations of the component strategies in $\mathcal{A}$
and $\mathcal{H}$, respectively, then $S_{n}\big(b_{n}^{\lambda_{n}}\big)=S_{n}\big(b_{n}^{\gamma_{n}}\big)$
for all $n$.\label{S(lamda Y) =00003D S(lamda gamma)} 
\end{prop}
\begin{proof}
See Appendix.
\end{proof}
Proposition \ref{S(lamda Y) =00003D S(lamda gamma)} states that if
the set of component strategies can be separated into $h$ subsets,
such that each subset contains a single strategy that never yields
lower returns than the other strategies in the same subset, then the
cumulative wealth of the baseline strategy will always be equal to
that of the best merged one of only $h$ dominating strategies. Consequently,
the maximizers $\lambda_{n}$ have only $h$ nonzero components regardless
of the number of component strategies. Hence, if there exists a single
strategy that generates higher returns than all others after a finite
number of time periods, the maximizers $\lambda_{n}$ will remain
close to a vertex of $\mathcal{B}^{k}$ as time develops, no matter
how the remainders behave, so $\big(W_{n}\big(b_{n}^{*}\big)-W_{n}\big(\ddot{b}_{n}\big)\big)\rightarrow0$.

\subsection{Small scale of component strategies}

In this section, we consider the case of a small number $k\geq2$
for the set $\big\{\big(b_{n}^{\alpha}\big):\alpha\in\left[k\right]\big\}$
of component strategies. After introducing a specific benchmark strategy
$\big(b_{n}^{*}\big)$ in the previous section, we propose now a combinatorial
strategy $\big(b_{n}\big)$, which does not depend on future data,
such that $\big(b_{n}\big)\sim\big(b_{n}^{*}\big)$ for all instances
of sequences in $\mathbb{R}_{++}^{m\times\infty}$. This goal can
be accomplished by demonstrating that the distance $\left|\log S_{n}\big(b_{n}^{*}\big)-\log S_{n}\big(b_{n}\big)\right|$
in the worst case of sequence $x_{1}^{n}$ is sublinear in the time
variable $n$. To this end, we propose a combinatorial strategy as
follows:\smallskip

\textbf{Strategy proposal}: At each time $n$, the investor makes
the following combination after observing the past cumulative wealths
$S_{n-1}\big(b_{n-1}^{\alpha}\big)$ of $k$ component strategies:
\begin{equation}
b_{1}\coloneqq{\displaystyle \dfrac{1}{k}{\displaystyle \sum_{i=1}^{k}b_{1}^{i}}}\text{ and }{\displaystyle b_{n}\coloneqq{\displaystyle \dfrac{{\displaystyle \int_{\mathcal{B}^{k}}{\displaystyle b_{n}}}S_{n-1}\big(b_{n-1}^{\lambda}\big)\mu\left(\lambda\right)\dif\lambda}{{\displaystyle \int_{\mathcal{B}^{k}}S_{n-1}\big(b_{n-1}^{\lambda}\big)\mu\left(\lambda\right)\dif\lambda}}},\,\forall n\geq2,}\label{eq:3.3}
\end{equation}
where $\mu\left(\lambda\right)$ is the probability density with respect
to the variable $\lambda\in\mathcal{B}^{k}$, and $\big(b_{n}^{\lambda}\big)$
is the strategy corresponding to a constant combination $\lambda$
of the component strategies.\smallskip

If the density $\mu\left(\lambda\right)$ is uniform, the proposed
combinatorial strategy given by (\ref{eq:3.3}) is always exceeded
by the benchmark strategy due to:
\begin{align}
S_{n}\big(b_{n}\big)={\displaystyle \prod_{i=1}^{n}\dfrac{{\displaystyle \int_{\mathcal{\mathcal{B}}^{k}}\big<{\displaystyle b_{i}^{\lambda}},x_{i}\big>}S_{i-1}\big({\displaystyle b_{i-1}^{\lambda}\big)}\mu\left(\lambda\right)\dif\lambda}{{\displaystyle \int_{\mathcal{\mathcal{B}}^{k}}S_{i-1}\big({\displaystyle b_{i-1}^{\lambda}\big)}\mu\left(\lambda\right)\dif\lambda}}} & {\displaystyle =\int_{\mathcal{\mathcal{B}}^{k}}S_{n}\big({\displaystyle b_{n}^{\lambda}\big)}\mu\left(\lambda\right)\dif\lambda\text{ \,(by telescoping)}}\nonumber \\
 & <S_{n}\big(b_{n}^{\lambda_{n}}\big)\int_{\mathcal{\mathcal{B}}^{k}}\mu\left(\lambda\right)\dif\lambda=S_{n}\big(b_{n}^{*}\big),\,\forall n.\label{eq:3.5}
\end{align}
In this case, the combinatorial strategy is equivalent to the \textit{exponentially
weighted average forecaster} with the learning rate parameter of one
and the logarithmic loss function, as established in \citet{CesaBianchi1997},
\citet{CesaBianchi1999,CesaBianchi2006}. Proposition \ref{finite strategy},
given below, establishes an upper bound on the positive difference
$\log S_{n}\big(b_{n}^{*}\big)-\log S_{n}\big(b_{n}\big)$ in the
worst case of the sequence $x_{1}^{n}$. It follows from the theorem
for the minimax problem shown in \citet{Cover1991}, \citet{Cover1996a},
\citet{Cover2006} and \citet{CesaBianchi2006}\footnote{Within the framework of \textit{Prediction with} \textit{expert advice},
as discussed in \citet{CesaBianchi1997}, \citet{CesaBianchi1999,CesaBianchi2006},
the Universal portfolio is equivalent to the Laplace mixture forecaster,
studied in \citet{Davisson1973}, \citet{Rissanen1986}, for the case
of uniform distribution, and to the Krichevsky\textendash Trofimov
mixture forecaster, proposed by \citet{Krichevsky1981}, for the Dirichlet
distribution in the Kelly betting market. It has also been investigated
within another umbrella of \textit{Online portfolio selection}, as
in\textbf{ }\citet{Ordentlich1996}, \citet{Helmbold1998}, \citet{Blum1999},
\citet{Fostera1999},\textbf{ }\citet{Vovk1998}, \citet{Erven2020}.}.
\begin{prop}
If the combinatorial strategy $\big(b_{n}\big)$ is constructed according
to (\ref{eq:3.3}) with the uniform probability density $\mu\left(\lambda\right)$
over the simplex $\mathcal{\mathcal{B}}^{k}$, then:
\[
{\displaystyle \max_{x_{1}^{n}}\left(\log S_{n}\left(b_{n}^{*}\right)-\log S_{n}\left(b_{n}\right)\right)\leq\left(k-1\right)\log\left(n+1\right),\text{ \ensuremath{\forall n}}},
\]
where $\big(b_{n}^{*}\big)$ is the benchmark strategy of all component
strategies in the set $\big\{\big(b_{n}^{\alpha}\big):\alpha\in\left[k\right]\big\}$.\label{finite strategy}
\end{prop}
\begin{proof}
See Appendix.
\end{proof}
Due to Proposition \ref{finite strategy} and the strict inequality
(\ref{eq:3.5}), ${\displaystyle \big({\textstyle W_{n}\big(b_{n}^{*}\big)-W_{n}\big(b_{n}\big)}\big)}\to0$,
i.e., $\big(b_{n}^{*}\big)\sim\big(b_{n}\big)$, for any sequence
$x_{1}^{\infty}$ by the sandwich theorem. Note that within the framework
of online portfolio optimization\textit{ }, there exist several other
algorithms as surveyed in \citet{Erven2020}, that satisfy the sublinear
upper bound for $\big(\log S_{n}\big(b_{n}^{*}\big)-\log S_{n}\big(b_{n}\big)\big)$.
However, despite their computational efficiency, these algorithms
rely on certain assumptions, which limit their guarantee in terms
of bounding the minimax problem that is required to theoretically
guarantee the preference $\big(b_{n}^{*}\big)\sim\big(b_{n}\big)$
for all sequences $x_{1}^{\infty}$. Moreover, our intentional choice
of the construction for a combinatorial strategy according to (\ref{eq:3.3})
is to allow further flexibility in its algorithm. In the experiment
section, we propose an accelerated variant for the combinatorial strategy
using a time-variant density instead of the fixed uniform one. This
acceleration method modifies the density by adaptively changing its
support over time, enabling the combinatorial strategy to even exceed
the cumulative wealth of the benchmark strategy.

\subsection{Arbitrarily large scale of component strategies}

When the scale of the set $\big\{\big(b_{n}^{\alpha}\big):\alpha\in\left[k\right]\big\}$
becomes larger, the computation of the proposed combinatorial strategy
becomes increasingly difficult, especially at relatively large scales.
Furthermore, the limit ${\displaystyle \big({\textstyle W_{n}\big(b_{n}^{*}\big)-W_{n}\big(b_{n}\big)}\big)}\to0$
is generally not guaranteed as $k\to\infty$, since the upper bound
established in Proposition \ref{finite strategy} is linear in $k$,
meaning that $\left(k-1\right)\log\left(n+1\right)\to\infty$. Additionally,
when $k\to\infty$, a convex combination of all the component portfolios
results in atomic weights, rendering the combination meaningless.
As a result, the combinatorial strategy reduces to nothing more than
an exponentially weighted average of all component strategies, as
formed by the formula (\ref{eq:WAE}) in the experiment section, thereby
failing to meet the Combination criterion, as $\big(b_{n}^{*}\big)\sim\big(\ddot{b}_{n}\big)$
for any sequence $x_{1}^{\infty}$. To address this issue, we propose
a modified combinatorial strategy below.\smallskip

\textbf{Strategy proposal}: First of all, we form a finite number
of base sets by partitioning an arbitrarily large set $\left[k\right]$
into $N$ smaller disjoint and nonempty sets, denoted as $\mathcal{A}^{j}$,
$j\in\left[N\right]\coloneqq\left\{ 1,...,N\right\} $, $\bigcup_{j\in\left[N\right]}\mathcal{A}^{j}=\left[k\right]$.
At each time period $n$, the investor decides the combination as
follows:
\begin{equation}
b_{1}\coloneqq{\displaystyle \dfrac{1}{N}{\displaystyle \sum_{i=1}^{N}\hat{b}_{1}^{i}}}\text{ and }b_{n}\coloneqq{\displaystyle \dfrac{{\displaystyle \int_{\mathcal{B}^{N}}{\displaystyle b_{n}}}S_{n-1}\big(b_{n-1}^{\lambda}\big)\mu\left(\lambda\right)\dif\lambda}{{\displaystyle \int_{\mathcal{B}^{N}}S_{n-1}\big(b_{n-1}^{\lambda}\big)\mu\left(\lambda\right)\dif\lambda}}},\,\forall n\geq2,\label{eq:3.6}
\end{equation}
where $\mu\left(\lambda\right)$ is the probability density with respect
to the variable $\lambda\in\mathcal{\mathcal{B}}^{N}$, and $\big({\displaystyle b_{n}^{\lambda}\big)}$
denotes the strategy corresponding to a constant combination $\lambda$
of all $N$ \textit{representative} \textit{strategies} $\big(\hat{b}_{n}^{i}\big)$.
In detail, each constant combinatorial portfolio $b_{n}^{\lambda}$
is constructed as follows:
\[
b_{n}^{\lambda}={\displaystyle \sum_{i=1}^{N}\lambda_{i}}\hat{b}_{n}^{i},\,\forall n,
\]
where the representative portfolios $\hat{b}_{n}^{i}$ are constructed
as:
\begin{equation}
\hat{b}_{1}^{i}\coloneqq{\displaystyle \sum_{\alpha\in\mathcal{A}^{i}}{\displaystyle b_{1}^{\alpha}}}\mu^{i}\big(\alpha\big),\,\forall i\in\left[N\right]\text{ and }\hat{b}_{n}^{i}\coloneqq\dfrac{{\displaystyle {\displaystyle \sum_{\alpha\in\mathcal{A}^{i}}}{\displaystyle b_{n}^{\alpha}}}S_{n-1}\left(b_{n-1}^{\alpha}\right)\mu^{i}\left(\alpha\right)}{{\displaystyle \sum_{\alpha\in\mathcal{A}^{i}}}{\displaystyle S_{n-1}\left(b_{n-1}^{\alpha}\right)\mu^{i}\left(\alpha\right)}},\forall n\geq2,i\in\left[N\right],\label{eq:3.7}
\end{equation}
where each $\mu^{i}\left(\alpha\right)$ is the probability mass with
respect to $\alpha\in\mathcal{A}^{i}$ and $i\in\left[N\right]$.\smallskip

The proposed combinatorial strategy utilizes the partition to reduce
dimensionality, where each base set represents a population of component
strategies corresponding to a distribution. Particularly, the representative
strategies $\big(\hat{b}_{n}^{i}\big)$, as per (\ref{eq:3.7}), represent
aggregate behaviors of all the component strategies in each corresponding
base set with respect to the mass $\mu^{i}\left(\alpha\right)$; then
the combinatorial strategy $\big(b_{n}\big)$ is constructed based
on these representative strategies according to (\ref{eq:3.6}), which
is similar to the formula (\ref{eq:3.3}). In addition, concerning
the allocation of wealth among the component strategies, let $P_{n}\big(b_{n}^{\alpha}\big)$
denote the time-dependent distribution of wealth over the corresponding
component strategies $\big(b_{n}^{\alpha}\big)$; then, the capital
allocated to each component portfolio $b_{n}^{\alpha}$ at time $n$
is given as follows:
\[
P_{n}\big(b_{n}^{\alpha}\big)={\displaystyle \dfrac{{\displaystyle \int_{\mathcal{B}^{k}}{\displaystyle \lambda_{i}}}S_{n-1}\big(b_{n-1}^{\lambda}\big)S_{n-1}\big(b_{n-1}^{\alpha}\big)\mu^{i}\big(\alpha\big)\mu\left(\lambda\right)\dif\lambda}{{\displaystyle \int_{\mathcal{B}^{k}}\sum_{\beta\in\mathcal{A}^{i}}S_{n-1}\big(b_{n-1}^{\lambda}\big){\displaystyle S_{n-1}\big(b_{n-1}^{\beta}\big)\mu^{i}\left(\beta\right)\mu\left(\lambda\right)}}\dif\lambda}},\,\forall\alpha\in\mathcal{A}^{i},i\in\left[N\right].
\]
Thus, the return at time $n$ of the combinatorial portfolio $\big(b_{n}\big)$
is ${\displaystyle \sum}_{\alpha\in\left[k\right]}\left\langle b_{n}^{\alpha},x_{n}\right\rangle P_{n}\big(\alpha\big)$.\smallskip

According to the construction of the proposed strategy, the benchmark
strategy $\big(b_{n}^{*}\big)$ is no longer formed by the constant
combinations of the component strategies but of the representative
strategies of the base sets, so the worst-case upper bound derived
in Proposition \ref{finite strategy} does not apply. Instead, Proposition
\ref{infinite strategy} derives the upper bound for the difference
$\log S_{n}\big(\ddot{b}_{n}\big)-\log S_{n}\big(b_{n}\big)$ in the
worst case of sequence $x_{1}^{n}$. However, contrary to Proposition
\ref{finite strategy}, this worst-case difference is not always lower
bounded by zero as in (\ref{eq:3.5}).
\begin{prop}
Consider a combinatorial strategy $\big(b_{n}\big)$ constructed according
to (\ref{eq:3.6}) with the uniform density $\mu\left(\lambda\right)$
over the simplex $\mathcal{\mathcal{B}}^{N}$, and the minimal mass
values ${\displaystyle \min_{i\in\left[N\right]}}\ensuremath{\mu^{i}\big(\alpha_{n}^{i}\big)}\eqqcolon\epsilon_{n}>0$\vspace{-1.3ex}
at $\alpha_{n}^{i}$, which are defined in hindsight as:
\[
\alpha_{n}^{i}=\operatorname*{argmax}_{\alpha\in\mathcal{A}^{i}}S_{n}\big(b_{n}^{\alpha}\big),\,\forall i\in\left[N\right],\,n\geq2,
\]
and the probability mass $\mu^{i}\left(\alpha\right)$ for the representative
portfolio $\hat{b}_{n}^{i}$ according to (\ref{eq:3.7}). We have:
\[
{\displaystyle \max_{x_{1}^{n}}\big(\log S_{n}\big(\ddot{b}_{n}\big)-\log S_{n}\big(b_{n}\big)\big)\leq\left(N-1\right)\log\left(n+1\right)-\log\epsilon_{n}},\text{ \ensuremath{\forall n}},
\]
where $\big(\ddot{b}_{n}\big)$ is the baseline strategy of all component
strategies in the set $\big\{\big(b_{n}^{\alpha}\big):\alpha\in\left[k\right]\big\}$.\label{infinite strategy}
\end{prop}
\begin{proof}
See Appendix.
\end{proof}
If the minimal mass values $\epsilon_{n}$ are chosen such that $n^{-1}\log\epsilon_{n}\to0$,
then by Proposition \ref{infinite strategy}, we have $\big(b_{n}\big)\succsim\big(\ddot{b}_{n}\big)$
for any sequence $x_{1}^{\infty}$, as the upper bound converges to
zero. The case of representative strategies with uniform probability
mass $\mu^{i}\left(\alpha\right)$ over each countable base set is
an instance ensuring such convergence. As we noted above, the infimum
limit is not bounded by zero, so there exist sequences $x_{1}^{\infty}$
where the combinatorial strategy $\big(b_{n}\big)$ is not always
exceeded by the baseline strategy $\big(\ddot{b}_{n}\big)$. The remark
below specifies conditions to guarantee the Combination criterion.
In addition, the time required for the combinatorial strategy to exceed
all component strategies increases significantly with larger numbers
of base sets and their respective scales. Additionally, the partitioning
into base sets affects the condition $\big(W_{n}\big(b_{n}\big)-W_{n}\big(\dot{b}_{t}\big)\big)\nrightarrow0$
and the evolutions of representative strategies' cumulative wealths.
Furthermore, the combinatorial strategy reduces to the one proposed
for a small number $k$ of component strategies by equating $N=k$.\smallskip
\begin{rem*}
Let $\big(\dot{b}_{n}\big)$ denote the best merged strategy of all
the representative strategies, as similarly defined as the baseline
strategy in Definition \ref{Baseline strategy}. The benchmark strategy
$\big(b_{n}^{*}\big)$ is defined as the constant combination of the
representative strategies, not the component strategies. If $n^{-1}\log\epsilon_{n}\to0$,
then $\big(\ddot{b}_{n}\big)\sim\big(\dot{b}_{n}\big)$ for any sequence
$x_{1}^{\infty}$; moreover, for sequences such that $\big(W_{n}\big(b_{n}^{*}\big)-W_{n}\big(\dot{b}_{n}\big)\big)\nrightarrow0$,
then $\big(b_{n}\big)\succ\big(\dot{b}_{n}\big)$ so $\big(b_{n}\big)\succ\big(\ddot{b}_{n}\big)$
which satisfies the Combination criterion, due to the transitivity
and asymmetry of the preference relation.
\end{rem*}

\section{Numerical experiments}

The integrals of the combinatorial portfolios corresponding to (\ref{eq:3.6})
or (\ref{eq:3.3}) and (\ref{eq:3.7}) can be numerically approximated
by Riemann sums that partition the simplices $\mathcal{B}^{N}$ and
$\mathcal{B}^{k}$ into finitely discrete grids. Essentially, a finer
partition can enhance the approximation, but it comes at the cost
of increasing computational complexity, especially with high-dimensional
simplices. Figure \ref{figure 2} provides an example from the experiment
section, showcasing the numerical approximation of the investor's
combinatorial strategy using single-step size discretization for the
1-dimensional simplex. It illustrates how more discretization points
and smaller step sizes can improve cumulative wealth. Additionally,
various step sizes and discretization points can be combined for further
refinement. It is important to note that while the theoretically optimal
constant combinations according to Definition \ref{def 4} may not
align with discretization points, fine partitions can ensure that
numerical maximizers closely approximate theoretical values. However,
enhancing numerical approximation is not the primary focus of this
experimental section.
\begin{figure}[H]
\begin{centering}
\includegraphics[scale=0.3]{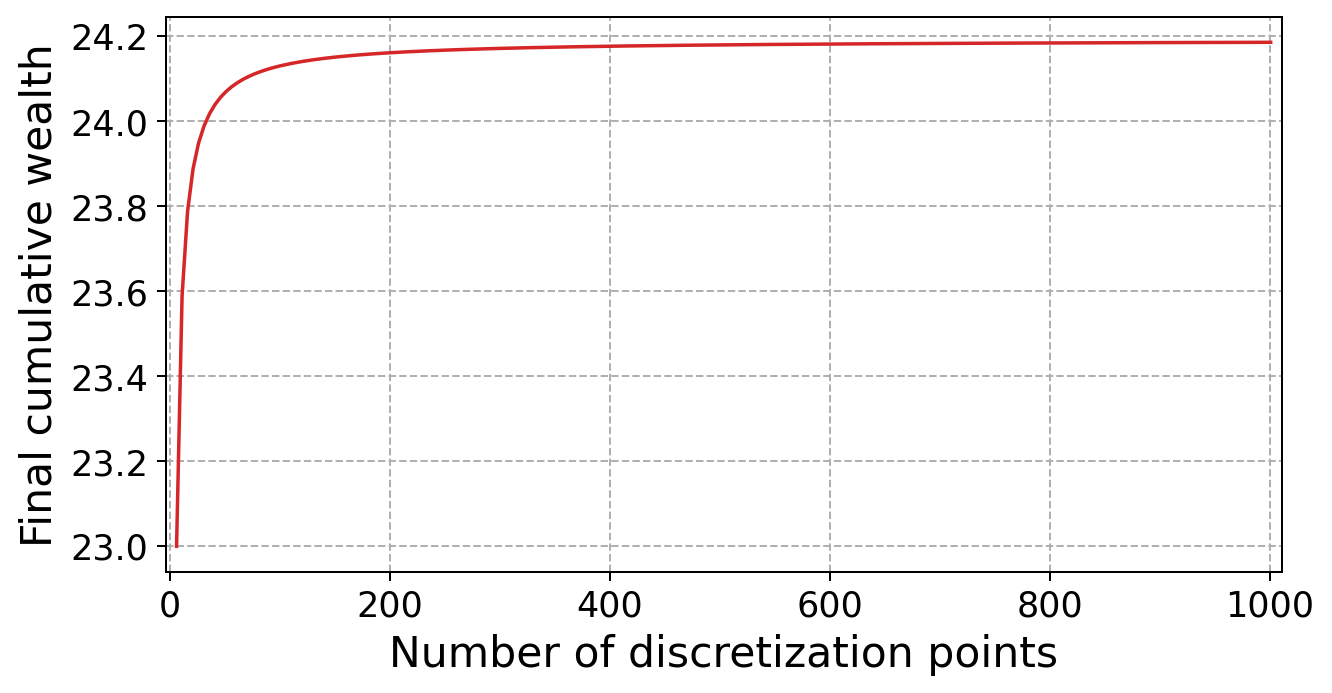}\includegraphics[scale=0.3]{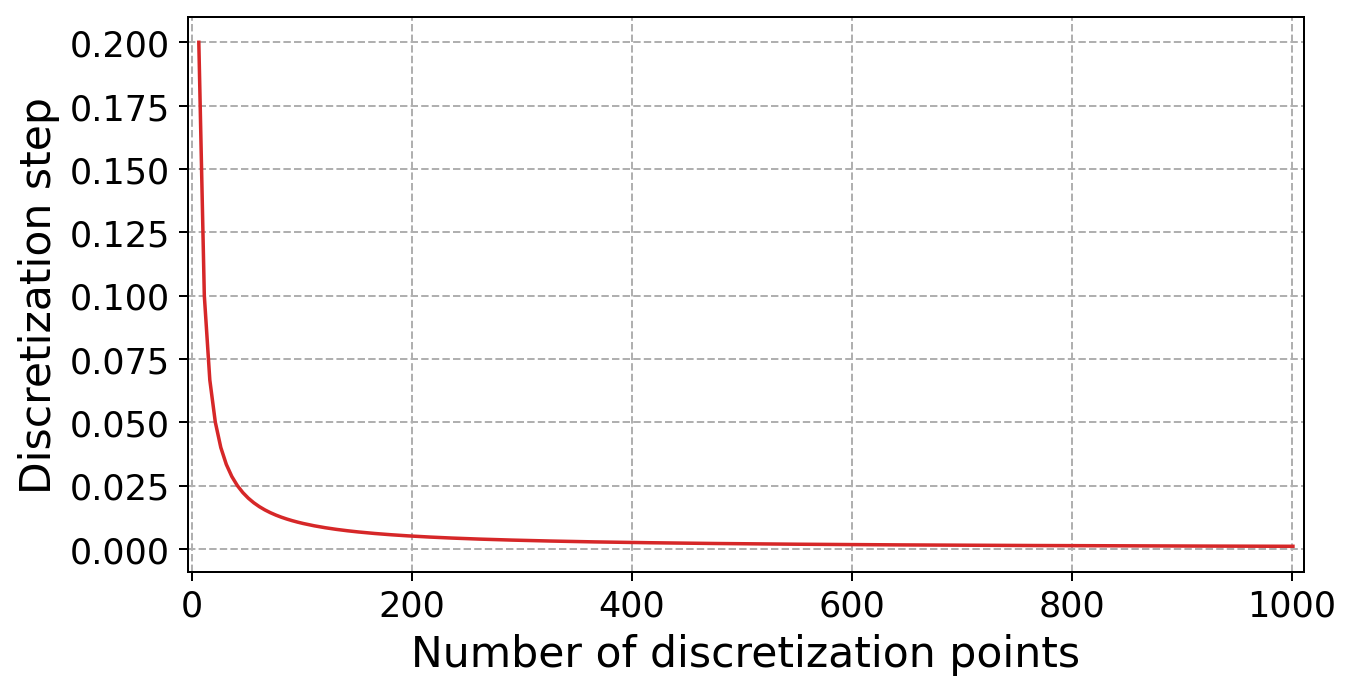}
\par\end{centering}
{\tiny\caption{Final cumulative wealth plotted against discretization points and
step sizes for approximating the combinatorial strategy with two component
strategies.\label{figure 2} }
}{\tiny\par}
\end{figure}

\subsection{Data and testing strategies}

We experiment with six diverse stocks over a 27-year span, from December
31, 1992, to December 31, 2019, covering $6798$ trading days. This
period encapsulates major market events like the dot-com bubble, the
2008 financial crisis, and pre-COVID-19 pandemic, as illustrated in
Figure \ref{figure 3}. The selected stocks are blue-chip symbols,
ranging across several sectors like industrials, energy, technology,
and consumer defensive, traded on the NYSE and NASDAQ exchanges. They
include Advanced Micro Devices, Inc. (AMD), The Boeing Company (BA),
Honeywell International Inc. (HON), The Coca-Cola Company (KO), 3M
Company (MMM), and Exxon Mobil Corporation (XOM). Their varied behaviors
during historic events challenge prediction, making them ideal for
testing our strategies. We used daily adjusted closing prices in USD
provided by The Center for Research in Security Prices, LLC (CRSP).
\begin{figure}[H]
\begin{centering}
\includegraphics[scale=0.33]{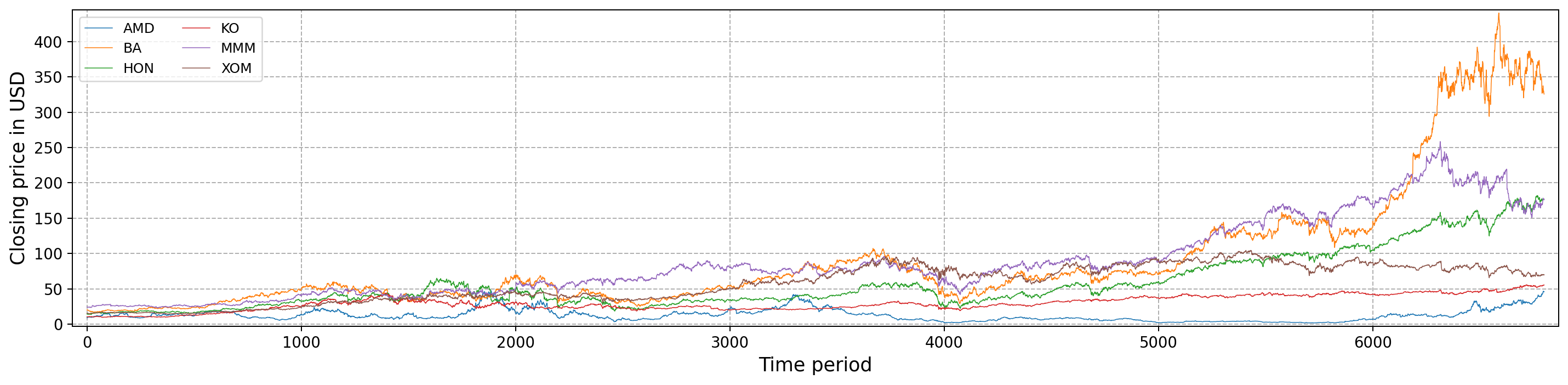}
\par\end{centering}
\caption{Adjusted closing prices in USD of the stocks from December 31, 1992,
to December 31, 2019.\label{figure 3}}
\end{figure}

\textbf{Market scenario}. Let's consider a market where many fund
managers have different risk profiles for their investment strategies,
and an investor can observe them as the components of a combinatorial
strategy. We consider the sequences of Mean-Variance portfolios corresponding
to the managers' risk aversion levels $\alpha\in\mathbb{R}_{+}$,
which are defined as follows, as component strategies:
\[
b_{n}^{\alpha}\coloneqq\operatorname*{argmax}_{b\in\mathcal{B}^{m}}\big(\alpha\left\langle b,\hat{\mu}_{n}\right\rangle -\big<b,\hat{\varSigma}_{n}b\big>\big),\forall n,
\]
where $\hat{\mu}_{n}$ and $\hat{\varSigma}_{n}$ denote the estimators
for the expected values and covariance matrix of the assets' returns
at time $n$. Specifically, we further assume that $\hat{\mu}_{n}$
and ${\displaystyle \hat{\varSigma}_{n}}$ are estimated using a mean-reversion
approach that utilizes past prices of six assets within a moving window
as:
\begin{align*}
\hat{\mu}_{n} & \ensuremath{\coloneqq\frac{{\displaystyle 1}}{J}\stackrel[i=n-J]{n-1}{\sum}x_{i}\text{ \,\,and\,\, }\hat{\varSigma}_{n}\left(i,j\right)\coloneqq{\displaystyle \frac{1}{J-1}}\stackrel[h=n-J]{n-1}{\sum}\big(x_{h,i}-\hat{\mu}_{n,i}\big)\big(x_{h,j}-\hat{\mu}_{n,j}\big)},\text{ \ensuremath{\forall i,j}\ensuremath{\in\left[6\right],\forall n>J,}}
\end{align*}
where $J$ is the scrolling window length set to four weeks $\left(J=20\right)$
for the entire investment simulation, starting from January 29, 1993,
which is considered as the origin day of time $n=1$. The M-V portfolios
are estimated numerically using the OSQP solver by \citet{Stellato2020},
implemented within the CVXPY module for Python by \citet{Diamond2016}
and \citet{Agrawal2018}.\smallskip

\textbf{Testing strategies}. We propose further an accelerated variant
of the combinatorial strategy and assess their performances against
various other (online learning) multi-strategies. The involved strategies
and their corresponding abbreviations are summarized as follows:
\begin{itemize}
\item The proposed combinatorial strategy according to (\ref{eq:3.6}) or
(\ref{eq:3.3}), and (\ref{eq:3.7}) with uniform $\mu\left(\cdot\right)$
and $\mu^{i}\left(\cdot\right)$. We name it as \textit{Universally
Combinatorial }strategy (UC) with associated notation $\big(b_{n}^{\text{UC}}\big)$,
inspired by the Universal Portfolio proposed by \citet{Cover1991}.
\item The aforementioned Mean-Variance strategies (M-V), denoted as $\big(b_{n}^{\alpha}\big)$
with various coefficients $\alpha\in\mathbb{R}_{+}$, which serve
as the component strategies for all the tested multi-strategies.
\item The \textit{Weighted Average Exponential }strategy (WAE). In detail,
the WAE portfolios, denoted by $\big(b_{n}^{\text{WAE}}\big)$, combine
$k$ component portfolios in the set $\big\{\big(b_{n}^{\alpha}\big):\alpha\in\left[k\right]\big\}$,
as follows:
\begin{equation}
b_{1}^{\text{WAE}}=\dfrac{1}{k}\sum_{\alpha\in\left[k\right]}b_{1}^{\alpha}\text{ and }b_{n}^{\text{WAE}}={\displaystyle {\displaystyle \dfrac{{\displaystyle \mathop{\sum_{\alpha\in\left[k\right]}b_{n}^{\alpha}}S_{n-1}\big(b_{n-1}^{\alpha}\big)}}{{\displaystyle \mathop{\sum_{\alpha\in\left[k\right]}}S_{n-1}\big(b_{n-1}^{\alpha}\big)}},\text{ \ensuremath{\forall n\geq2}}.}}\label{eq:WAE}
\end{equation}
\item The \textit{Following-the-Leader }strategy (FL), denoted by $b_{n}^{\text{FL}}$,
uses the best component strategy, among $k$ component strategies
in $\big\{\big(b_{n}^{\alpha}\big):\alpha\in\left[k\right]\big\}$,
in hindsight after $n-1$ periods for the next period $n$. Specifically,
it is defined as follows:
\[
b_{1}^{\text{FL}}=\dfrac{1}{k}\sum_{\alpha\in\left[k\right]}b_{1}^{\alpha}\text{ \,and\, }b_{n}^{\text{FL}}=b_{n-1}^{\alpha_{n-1}},\text{ where }\text{\ensuremath{\alpha_{n-1}}}\coloneqq\operatorname*{argmax}_{\alpha\in\left[k\right]}S_{n-1}\big(b_{n-1}^{\alpha}\big),\forall n\geq2.
\]
\item The accelerated variant of the combinatorial strategy uses a time-varying
density in place of the time-invariant uniform one over the simplex
$\mathcal{B}^{k}$. Namely, we introduce the \textit{Universal combination
on the Winners} (UC-W) as a UC strategy, where the uniform density
$\mu_{n}$ has support on $B_{n}\subset\mathcal{B}^{k}$ at period
$n\geq2$, defined as follows:
\begin{equation}
B_{n}\coloneqq\left\{ \lambda\in\mathcal{B}^{k}:\,S_{n-1}\big(b_{n-1}^{\lambda}\big)>S_{n-1}\big(b_{n-1}^{\bar{\lambda}}\big),\,\forall\bar{\lambda}\in\mathcal{B}^{k}\backslash B_{n}\right\} ,\forall n\geq2.\label{UC-W}
\end{equation}
Accordingly, the \textit{Universal combination on the Losers} (UC-L)
is formed using the subsets:
\begin{equation}
\bar{B}_{n}\coloneqq\left\{ \lambda\in\mathcal{B}^{k}:\,S_{n-1}\big(b_{n-1}^{\lambda}\big)<S_{n-1}\big(b_{n-1}^{\bar{\lambda}}\big),\,\forall\bar{\lambda}\in\mathcal{B}^{k}\backslash\bar{B}_{n}\right\} ,\forall n\geq2.\label{UC-L}
\end{equation}
This method is applicable to the simplex $\mathcal{B}^{N}$ for the
strategy constructed according to (\ref{eq:3.6}). The UC-W and UC-L
strategies are also computed using Riemann approximation.\smallskip
\end{itemize}
The experiments illustrate the evolution of cumulative wealths of
the testing strategies through several graphics. Specifically, we
investigate whether they exhibit a strict preference over the baseline
strategies by eventually exceeding all the M-V component strategies.
Additionally, we report common performance measures for the testing
strategies, including their final cumulative wealth, growth rate,
average return (i.e., empirical expected return), and empirical Sharpe
ratio.

\subsection{Experiments with small scales of component strategies}

In the first four experiments, we implement different small numbers
of component M-V strategies to examine the impact of scales on the
behaviors of the multi-strategies, including $\left\{ 0.005,1\right\} $,
$\left\{ 0.002,0.02,100\right\} $, $\left\{ 0.05,0.2,0.5,0.8,1\right\} $
and $\left\{ 0.002,0.05\right\} $ for risk aversions in Experiments
1, 2, 3, and 4, respectively. Generally, Figure \ref{Figure 4} illustrates
the UC strategies eventually exceed all the component M-V strategies,
except in Experiment 4. The strategies WAE and FL are always exceeded
by the best M-V strategies in all cases. The results indicate that
the FL strategies demonstrate relative effectiveness if there exists
a single M-V strategy that almost always exceeds the others after
$N$ periods, then $\big(\ddot{b}_{n}\big)\sim\big(b_{n}^{\text{FL}}\big)$
due to the following:
\[
\lim_{n\to\infty}\big(W_{n}\big(\ddot{b}_{n}\big)-W_{n}\big(b_{n}^{\text{FL}}\big)\big)=\lim_{n\to\infty}\frac{1}{n}\big(\log S_{N+1}\big(\ddot{b}_{N+1}\big)-\log S_{N+1}\big(b_{N+1}^{\text{FL}}\big)\big)=0.\vspace{0.25ex}
\]
However, the mechanisms of the FL and WAE strategies inherently constrain
their growth rates, being bounded by the baseline strategy. Consequently,
despite the possibility of an indifferent preference with the latter,
the cumulative wealth evolutions of the two former strategies are
worse. These experiments emphasize that attaining the same asymptotic
growth rate as the best component strategy does not guarantee a multi-strategy
not being the worst performer.\smallskip

Table \ref{Table 1} presents performance metrics for the strategies
involved, while Figure \ref{figure 5} displays the final cumulative
wealth against the numerical constant combination. The Sharpe ratios
of the multi-strategies do not surpass the best among the component
M-V strategies. However, the UC strategies offer the most benefit
as they achieve the highest final cumulative wealths in the first
three experiments, with relatively small tradeoffs in terms of the
Sharpe ratios. In contrast, all M-V strategies that yield higher final
wealth and mean of return experience significant decreases in Sharpe
ratios. These favorable performances of the UC strategies are attained
through approximations with respect to the numerical combinations
$\lambda$. Although their performances could potentially improve
with a finer grid, Figure \ref{figure 5} demonstrates that such marginal
improvement is not worthwhile, as the computation may become infeasible
for a large scale of component strategies.
\begin{figure}[H]
\begin{centering}
\includegraphics[scale=0.33]{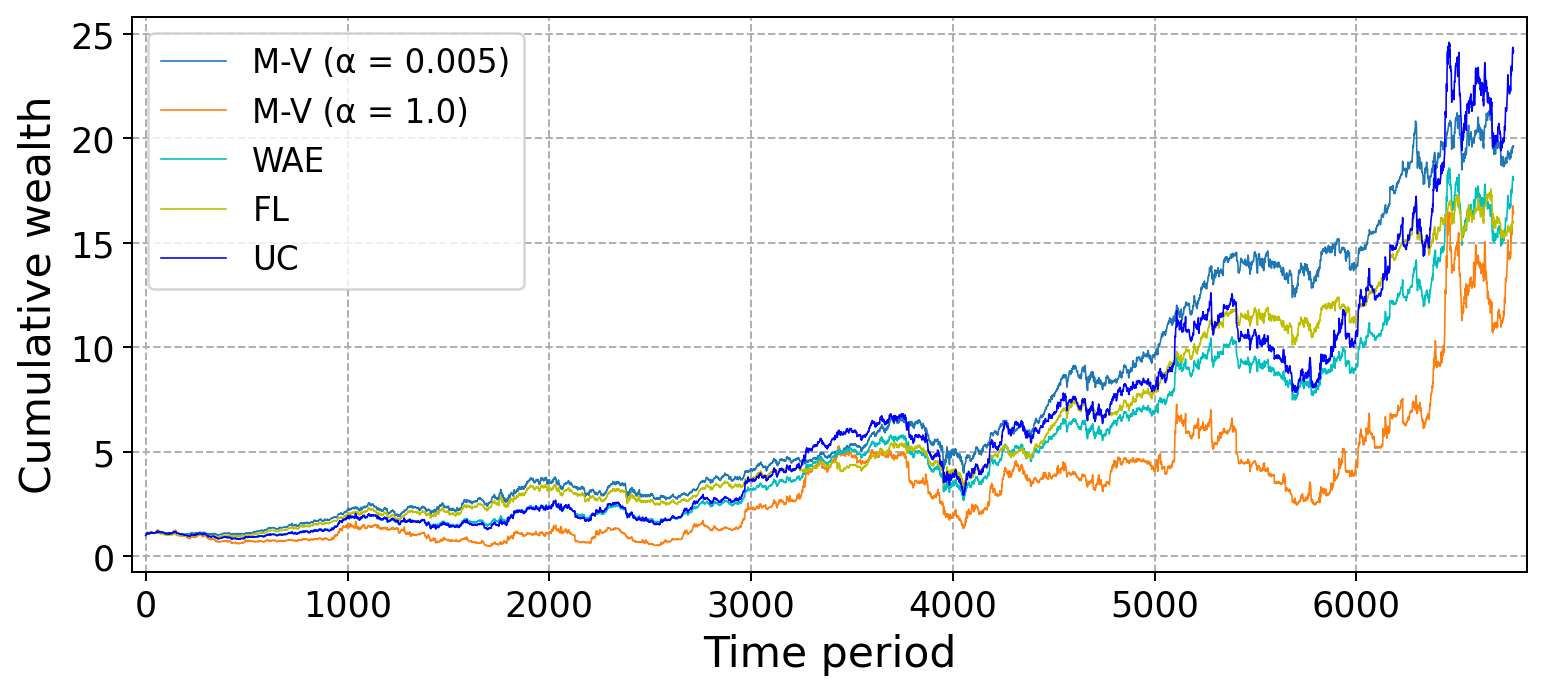} ~\includegraphics[scale=0.33]{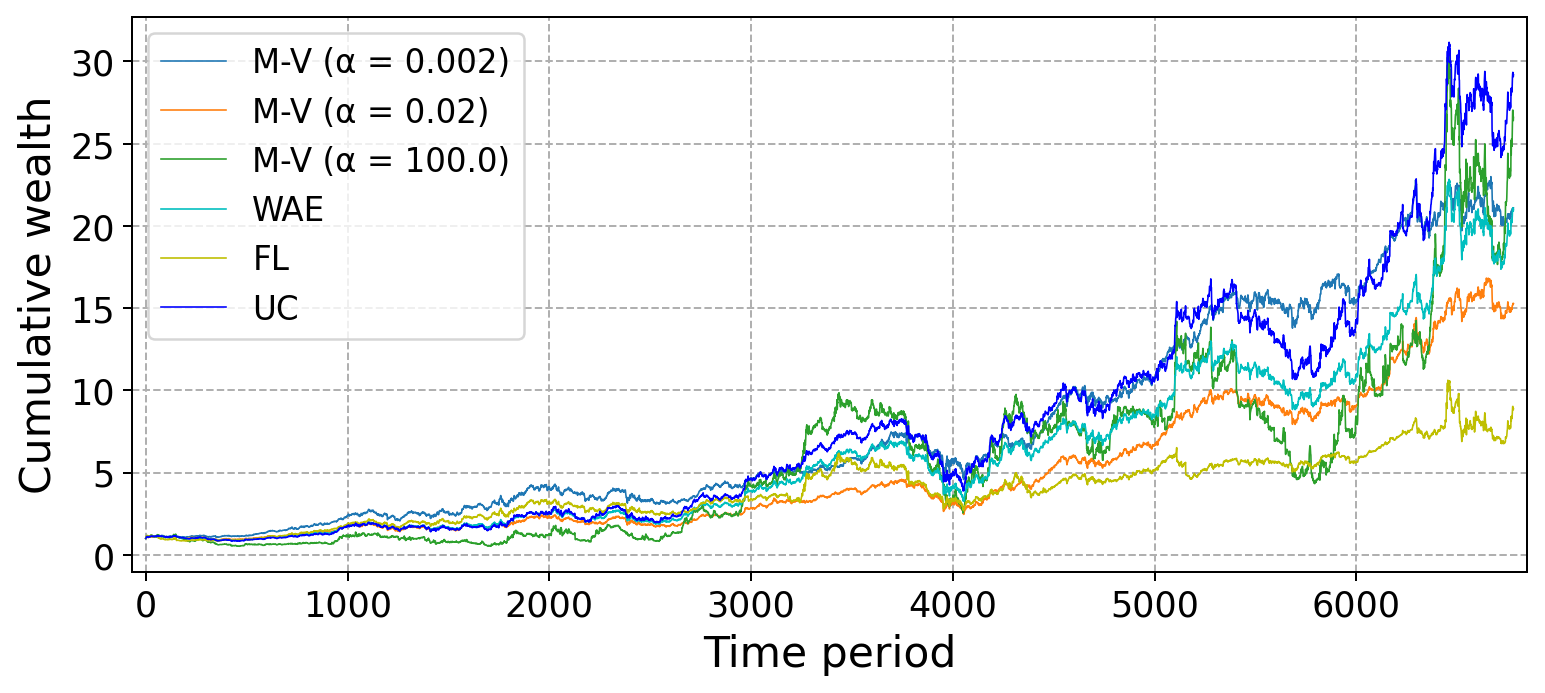}
\par\end{centering}
\begin{centering}
\includegraphics[scale=0.33]{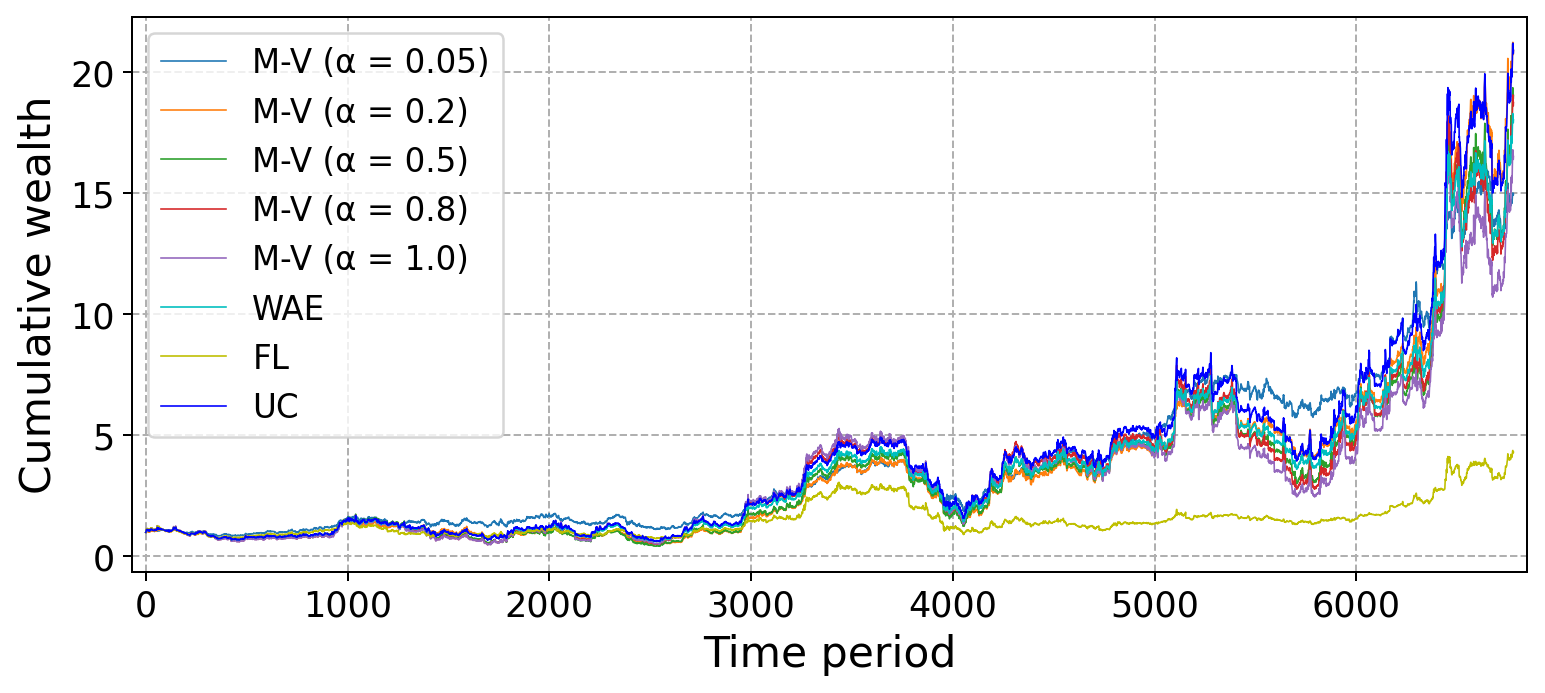} ~\includegraphics[scale=0.33]{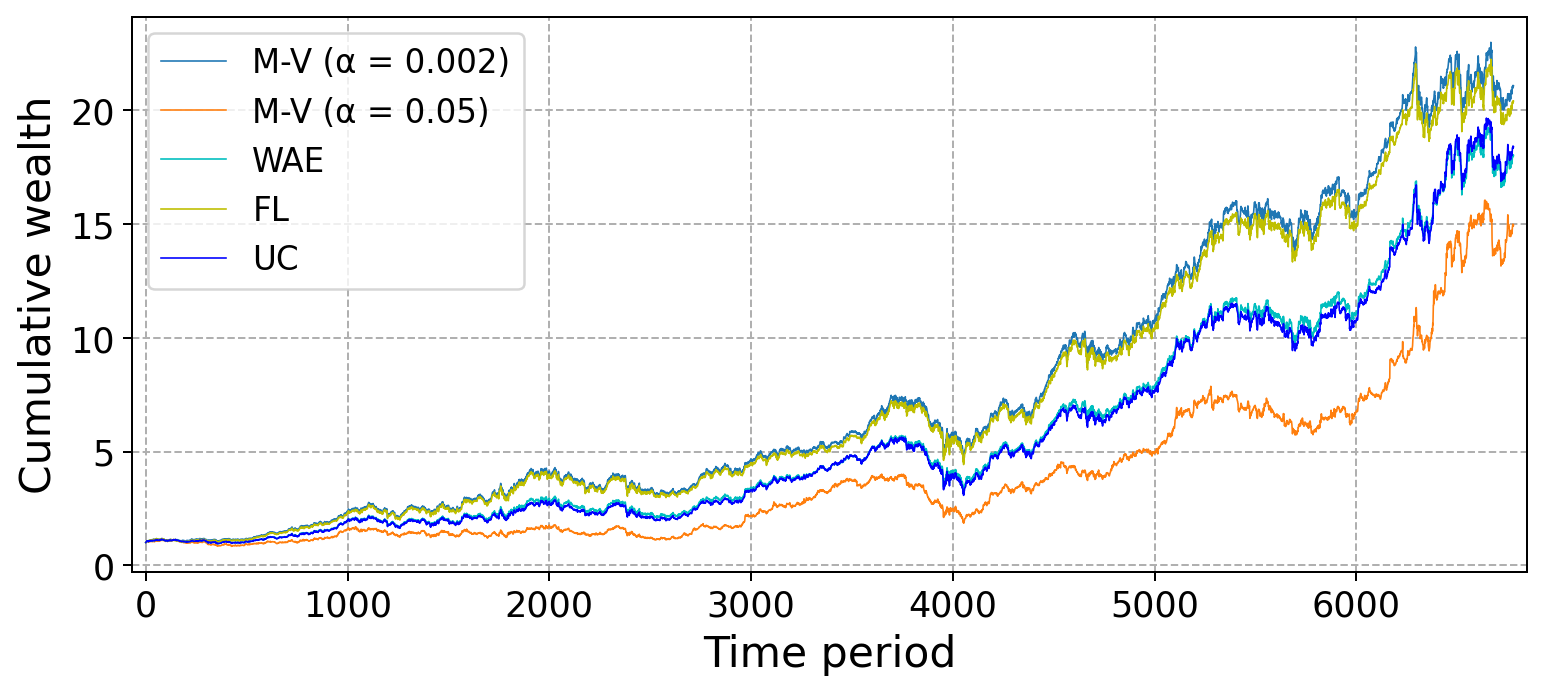}
\par\end{centering}
\caption{Evolutions of the multi-strategies over time in Experiments 1, 2,
3, and 4 (from left-to-right and top-to-bottom order).\label{Figure 4}}

\end{figure}
\renewcommand{\arraystretch}{0.3}
\begin{table}[H]
\caption{Performance measures for the strategies in Experiments 1, 2, 3 and
4.\label{Table 1} }

\begin{centering}
\begin{tabular*}{12cm}{@{\extracolsep{\fill}}>{\raggedright}m{2.5cm}>{\raggedright}m{2cm}>{\raggedright}m{2cm}>{\raggedright}m{2cm}>{\raggedright}m{2cm}}
\toprule 
{\tiny\textbf{Strategy}} & {\tiny\textbf{Final wealth}} & {\tiny\textbf{Average growth rate}} & {\tiny\textbf{Average return}} & {\tiny\textbf{Sharpe ratio}}\tabularnewline
\midrule
\midrule 
{\tiny M-V ($\alpha=0.002$)} & {\tiny 21.081643} & {\tiny 0.000450} & {\tiny 1.000511} & {\tiny 90.884103}\tabularnewline
{\tiny M-V ($\alpha=0.005$)} & {\tiny 19.626578} & {\tiny 0.000439} & {\tiny 1.000500} & {\tiny 90.777017}\tabularnewline
{\tiny M-V ($\alpha=0.02$)} & {\tiny 15.270352} & {\tiny 0.000402} & {\tiny 1.000468} & {\tiny 87.209390}\tabularnewline
{\tiny M-V ($\alpha=0.05$)} & {\tiny 14.951959} & {\tiny 0.000399} & {\tiny 1.000482} & {\tiny 77.889510}\tabularnewline
{\tiny M-V ($\alpha=0.2$)} & {\tiny 20.864838} & {\tiny 0.000448} & {\tiny 1.000615} & {\tiny 54.677952}\tabularnewline
{\tiny M-V ($\alpha=0.5$)} & {\tiny 19.022646} & {\tiny 0.000435} & {\tiny 1.000700} & {\tiny 43.451483}\tabularnewline
{\tiny M-V ($\alpha=0.8$)} & {\tiny 18.728020} & {\tiny 0.000432} & {\tiny 1.000737} & {\tiny 40.520186}\tabularnewline
{\tiny M-V ($\alpha=1$)} & {\tiny 16.503138} & {\tiny 0.000414} & {\tiny 1.000734} & {\tiny 39.543545}\tabularnewline
{\tiny M-V ($\alpha=100$)} & {\tiny 26.641673} & {\tiny 0.000484} & {\tiny 1.000849} & {\tiny 37.017071}\tabularnewline
\end{tabular*}
\par\end{centering}
\begin{centering}
\begin{tabular*}{12cm}{@{\extracolsep{\fill}}>{\raggedright}p{2.5cm}>{\raggedright}p{2cm}>{\raggedright}p{2cm}>{\raggedright}p{2cm}>{\raggedright}p{2cm}}
\toprule 
\multicolumn{5}{c}{{\tiny Experiment 1 with risk aversion $\alpha\in\left\{ 0.005,1\right\} $}}\tabularnewline
\midrule
{\tiny UC } & {\tiny 24.178108} & {\tiny 0.000470} & {\tiny 1.000580} & {\tiny 67.419175}\tabularnewline
{\tiny WAE} & {\tiny 18.060145} & {\tiny 0.000427} & {\tiny 1.000513} & {\tiny 76.449559}\tabularnewline
{\tiny FL} & {\tiny 16.015280} & {\tiny 0.000409} & {\tiny 1.000474} & {\tiny 88.246262}\tabularnewline
\end{tabular*}
\par\end{centering}
\begin{centering}
\begin{tabular*}{12cm}{@{\extracolsep{\fill}}>{\raggedright}p{2.5cm}>{\raggedright}p{2cm}>{\raggedright}p{2cm}>{\raggedright}p{2cm}>{\raggedright}p{2cm}}
\toprule 
\multicolumn{5}{c}{{\tiny Experiment 2 with risk aversion $\alpha\in\left\{ 0.002,0.02,100\right\} $}}\tabularnewline
\midrule
{\tiny UC} & {\tiny 29.192687} & {\tiny 0.000498} & {\tiny 1.000597} & {\tiny 71.217927}\tabularnewline
{\tiny WAE} & {\tiny 20.982959} & {\tiny 0.000449} & {\tiny 1.000535} & {\tiny 76.247127}\tabularnewline
{\tiny FL} & {\tiny 8.8740965} & {\tiny 0.000322} & {\tiny 1.000416} & {\tiny 73.280794}\tabularnewline
\end{tabular*}
\par\end{centering}
\begin{centering}
\begin{tabular*}{12cm}{@{\extracolsep{\fill}}>{\raggedright}p{2.5cm}>{\raggedright}p{2cm}>{\raggedright}p{2cm}>{\raggedright}p{2cm}>{\raggedright}p{2cm}}
\toprule 
\multicolumn{5}{c}{{\tiny Experiment 3 with risk aversion $\alpha\in\left\{ 0.05,0.2,0.5,0.8,1\right\} $}}\tabularnewline
\midrule
{\tiny UC} & {\tiny 20.887842} & {\tiny 0.000448} & {\tiny 1.000645} & {\tiny 50.409559}\tabularnewline
{\tiny WAE} & {\tiny 18.014120} & {\tiny 0.000427} & {\tiny 1.000614} & {\tiny 51.686007}\tabularnewline
{\tiny FL} & {\tiny 4.3021720} & {\tiny 0.000215} & {\tiny 1.000340} & {\tiny 63.518664}\tabularnewline
\end{tabular*}
\par\end{centering}
\begin{centering}
\begin{tabular*}{12cm}{@{\extracolsep{\fill}}>{\raggedright}p{2.5cm}>{\raggedright}p{2cm}>{\raggedright}p{2cm}>{\raggedright}p{2cm}>{\raggedright}p{2cm}}
\toprule 
\multicolumn{5}{c}{{\tiny Experiment 4 with risk aversion $\alpha\in\left\{ 0.002,0.05\right\} $}}\tabularnewline
\midrule
{\tiny UC } & {\tiny 18.401030} & {\tiny 0.000430} & {\tiny 1.000493} & {\tiny 88.740964}\tabularnewline
{\tiny WAE} & {\tiny 18.016801} & {\tiny 0.000427} & {\tiny 1.000489} & {\tiny 89.929546}\tabularnewline
{\tiny FL} & {\tiny 20.401686} & {\tiny 0.000445} & {\tiny 1.000506} & {\tiny 90.848864}\tabularnewline
\bottomrule
\end{tabular*}
\par\end{centering}
\centering{}%
\begin{minipage}[t]{12cm}%
\begin{spacing}{0.5}
{\tiny\textbf{Note}}{\tiny . The UC strategies are approximated by
2001 discretization points by step 0.0005 in Experiments 1 and 4,
4947 discretization points by step 0.01 in Experiment 2, and 9821
discretization points by step 0.05 in Experiment 3.}
\end{spacing}
\end{minipage}
\end{table}
\vspace{-2ex}

The first row of Figure \ref{Figure 6} illustrates the growth rate
differences between the numerically computed benchmark strategy $\big(b_{n}^{*}\big)$
and the UC strategy $\big(b_{n}^{\text{UC}}\big)$, as well as between
the benchmark strategy and the baseline strategy $\big(\ddot{b}_{n}\big)$
of all related M-V strategies. The Combination criterion is satisfied
in the first three experiments but not in the remainder, as $\big(W_{n}\big(\ddot{b}_{n}\big)-W_{n}\big(b_{n}^{*}\big)\big)\to0$
as $n$ increases. In Experiment 1, the M-V strategy $\big(b_{n}^{0.005}\big)$
almost always exceeds the M-V strategy $\big(b_{n}^{1}\big)$, but
the UC strategy is able to exceed the former. Conversely, in Experiment
4, the UC strategy never exceeds the best M-V strategy $\big(b_{n}^{0.002}\big)$,
which almost always exceeds the M-V strategy $\big(b_{n}^{0.05}\big)$.
In the two remaining experiments, the UC strategy is capable of exceeding
the best M-V strategies, where one component strategy experiences
strong fluctuations while the other two experience stable growth over
time in Experiment 2, and five comparable component strategies continuously
compete for the leading positions over time in Experiment 4.
\begin{figure}[H]
\begin{centering}
\begin{tabular}{cc||c||c}
\includegraphics[scale=0.29]{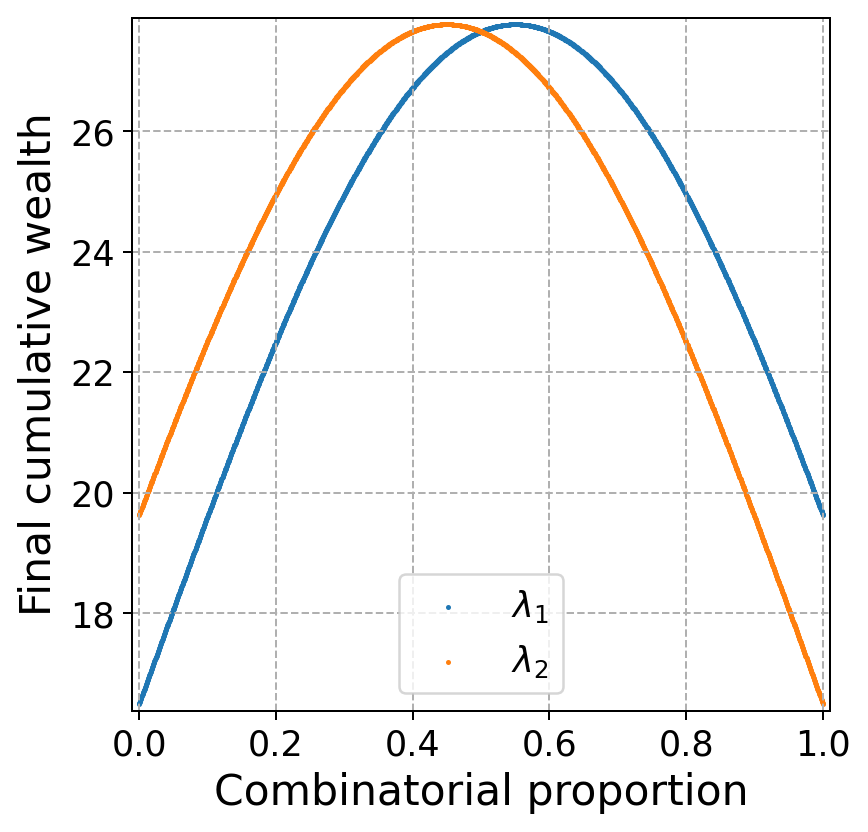} & \multicolumn{3}{c}{\includegraphics[scale=0.29]{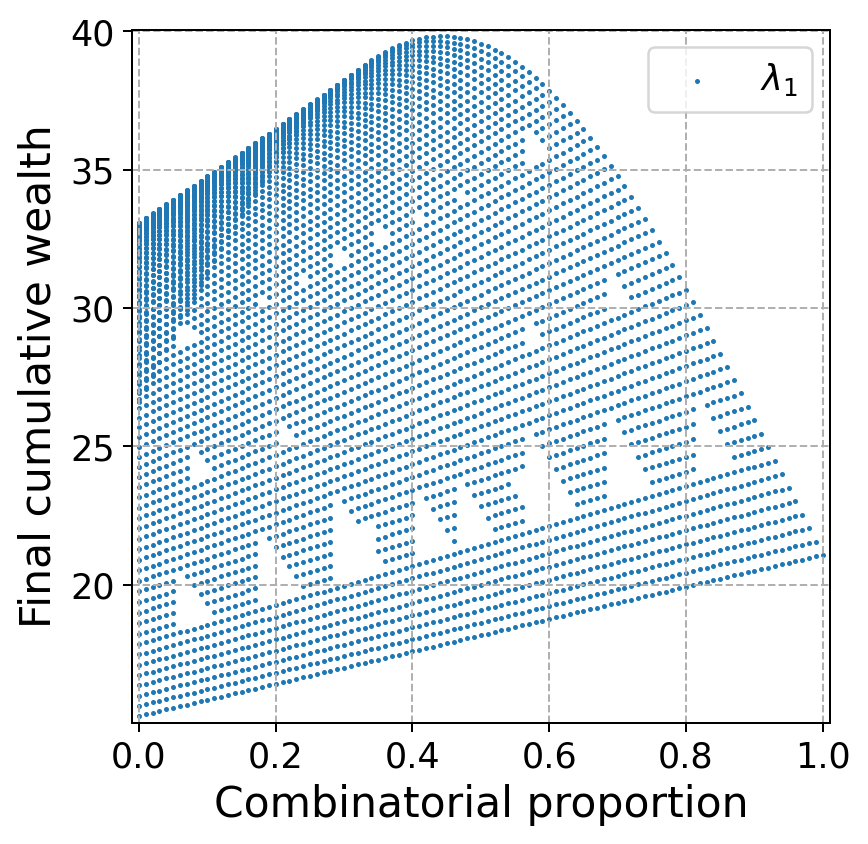}\includegraphics[scale=0.29]{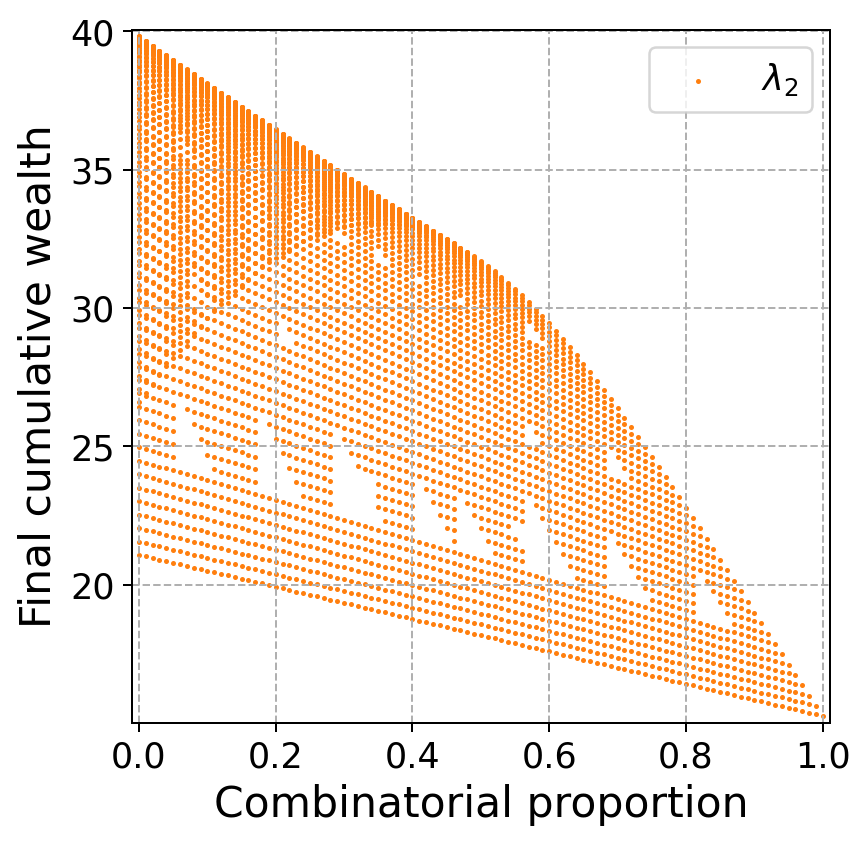}\includegraphics[scale=0.29]{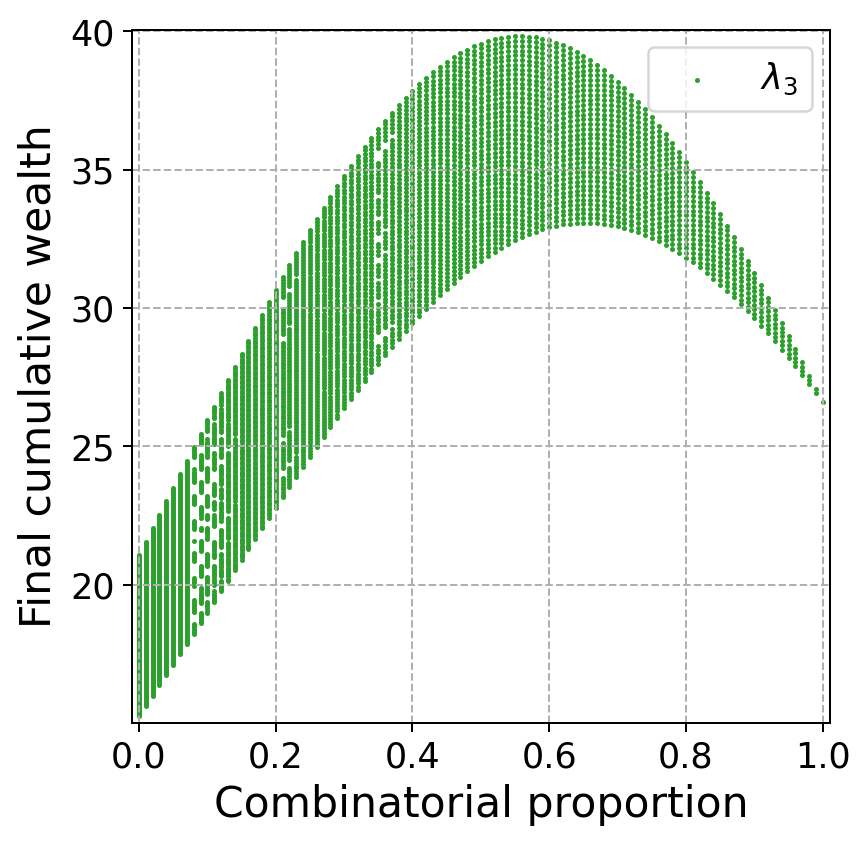}}\tabularnewline
{\footnotesize ~~Experiment 1} & \multicolumn{3}{c}{{\footnotesize ~~~Experiment 2}}\tabularnewline
\includegraphics[scale=0.29]{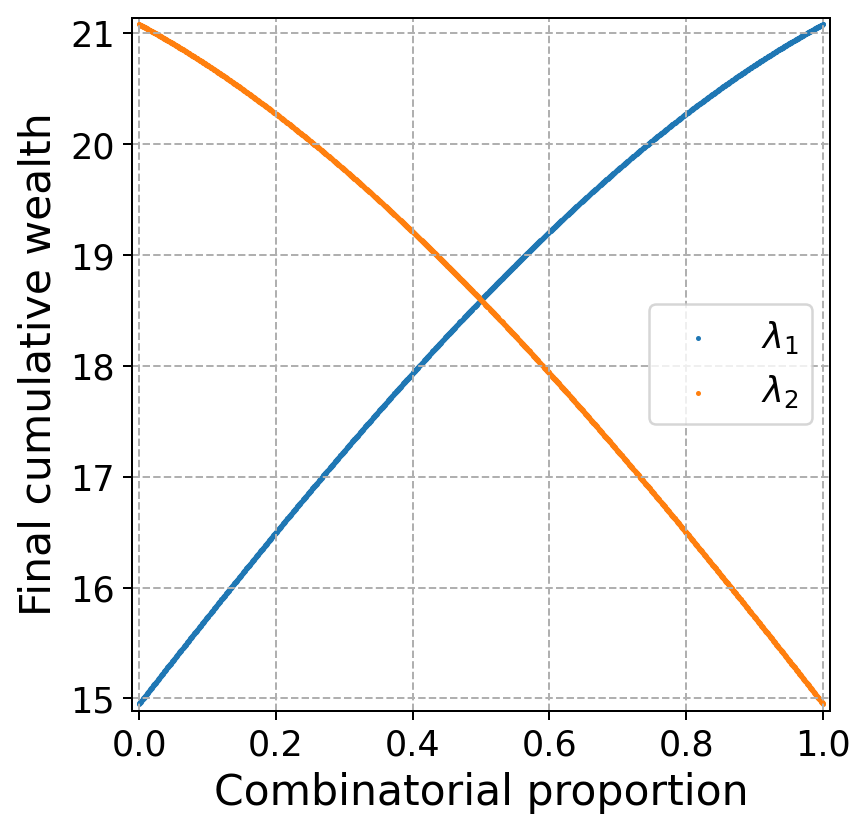} & \multicolumn{3}{c}{\includegraphics[scale=0.29]{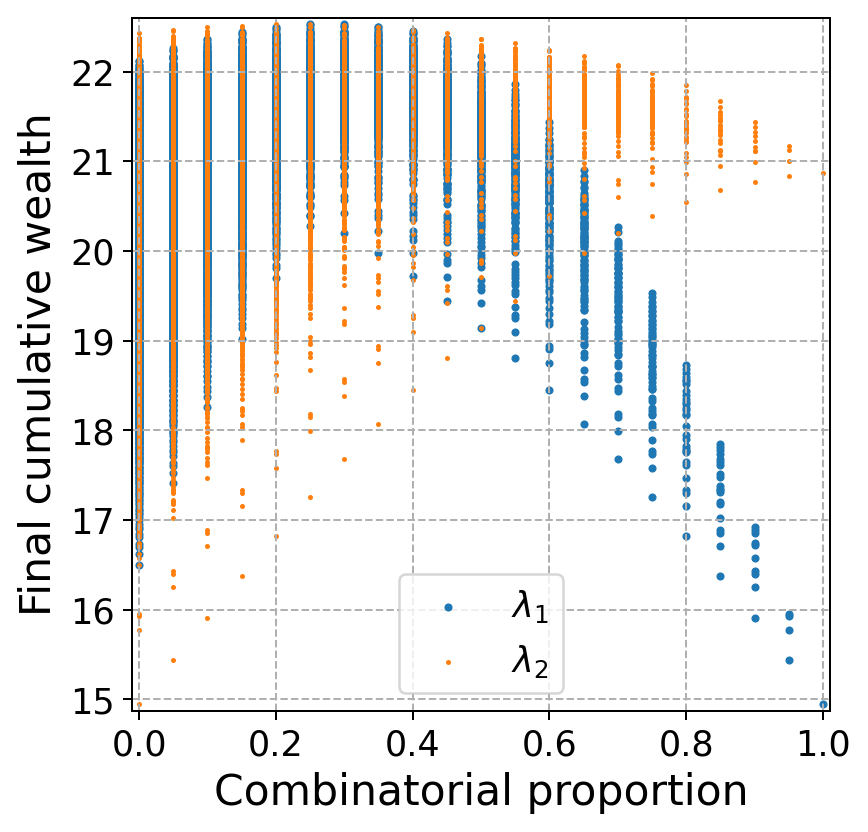}\includegraphics[scale=0.29]{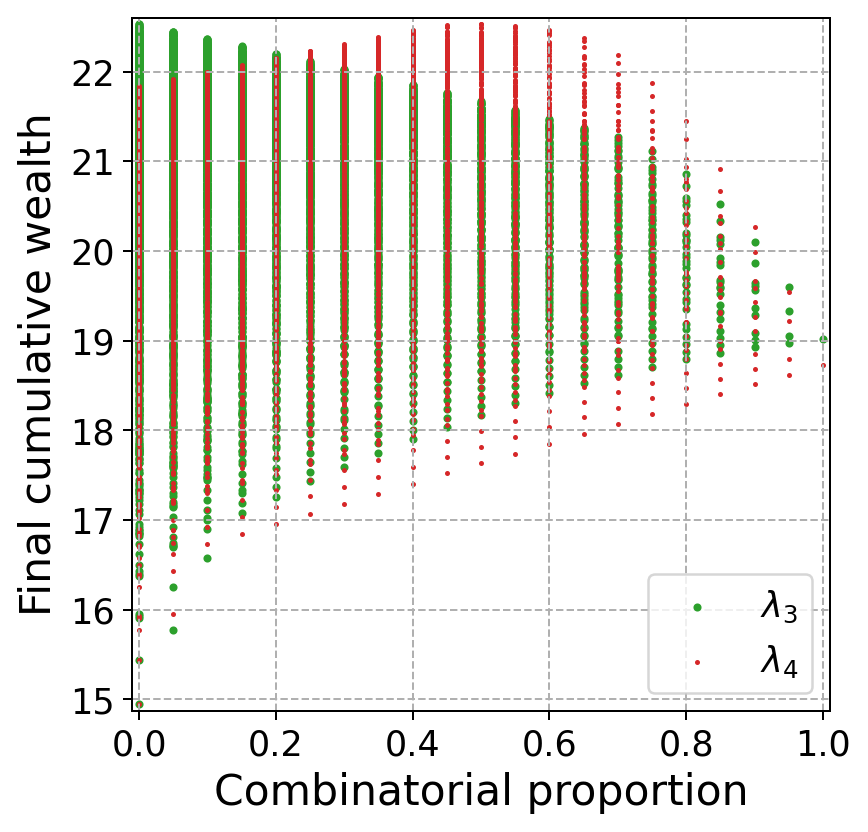}\includegraphics[scale=0.29]{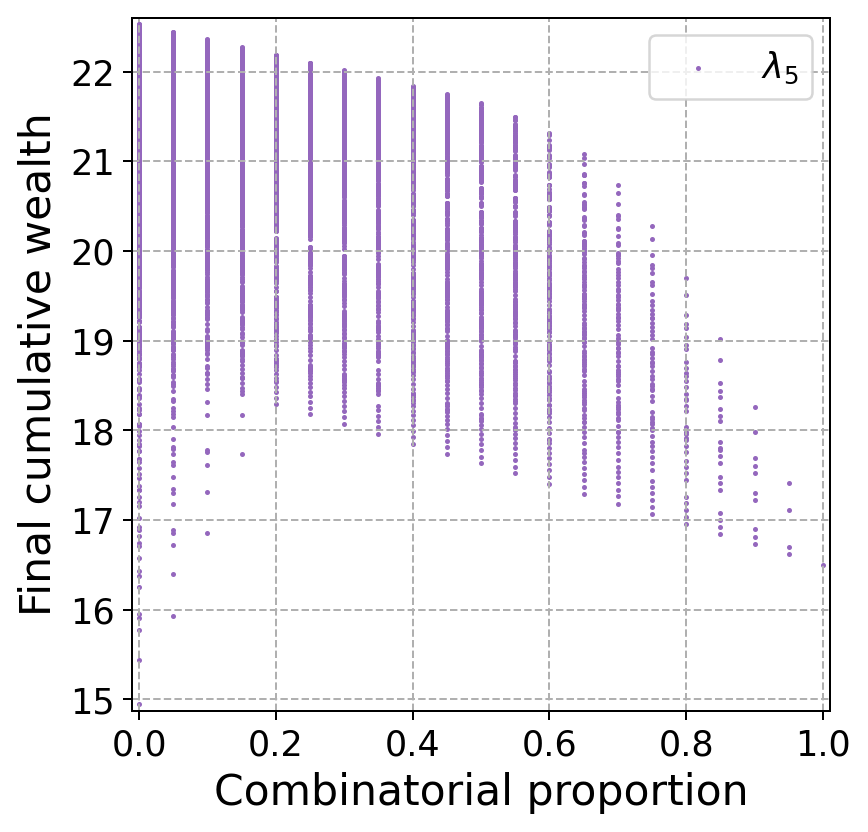}}\tabularnewline
{\footnotesize ~~Experiment 4} & \multicolumn{3}{c}{{\footnotesize ~~~Experiment 3}}\tabularnewline
\end{tabular}
\par\end{centering}
\caption{The constant combination $\lambda$ plotted against the final cumulative
wealth with respect to the discretization. The best constant combinations
in Experiment 1, 2, 3, and 4 are $\left(0.55,0.45\right)$, $\left(0.44,0,0.56\right)$,
$\left(0.3,0.2,0,0.5,0\right)$ and $\left(1,0\right)$, respectively.\label{figure 5}}
\end{figure}
\begin{figure}[H]
\begin{centering}
\begin{tabular}{cccc}
\includegraphics[scale=0.28]{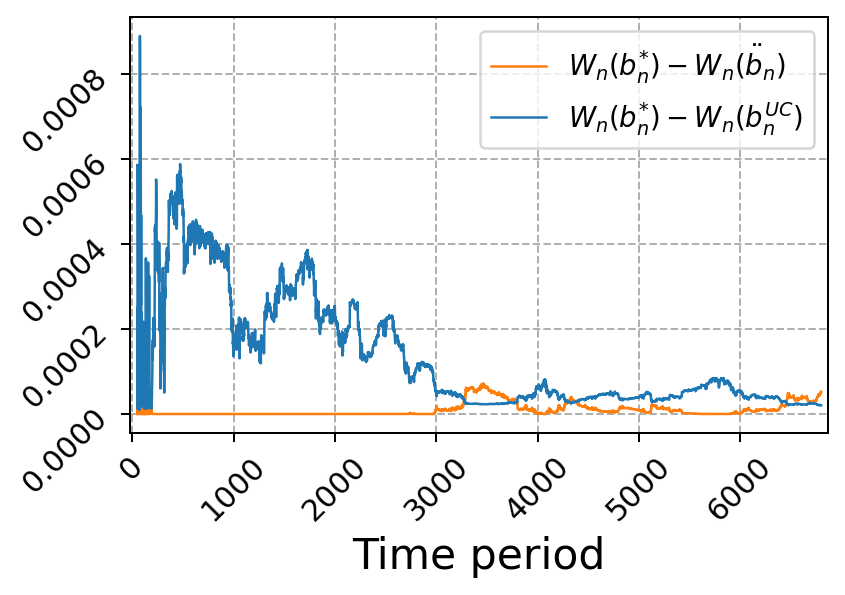} & \includegraphics[scale=0.28]{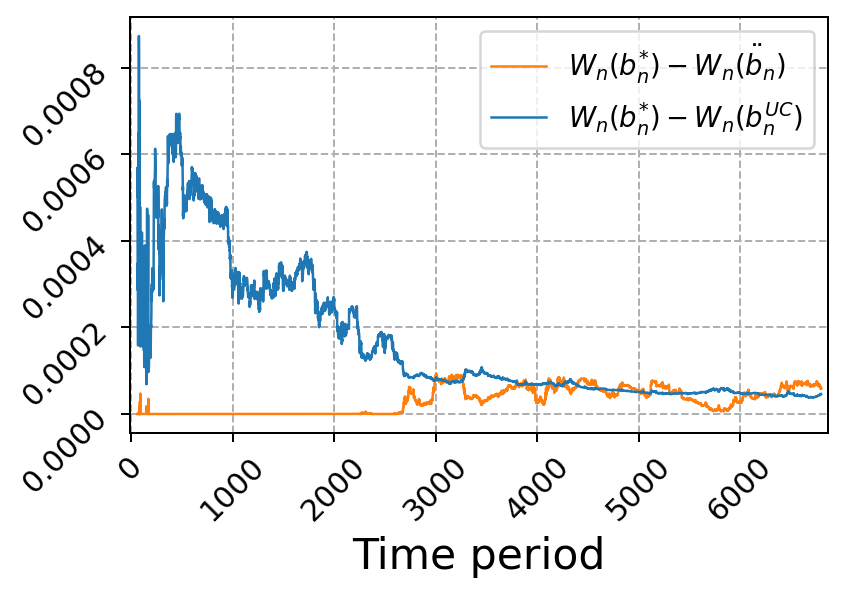} & \includegraphics[scale=0.28]{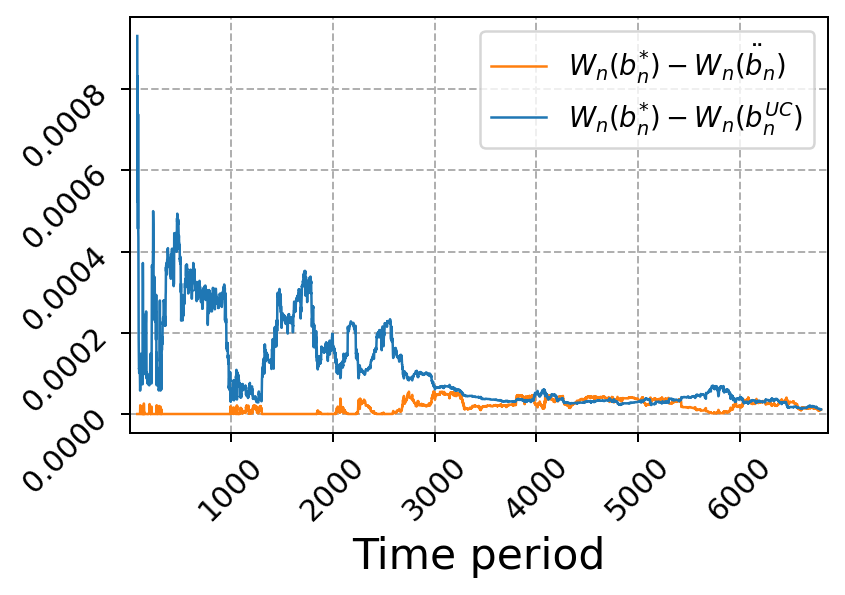} & \includegraphics[scale=0.28]{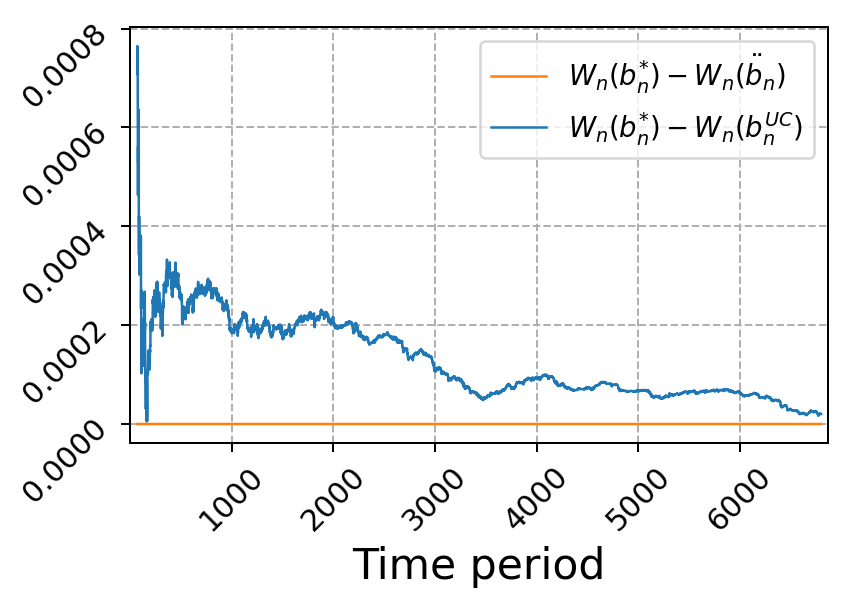}\tabularnewline
\includegraphics[scale=0.28]{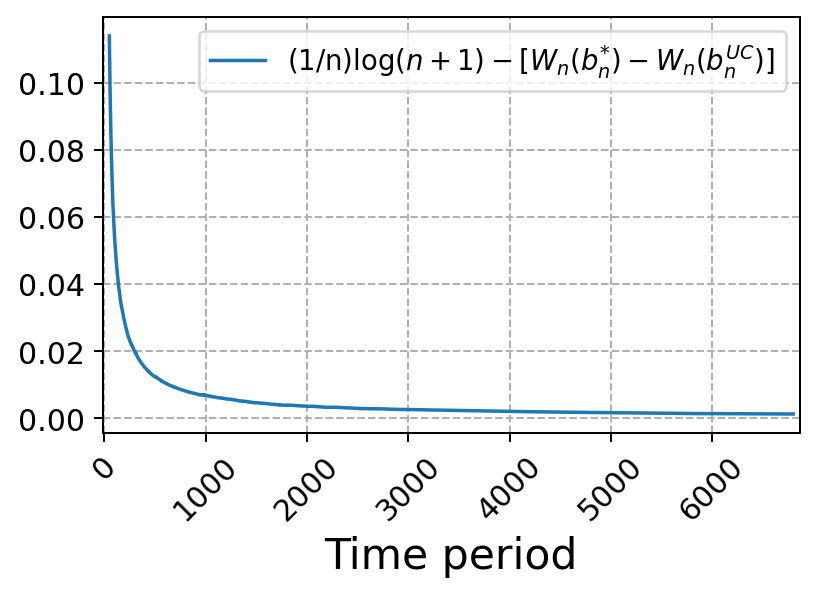} & \includegraphics[scale=0.28]{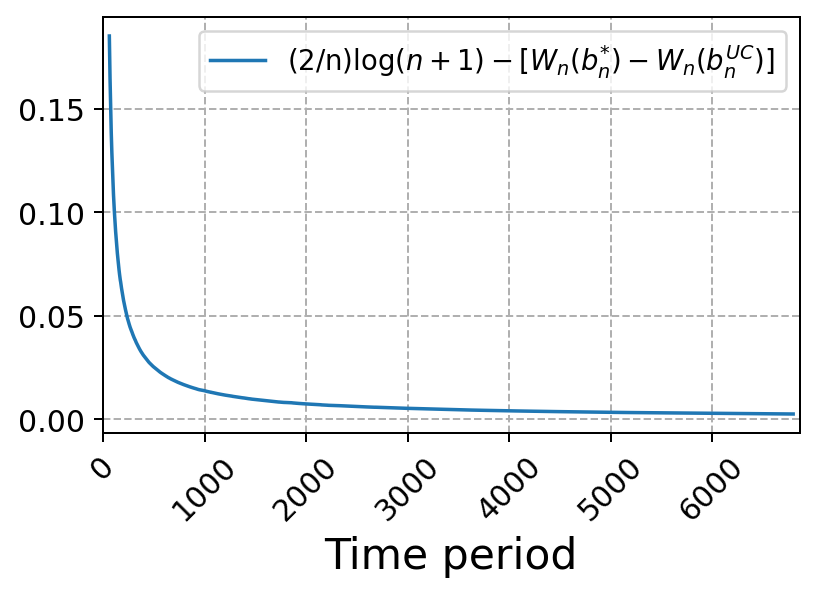} & \includegraphics[scale=0.28]{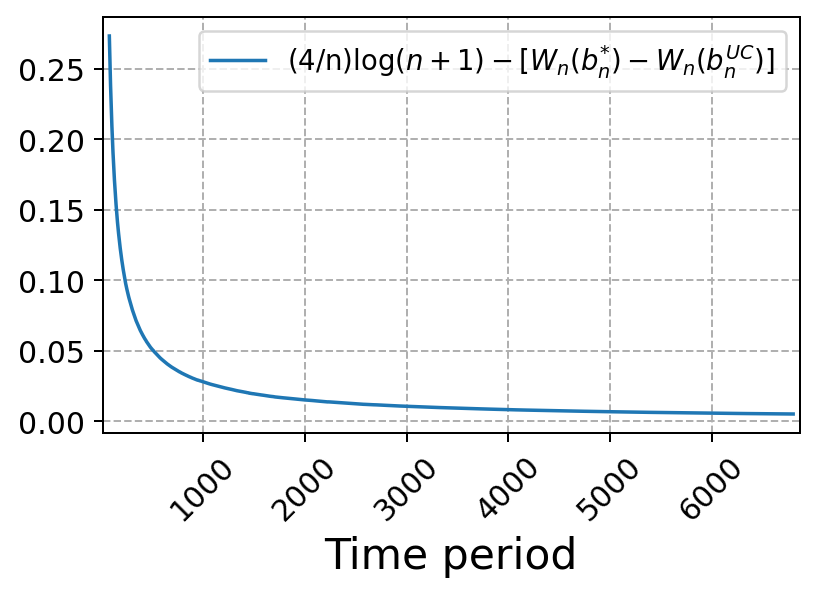} & \includegraphics[scale=0.28]{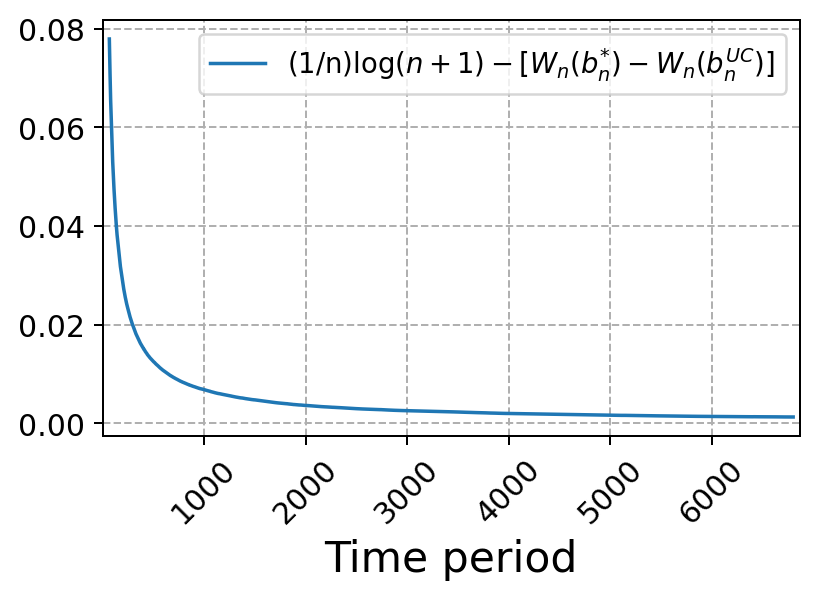}\tabularnewline
{\footnotesize ~~~Experiment 1} & {\footnotesize ~~~Experiment 2} & {\footnotesize ~~~Experiment 3} & {\footnotesize ~~~Experiment 4}\tabularnewline
\end{tabular}
\par\end{centering}
\caption{Upper row: differences over time between the growth rates of the strategies
$\big(\ddot{b}_{n}\big)$, $\big(b_{n}^{*}\big)$ and $\big(b_{n}^{\text{UC}}\big)$.
Lower row: upper bound over time according to Proposition \ref{finite strategy}.\label{Figure 6}}
\end{figure}

The second row of Figure \ref{Figure 6} depicts the upper bound for
the difference between the two numerically computed strategies, $\log S_{n}\big(b_{n}^{*}\big)-\log S_{n}\big(b_{n}^{\text{UC}}\big)$,
as established in Proposition \ref{finite strategy}. Furthermore,
in order to fully comprehend the behavior of the benchmark strategy
$\big(b_{n}^{*}\big)$, the first row of Figure \ref{Figure 7} exhibits
the best numerically constant combinations over time with respect
to different discretization step sizes in each experiment. The first
three graphics indicate the presence of only two non-zero components
of the time-variant maximizer $\lambda_{n}$, while the last graphic
shows only one non-zero component. This is due to the M-V strategy
$\big(b_{n}^{0.002}\big)$ consistently dominating the M-V strategy
$\big(b_{n}^{0.05}\big)$ in the sense of Proposition \ref{S(lamda Y) =00003D S(lamda gamma)},
which explains the violation of the Combination criterion in Experiment
4. In contrast, the second row of Figure \ref{Figure 7} illustrates
the proportions in the UC strategy that the investor allocates to
each component M-V strategy over time, showing that all of them are
consistently invested in throughout.
\begin{figure}[H]
\begin{centering}
\begin{tabular}{cccc}
\includegraphics[scale=0.23]{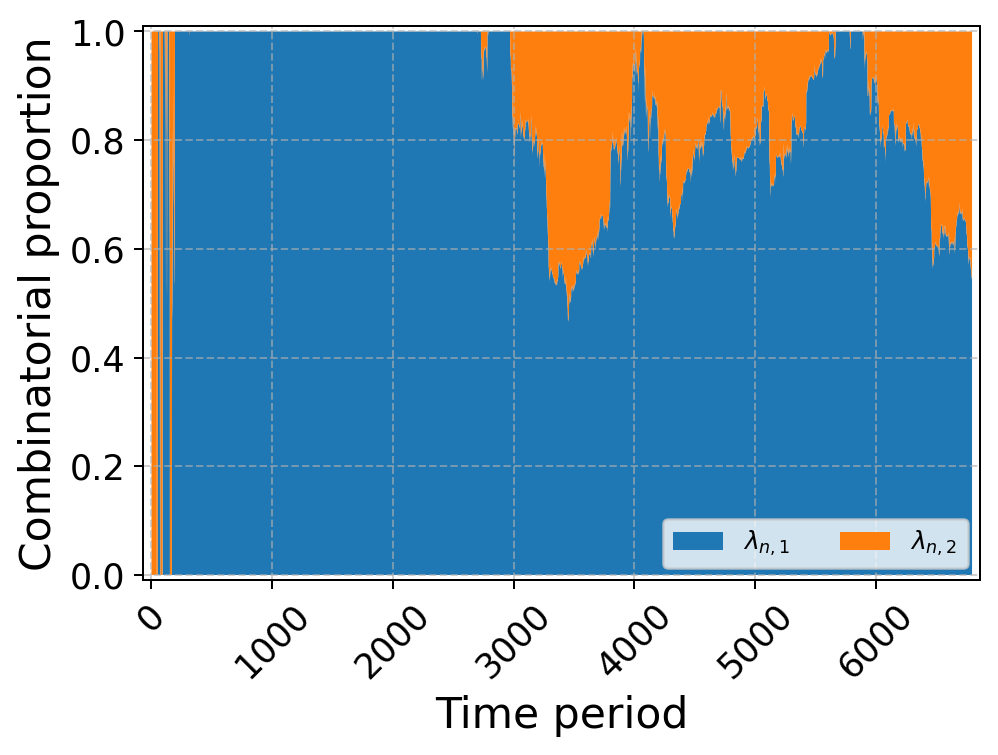} & \includegraphics[scale=0.23]{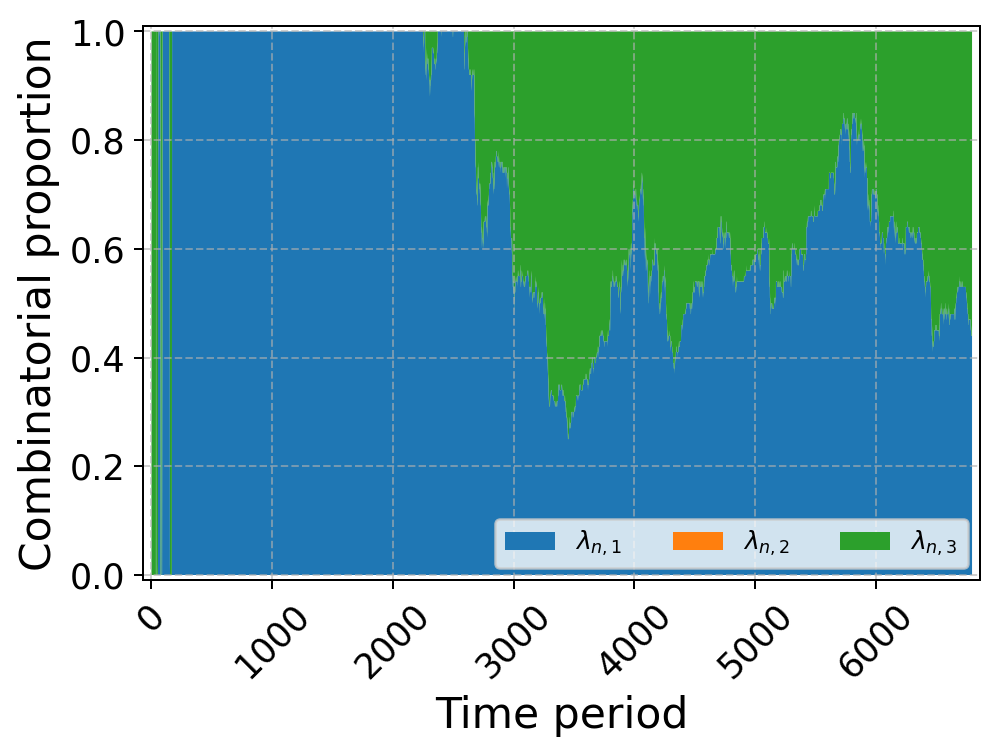} & \includegraphics[scale=0.23]{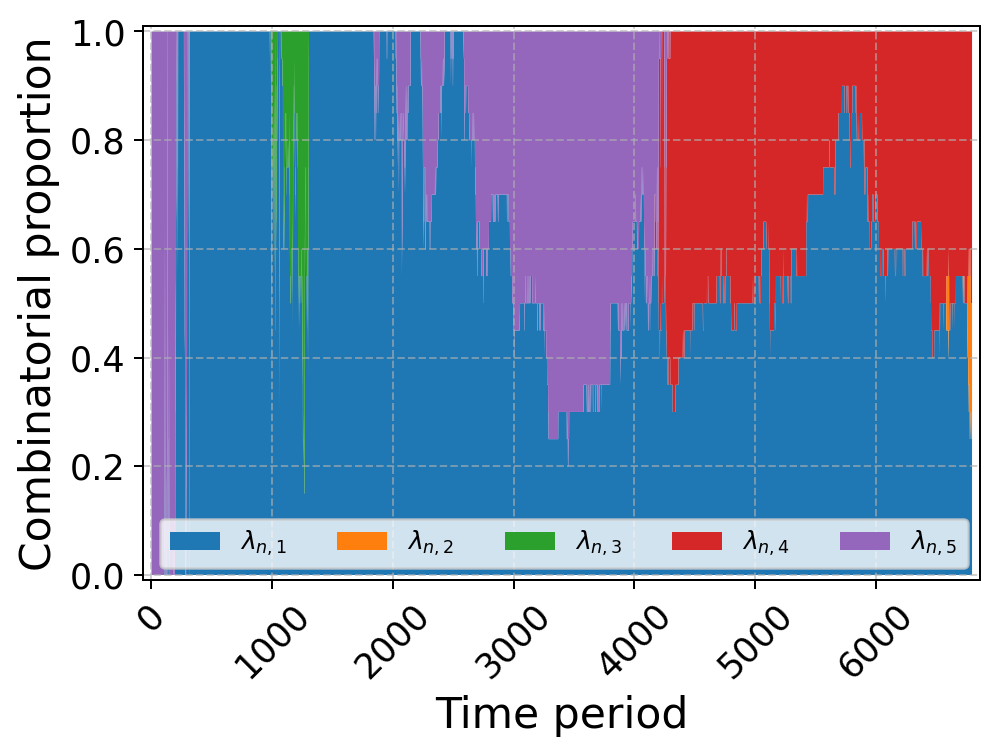} & \multicolumn{1}{c}{\includegraphics[scale=0.23]{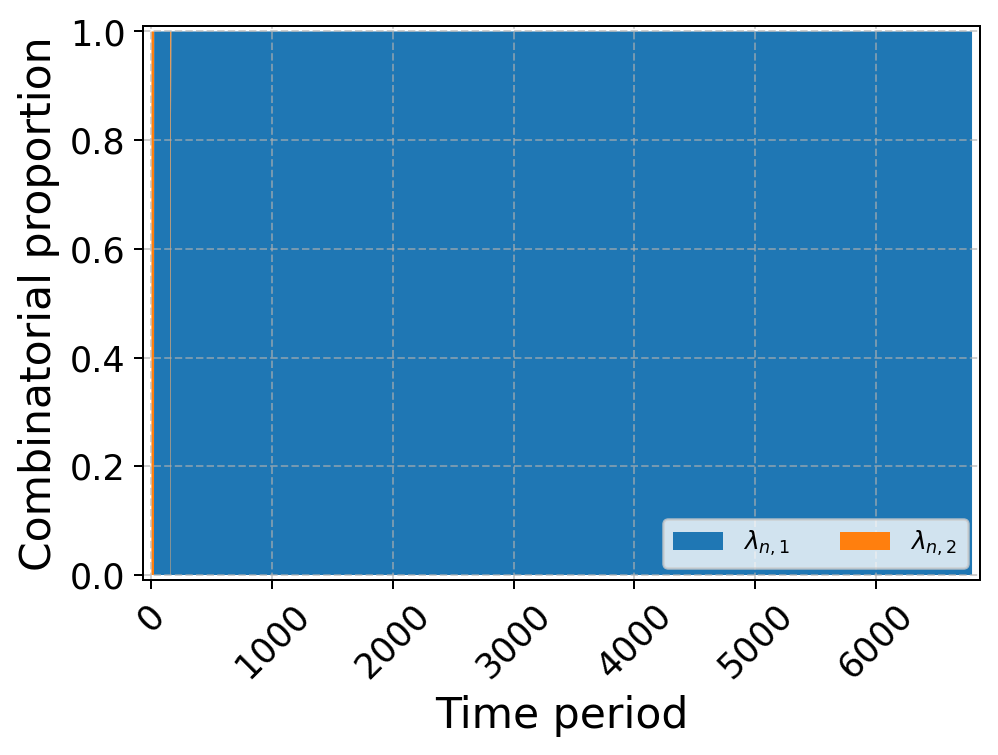}}\tabularnewline
\includegraphics[scale=0.23]{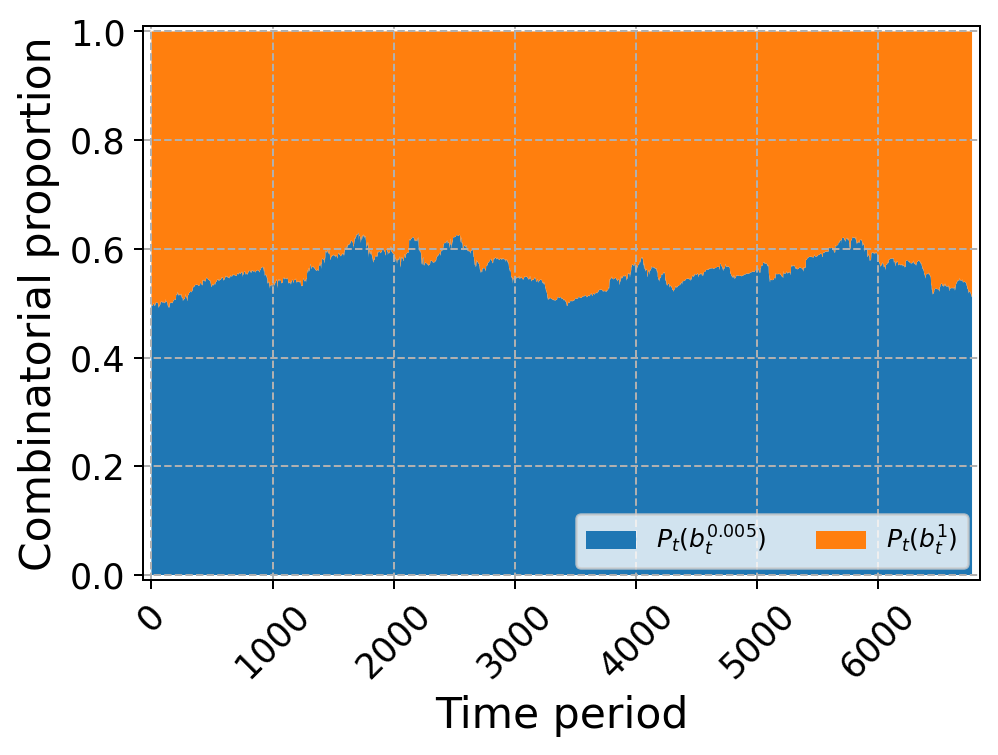} & \includegraphics[scale=0.23]{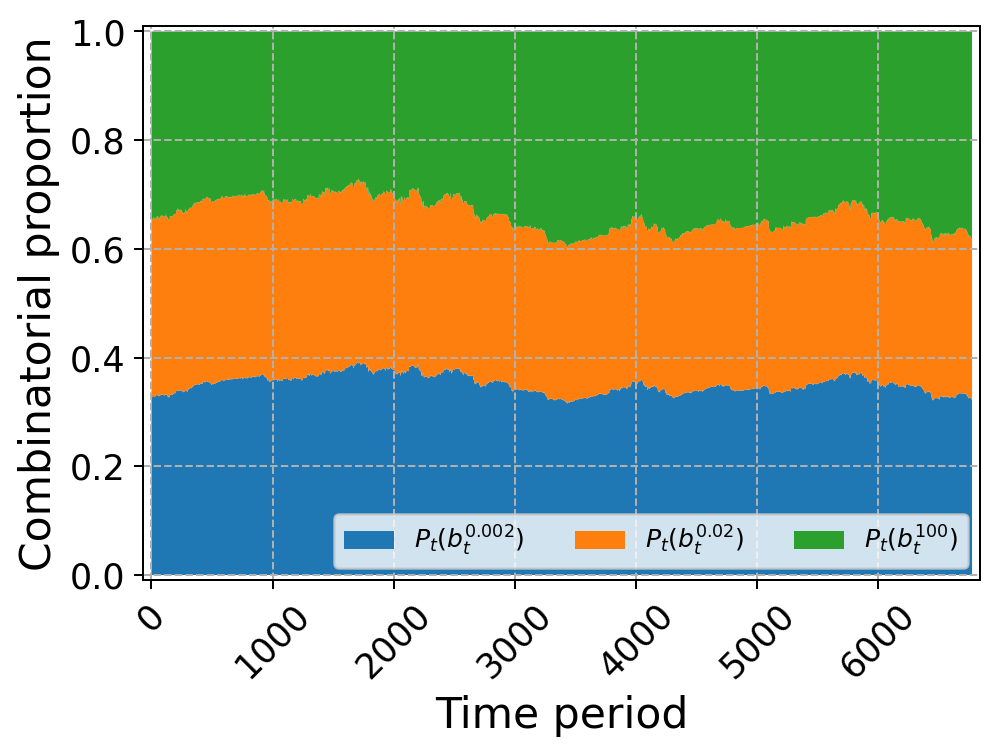} & \includegraphics[scale=0.23]{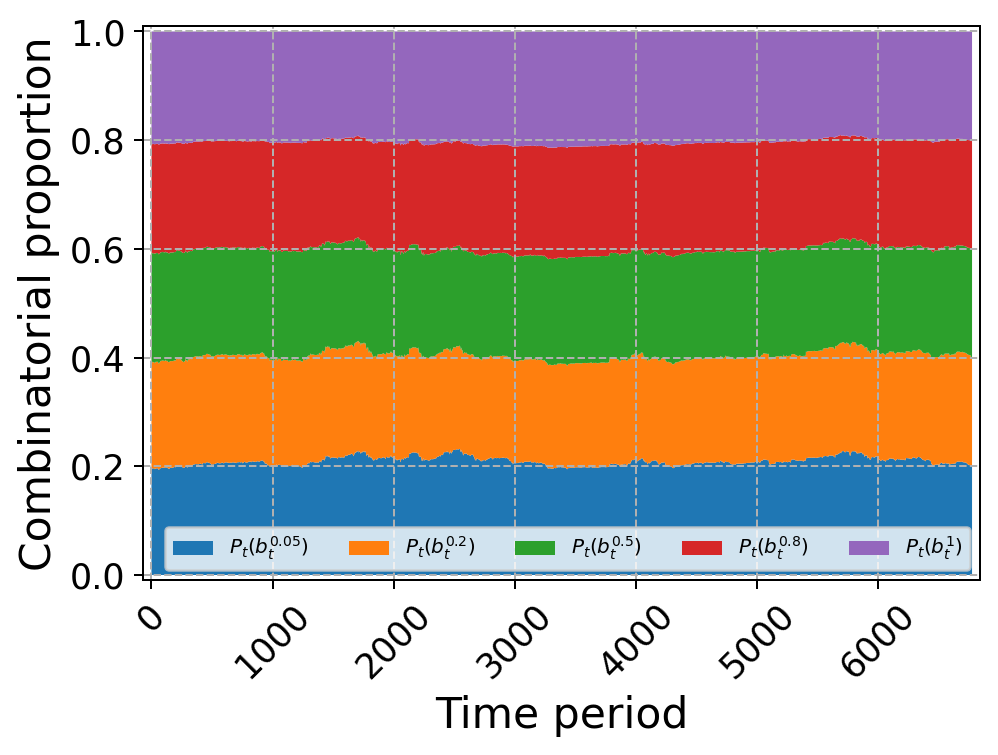} & \includegraphics[scale=0.23]{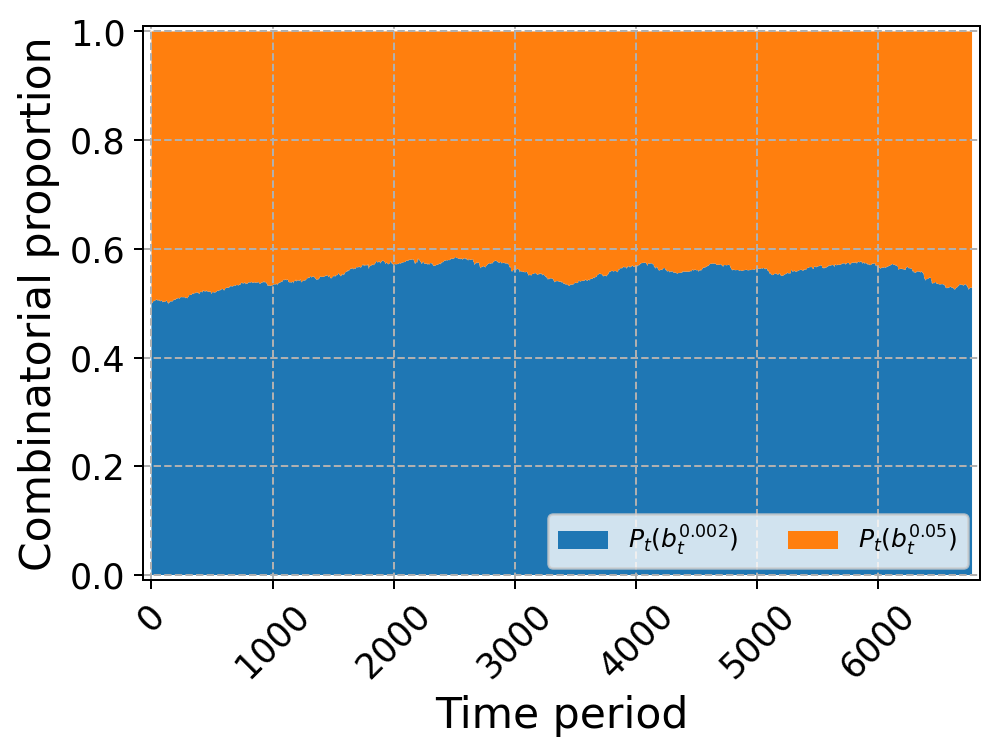}\tabularnewline
{\footnotesize ~~~Experiment 1} & {\footnotesize ~~~Experiment 2} & {\footnotesize ~~~Experiment 3} & \multicolumn{1}{c}{{\footnotesize ~~~Experiment 4}}\tabularnewline
\end{tabular}
\par\end{centering}
\caption{Upper row: the best constant combinations $\lambda_{n}$ of the component
M-V strategies over time with respect to the discretization. Lower
row: allocations of capital proportions into the component strategies,
with $P_{t}$ representing the allocations over time $t$.\label{Figure 7}}
\end{figure}
\smallskip

\textbf{Universally combinatorial strategy with acceleration}. The
results of the first four experiments emphasize the necessity of accelerating
the UC strategy to enhance the converging speed of its growth rate
to that of the benchmark strategy or even surpassing it. Since the
UC strategy allocates capital to all component strategies over time,
it includes unnecessarily those with inferior performance. Particularly,
the underperformance of the UC strategy compared to the significantly
better FL strategy in Experiment 4 justifies our motivation to incorporate
the mechanism of the latter into the former. To address this issue,
the proposed UC-L and UC-W strategies aim to reduce the amount of
capital invested in underperforming component strategies. However,
unlike the original UC strategy, these accelerated strategies do not
come with a theoretical guarantee. Therefore, the UC strategy will
be tested alongside both of its variants UC-L and UC-W in Experiments
5, 6, 7, and 8 to empirically examine their performances. In accordance
with the definitions of UC-L and UC-W strategies, we consider the
time-variant sets $B_{n}$ and $\bar{B}_{n}$ of constant combinations,
which are defined in (\ref{UC-W}) and (\ref{UC-L}), in the top and
bottom $30\%$ and $50\%$ of cumulative wealth from the previous
time period $n-1$ among all combinations in the simplex $\mathcal{B}^{k}$.
Only in Experiment 8, we further consider an additional UC-W strategy
with the top $5\%$ of combinations. The performances of these strategies
are presented in Figure \ref{Figure 8} and Table \ref{Table 2}.
\begin{figure}[H]
\begin{centering}
\includegraphics[scale=0.33]{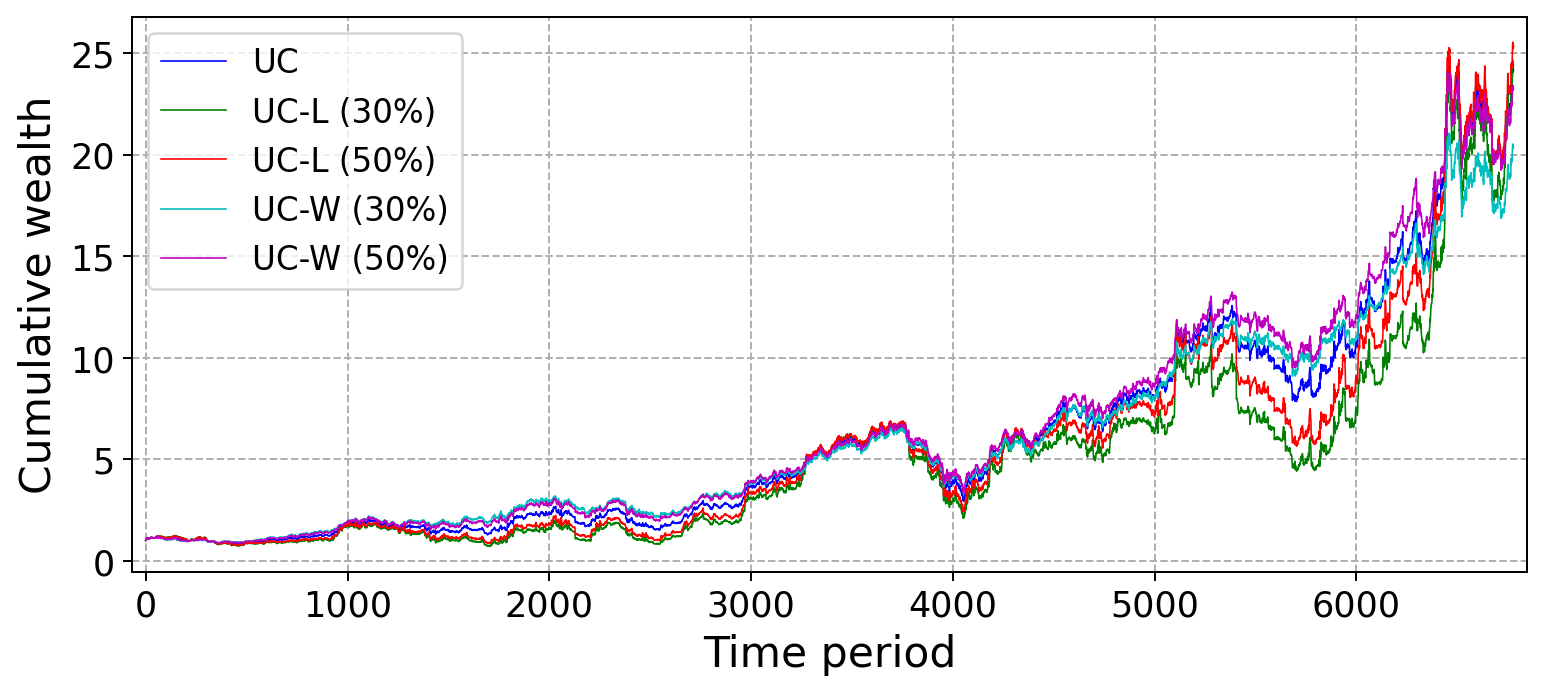} ~\includegraphics[scale=0.33]{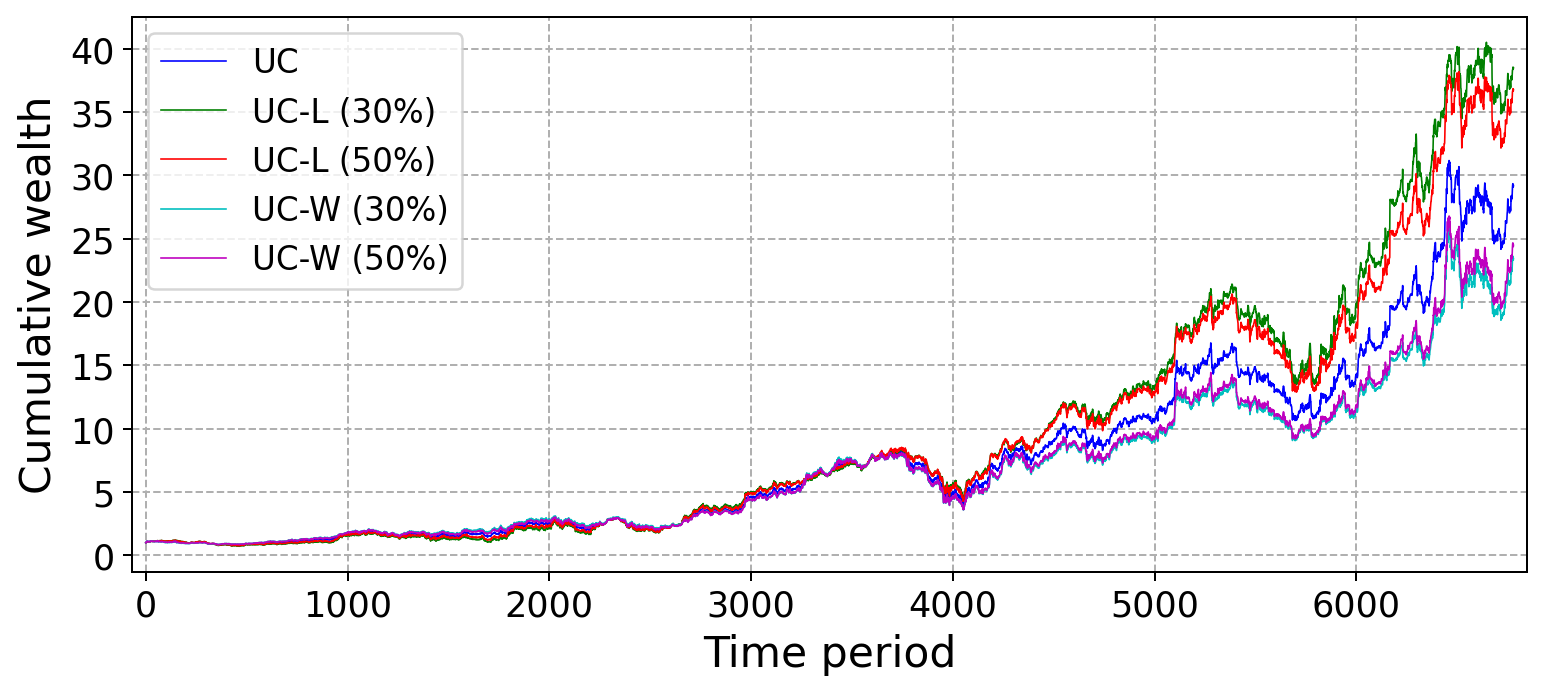}
\par\end{centering}
\begin{centering}
~\includegraphics[scale=0.33]{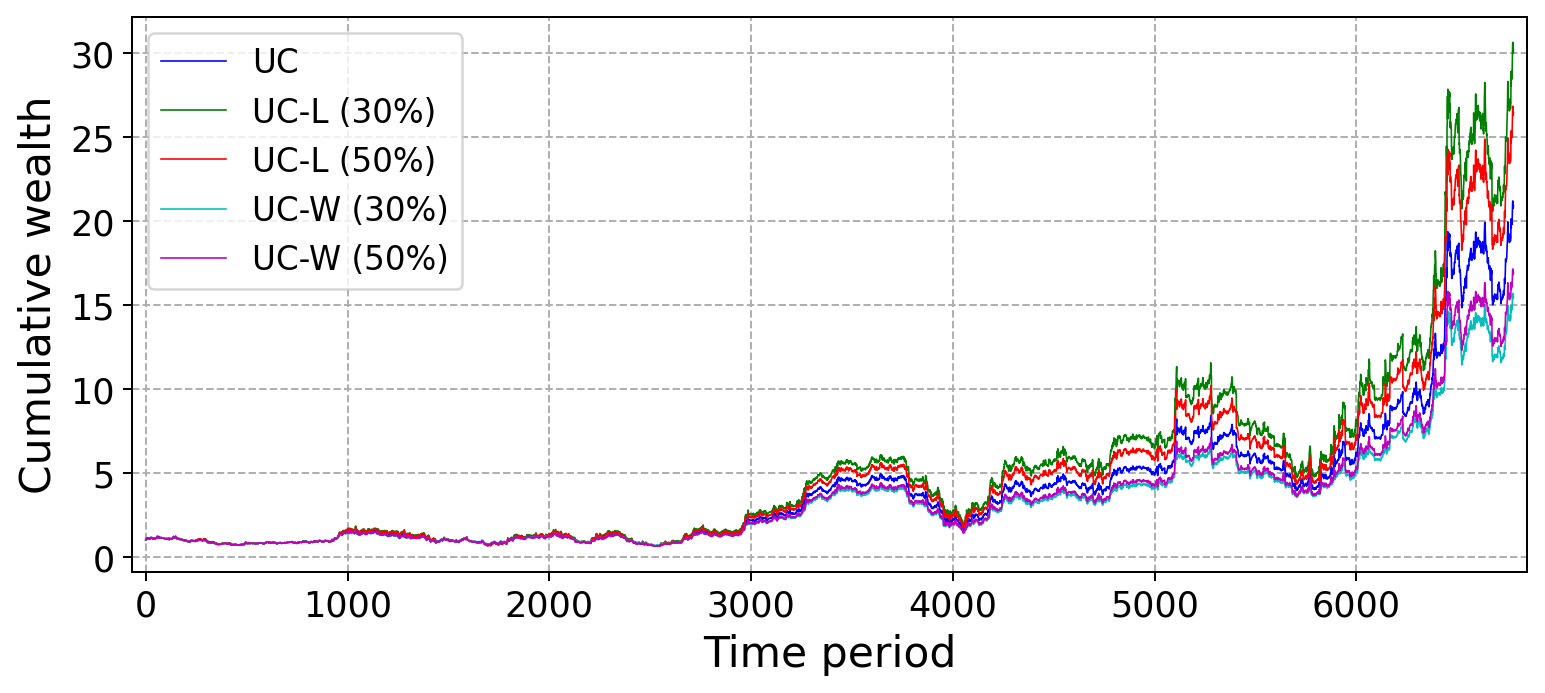} ~~\includegraphics[scale=0.33]{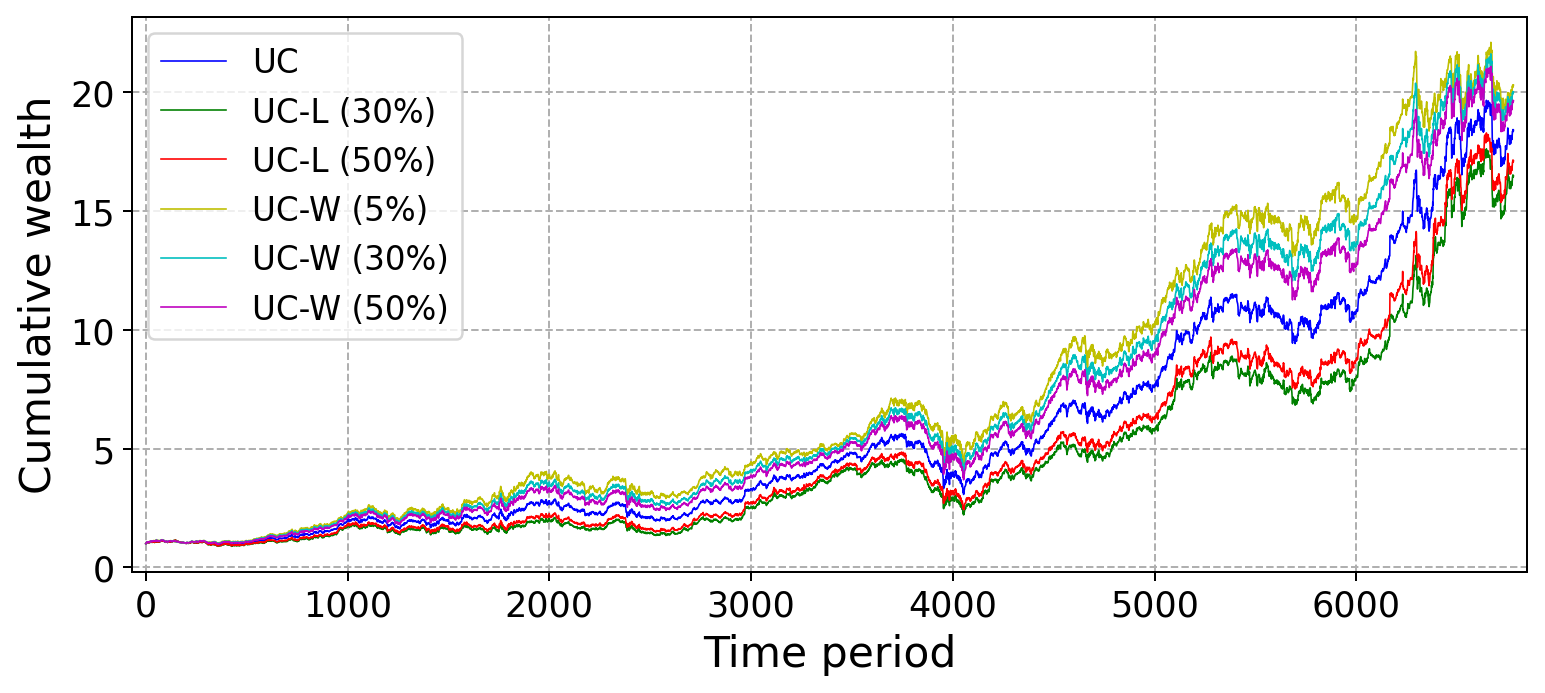}
\par\end{centering}
\caption{Evolutions of the accelerated combinatorial strategies over time in
Experiments 5, 6, 7, and 8 (from left-to-right and top-to-bottom order).
\label{Figure 8}}
\end{figure}
\renewcommand{\arraystretch}{0.3}
\begin{table}[H]
\caption{Performance measures for the strategies in Experiments 5, 6, 7 and
8.\label{Table 2} }

\begin{centering}
\begin{tabular*}{12cm}{@{\extracolsep{\fill}}>{\raggedright}m{2.5cm}>{\raggedright}m{2cm}>{\raggedright}m{2cm}>{\raggedright}m{2cm}>{\raggedright}m{2cm}}
\toprule 
{\tiny\textbf{Strategy}} & {\tiny\textbf{Final wealth}} & {\tiny\textbf{Average growth rate}} & {\tiny\textbf{Average return}} & {\tiny\textbf{Sharpe ratio}}\tabularnewline
\bottomrule
\end{tabular*}
\par\end{centering}
\begin{centering}
\begin{tabular*}{12cm}{@{\extracolsep{\fill}}>{\raggedright}p{2.5cm}>{\raggedright}p{2cm}>{\raggedright}p{2cm}>{\raggedright}p{2cm}>{\raggedright}p{2cm}}
\toprule 
\multicolumn{5}{c}{{\tiny Experiment 5 with risk aversion $\alpha\in\left\{ 0.005,1\right\} $}}\tabularnewline
\midrule
{\tiny UC} & {\tiny 24.178108} & {\tiny 0.000470} & {\tiny 1.000580} & {\tiny 67.419175}\tabularnewline
{\tiny UC-L (30\%)} & {\tiny 24.414120} & {\tiny 0.000471} & {\tiny 1.000710} & {\tiny 45.790848}\tabularnewline
{\tiny UC-L (50\%)} & {\tiny 25.333606} & {\tiny 0.000477} & {\tiny 1.000668} & {\tiny 51.143733}\tabularnewline
{\tiny UC-W (30\%)} & {\tiny 20.393190} & {\tiny 0.000445} & {\tiny 1.000519} & {\tiny 82.463092}\tabularnewline
{\tiny UC-W (50\%)} & {\tiny 23.260910} & {\tiny 0.000464} & {\tiny 1.000545} & {\tiny 78.922235}\tabularnewline
\end{tabular*}
\par\end{centering}
\begin{centering}
\begin{tabular*}{12cm}{@{\extracolsep{\fill}}>{\raggedright}p{2.5cm}>{\raggedright}p{2cm}>{\raggedright}p{2cm}>{\raggedright}p{2cm}>{\raggedright}p{2cm}}
\toprule 
\multicolumn{5}{c}{{\tiny Experiment 6 with risk aversion $\alpha\in\left\{ 0.002,0.02,100\right\} $}}\tabularnewline
\midrule
{\tiny UC} & {\tiny 29.192687} & {\tiny 0.000498} & {\tiny 1.000597} & {\tiny 71.217927}\tabularnewline
{\tiny UC-L (30\%)} & {\tiny 38.484047} & {\tiny 0.000539} & {\tiny 1.000668} & {\tiny 62.238247}\tabularnewline
{\tiny UC-L (50\%)} & {\tiny 36.758159} & {\tiny 0.000532} & {\tiny 1.000646} & {\tiny 66.163429}\tabularnewline
{\tiny UC-W (30\%)} & {\tiny 23.426020} & {\tiny 0.000465} & {\tiny 1.000559} & {\tiny 72.989659}\tabularnewline
{\tiny UC-W (50\%)} & {\tiny 24.483100} & {\tiny 0.000472} & {\tiny 1.000567} & {\tiny 72.393094}\tabularnewline
\end{tabular*}
\par\end{centering}
\begin{centering}
\begin{tabular*}{12cm}{@{\extracolsep{\fill}}>{\raggedright}p{2.5cm}>{\raggedright}p{2cm}>{\raggedright}p{2cm}>{\raggedright}p{2cm}>{\raggedright}p{2cm}}
\toprule 
\multicolumn{5}{c}{{\tiny Experiment 7 with risk aversion $\alpha\in\left\{ 0.05,0.2,0.5,0.8,1\right\} $}}\tabularnewline
\midrule
{\tiny UC} & {\tiny 20.887842} & {\tiny 0.000448} & {\tiny 1.000645} & {\tiny 50.409559}\tabularnewline
{\tiny UC-L (30\%)} & {\tiny 30.182974} & {\tiny 0.000503} & {\tiny 1.000741} & {\tiny 45.813842}\tabularnewline
{\tiny UC-L (50\%)} & {\tiny 26.442502} & {\tiny 0.000483} & {\tiny 1.000712} & {\tiny 46.693717}\tabularnewline
{\tiny UC-W (30\%)} & {\tiny 15.468676} & {\tiny 0.000404} & {\tiny 1.000564} & {\tiny 55.928931}\tabularnewline
{\tiny UC-W (50\%)} & {\tiny 16.903366} & {\tiny 0.000417} & {\tiny 1.000590} & {\tiny 53.769018}\tabularnewline
\end{tabular*}
\par\end{centering}
\begin{centering}
\begin{tabular*}{12cm}{@{\extracolsep{\fill}}>{\raggedright}p{2.5cm}>{\raggedright}p{2cm}>{\raggedright}p{2cm}>{\raggedright}p{2cm}>{\raggedright}p{2cm}}
\toprule 
\multicolumn{5}{c}{{\tiny Experiment 8 with risk aversion $\alpha\in\left\{ 0.002,0.05\right\} $}}\tabularnewline
\midrule
{\tiny UC} & {\tiny 18.401030} & {\tiny 0.000430} & {\tiny 1.000493} & {\tiny 88.740964}\tabularnewline
{\tiny UC-L (30\%)} & {\tiny 16.480301} & {\tiny 0.000413} & {\tiny 1.000489} & {\tiny 81.431545}\tabularnewline
{\tiny UC-L (50\%)} & {\tiny 17.113231} & {\tiny 0.000419} & {\tiny 1.000490} & {\tiny 83.761159}\tabularnewline
{\tiny UC-W (5\%)} & {\tiny 20.299331} & {\tiny 0.000444} & {\tiny 1.000505} & {\tiny 90.939658}\tabularnewline
{\tiny UC-W (30\%)} & {\tiny 20.010540} & {\tiny 0.000442} & {\tiny 1.000503} & {\tiny 91.026628}\tabularnewline
{\tiny UC-W (50\%)} & {\tiny 19.641169} & {\tiny 0.000439} & {\tiny 1.000500} & {\tiny 90.694303}\tabularnewline
\bottomrule
\end{tabular*}
\par\end{centering}
\centering{}%
\begin{minipage}[t]{12cm}%
\begin{spacing}{0.5}
{\tiny\textbf{Note}}{\tiny . The UC strategies are approximated by
2001 discretization points by step 0.0005 in Experiments 5 and 8,
4947 discretization points by step 0.01 in Experiment 6, and 9821
discretization points by step 0.05 in Experiment 7.}
\end{spacing}
\end{minipage}
\end{table}
\vspace{-2ex}

The results presented in Figure \ref{Figure 8} align with our expectations.
Since the FL strategy works in Experiment 4, it results in the UC-W
($5\%$) strategy demonstrating the most favorable evolution of cumulative
wealth compared to the other tested strategies in the same scenario
in Experiment 8. This can be attributed to the almost always exceedance
of the M-V strategy $\big(b_{n}^{0.002}\big)$ over the M-V strategy
$\big(b_{n}^{0.05}\big)$, which causes a combinatorial strategy corresponding
to a constant combination closer to the best one, $\lambda_{n}=(1,0)$
for all $n$, to yield higher cumulative wealth. Consequently, the
UC-W strategy performs better with a small subset $B_{n}$ that includes
the point $(1,0)$, while the UC-L strategy with a small subset $\bar{B}_{n}$
that includes the point $(0,1)$ leads to a decline in cumulative
wealth. However, in the other three experiments, the UC-W strategies
underperform, whereas the UC-L strategies show significant improvement.
A possible explanation for this behavior is similar to the argument
explaining the poor performance of the FL strategy in the corresponding
scenarios. When strongly fluctuating strategies are present among
the component M-V strategies, relying solely on constant combination
strategies that previously led in terms of cumulative wealth becomes
less reliable. Additionally, Table \ref{Table 2} reaffirms a characteristic
inherent in the UC-W and UC-L strategies derived from the UC strategy,
namely, that a strategy with higher final cumulative wealth tends
to produce a slight decrease in its Sharpe ratio{\small .}{\small{}
\begin{figure}[H]
\begin{centering}
{\small{}%
\begin{tabular}{cccc}
{\small\includegraphics[scale=0.23]{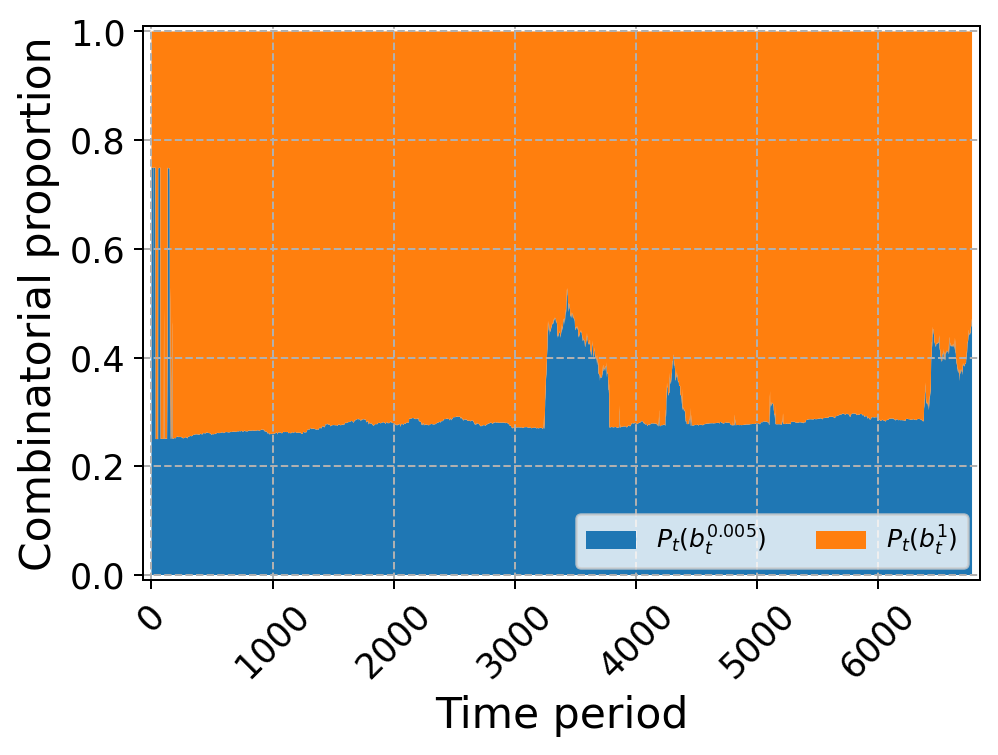}} & {\small\includegraphics[scale=0.23]{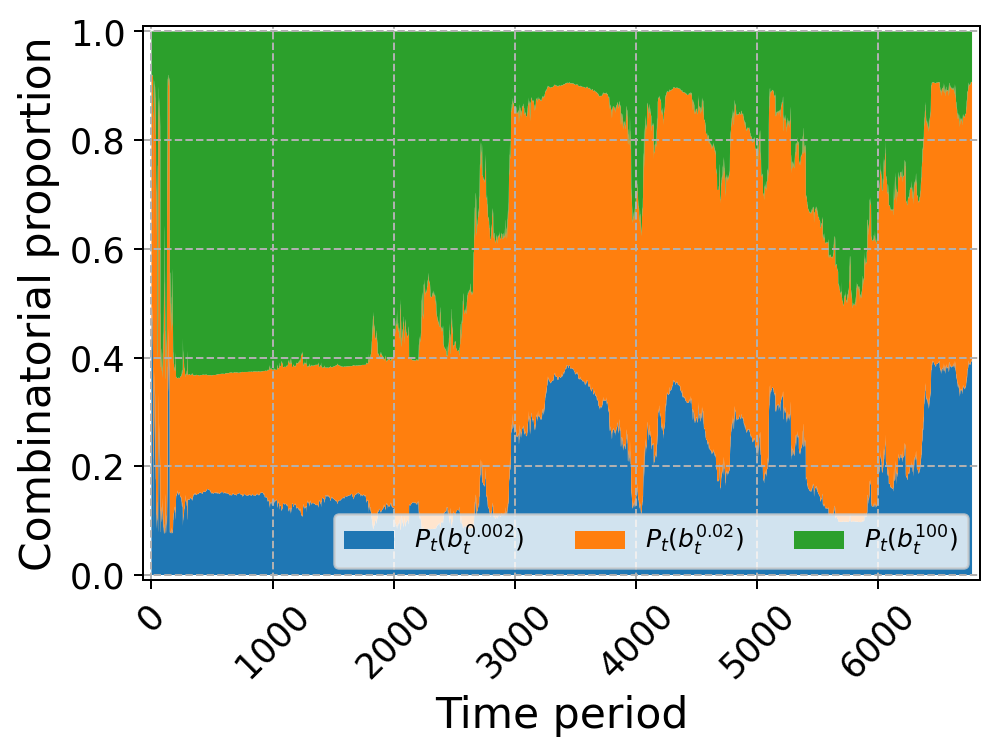}} & {\small\includegraphics[scale=0.23]{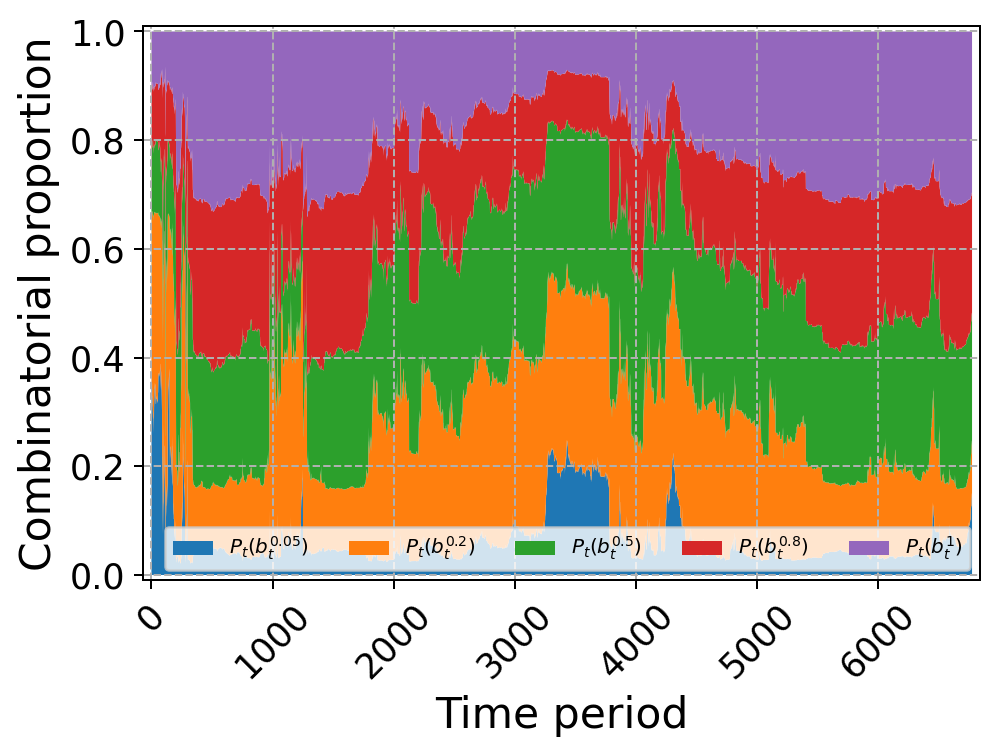}} & {\small\includegraphics[scale=0.23]{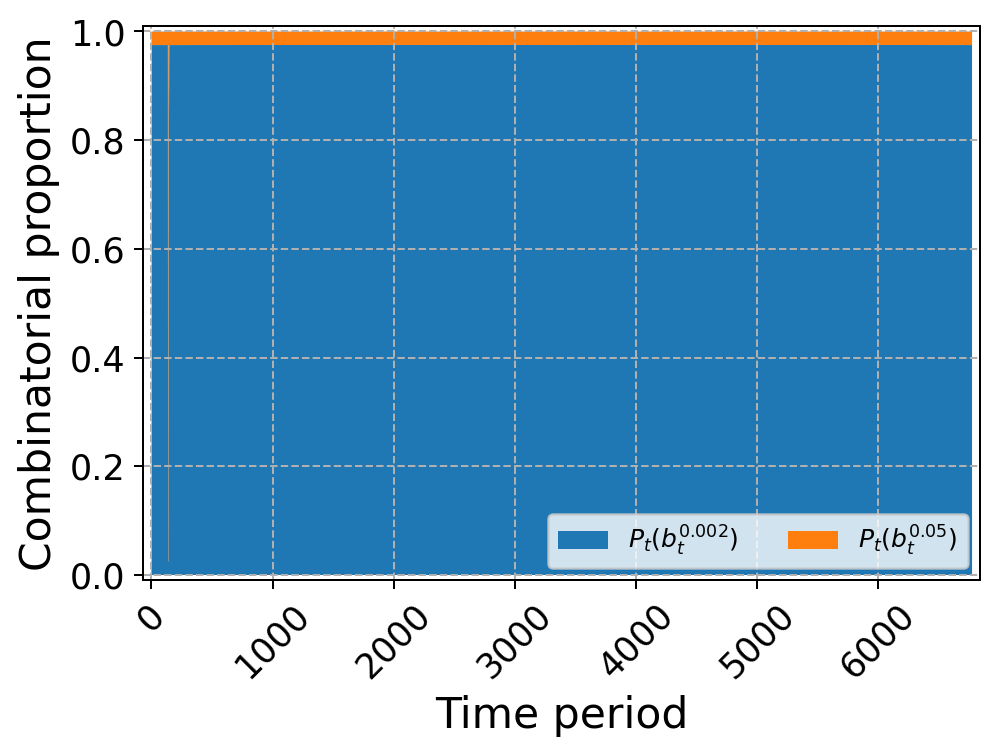}}\tabularnewline
{\small\includegraphics[scale=0.27]{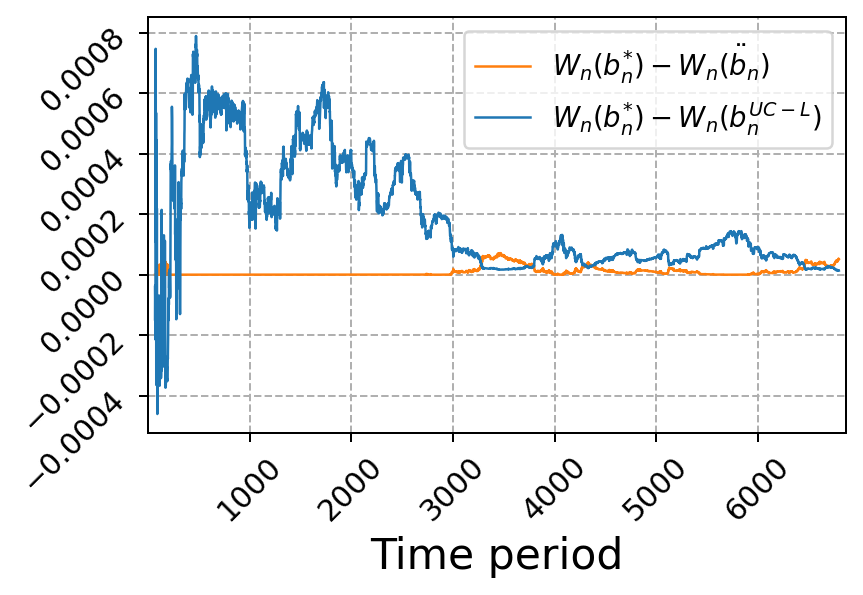}} & {\small\includegraphics[scale=0.27]{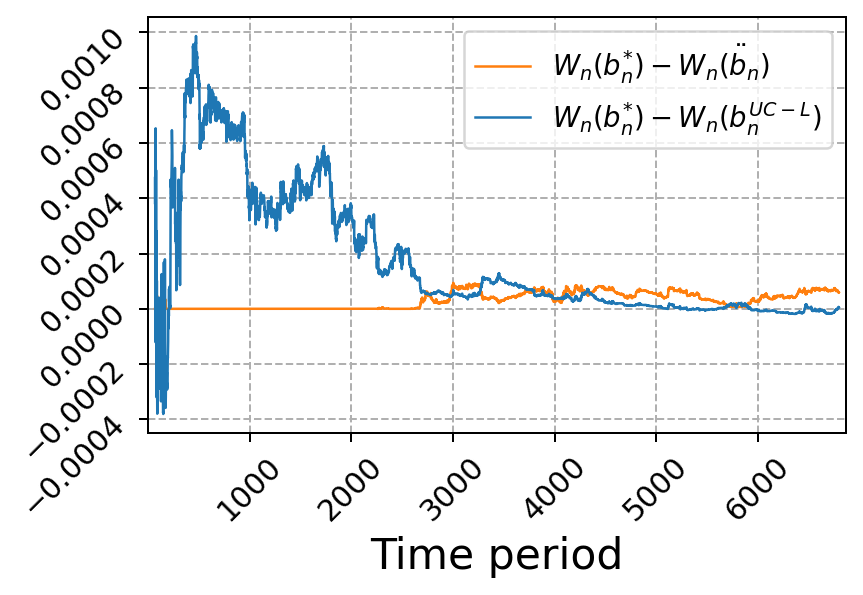}} & {\small\includegraphics[scale=0.27]{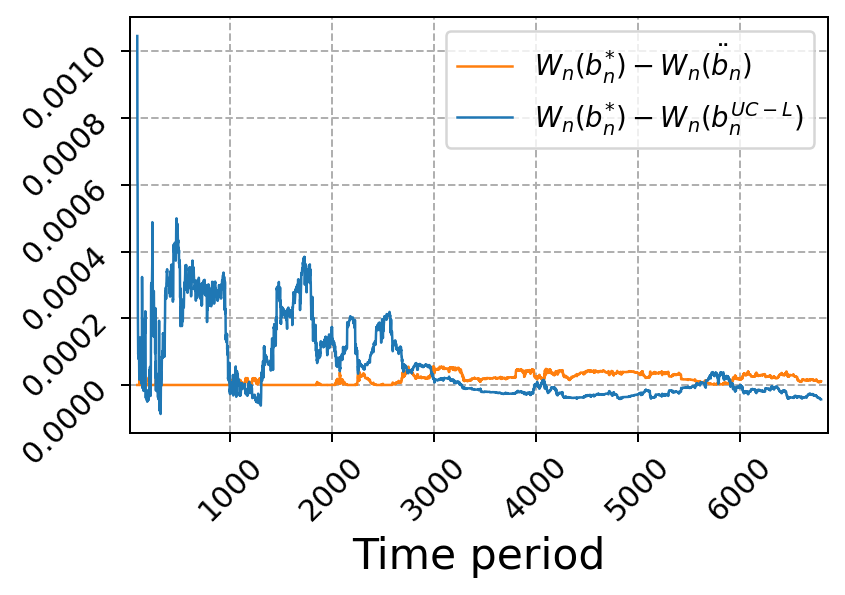}} & {\small\includegraphics[scale=0.27]{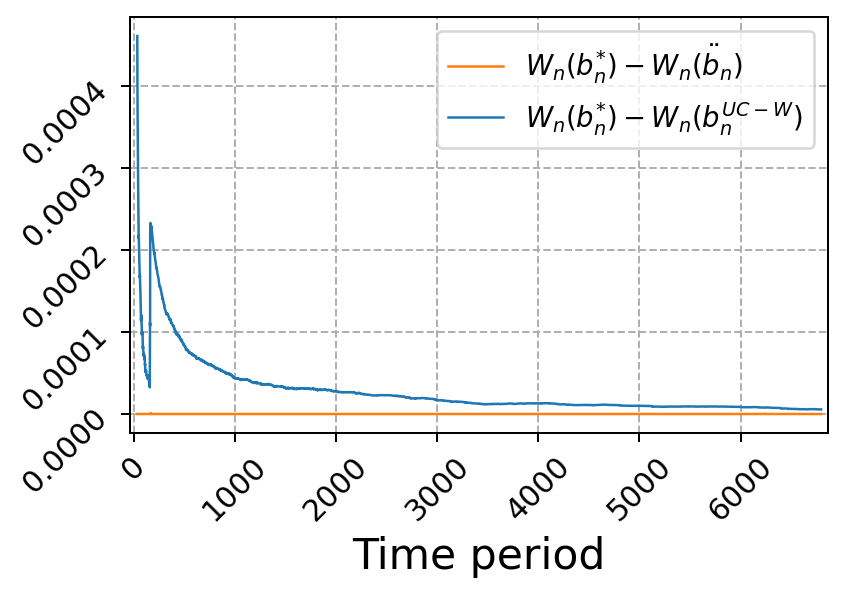}}\tabularnewline
{\scriptsize ~Experiment 5: UC-L (30\%)} & {\scriptsize ~Experiment 6: UC-L (30\%)} & {\scriptsize ~Experiment 7: UC-L (30\%)} & {\scriptsize ~~Experiment 8: UC-W (5\%)}\tabularnewline
\end{tabular}}{\small\par}
\par\end{centering}
{\small\caption{Upper row: allocations of capital proportions into the component M-V
strategies, with $P_{t}$ representing the allocations over time $t$.
Lower row: differences over time between the growth rates of the strategies
$\big(\ddot{b}_{n}\big)$, $\big(b_{n}^{*}\big)$ and $\big(b_{n}^{\text{UC-L}}\big)$
or $\big(b_{n}^{\text{UC-W}}\big)$.\label{Figure 9}}
}{\small\par}
\end{figure}
}{\small\par}

The acceleration results in a noticeable change in the capital allocation,
as depicted in the first row of \ref{Figure 9}. In Experiment 8,
the UC-W ($5\%$) strategy primarily allocates capital to the M-V
strategy $\big(b_{n}^{0.002}\big)$, which closely aligns with the
optimal constant combination $(1,0)$ over time. Conversely, in the
remaining three experiments, the UC-L ($30\%$) strategy exhibits
significant fluctuations in the allocated proportions among the involved
component M-V strategies over time, in contrast to the gradual changes
observed in the proportions of the UC strategy. These substantial
fluctuations in allocated proportions subsequently impact the differences
in growth rate between the benchmark strategy and the accelerated
strategies, as shown in the second row of Figure \ref{Figure 9}.
The first three graphs indicate that the growth rate $W_{n}\big(b_{n}^{\text{UC-L}}\big)$
is not bounded by $W_{n}\big(b_{n}^{*}\big)$, whereas $W_{n}\big(b_{n}^{\text{UC-W}}\big)$
approaches the upper bound $W_{n}\big(b_{n}^{*}\big)$ in the final
graph. The accelerated UC-L strategy in Experiments 6 and 7 quickly
surpasses the benchmark strategy, likely due to the number of component
M-V strategies that exhibit significant fluctuations in cumulative
wealth evolution. In Experiment 5, the consistent exceeding of the
M-V strategy $\big(b_{n}^{0.005}\big)$ over the M-V strategy $\big(b_{n}^{1}\big)$
enables the UC-W strategies, which allocate their capital to small
sets of constant combination strategies that previously led in cumulative
wealth, to perform better than the UC-L strategies most of the time.
However, as the M-V strategy $\big(b_{n}^{1}\big)$ experiences stronger
fluctuations after about $5000$ time periods, the UC-L strategies
gain momentum and eventually surpass the UC-W strategies. In Experiments
6 and 7, characterized by significant fluctuations in at least one
M-V strategy and frequent changes in the leading position, the UC-L
strategy, which consistently allocates its capital to a small set
of constant combination strategies that do not lead in cumulative
wealth but provide stability in returns, quickly surpasses the benchmark
strategy.

\subsection{Experiments with large scales of component strategies}

After conducting eight experiments in the previous section, we extend
our investigation to the UC strategy given by (\ref{eq:3.6}) and
its accelerated variant in a scenario involving a large scale of component
M-V strategies. For this analysis, we consider a set $\mathcal{A}$
of risk aversion coefficients encompassing values ranging from $0.001$
to $0.009$ with a step size of $0.001$, from $0.01$ to $0.09$
with a step size of $0.01$, from $0.1$ to $0.9$ with a step size
of $0.1$, and from $1$ to $10$ with a step size of $0.5$. In Experiment
9, we examine the impact of partitioning $\mathcal{A}$ into multiple
base sets and their respective sizes on the performance and behavior
of the UC strategy. To investigate this further, we analyze partitions
into two, four, and six base sets. Specifically, in the case of two
base sets, our analysis includes examining balanced sets with equal
sizes spanning from $0.001$ to $0.5$ and from $0.6$ to $10$, as
well as unbalanced sets covering the ranges of $0.001$ to $0.02$
and $0.03$ to $10$. In the case of four base sets, we divide the
coefficients from $0.001$ to $5$ into three equal-sized sets, while
assigning the remainder to the fourth set. In a similar vein, for
six base sets, we evenly segment the coefficients ranging from $0.001$
to $7$ into the first five sets, while leaving the remainder for
the final set. Following Experiment 9, we proceed to implement accelerated
strategies tailored to the scenario involving six base sets in Experiment
10.
\begin{figure}[H]
\begin{centering}
\includegraphics[scale=0.33]{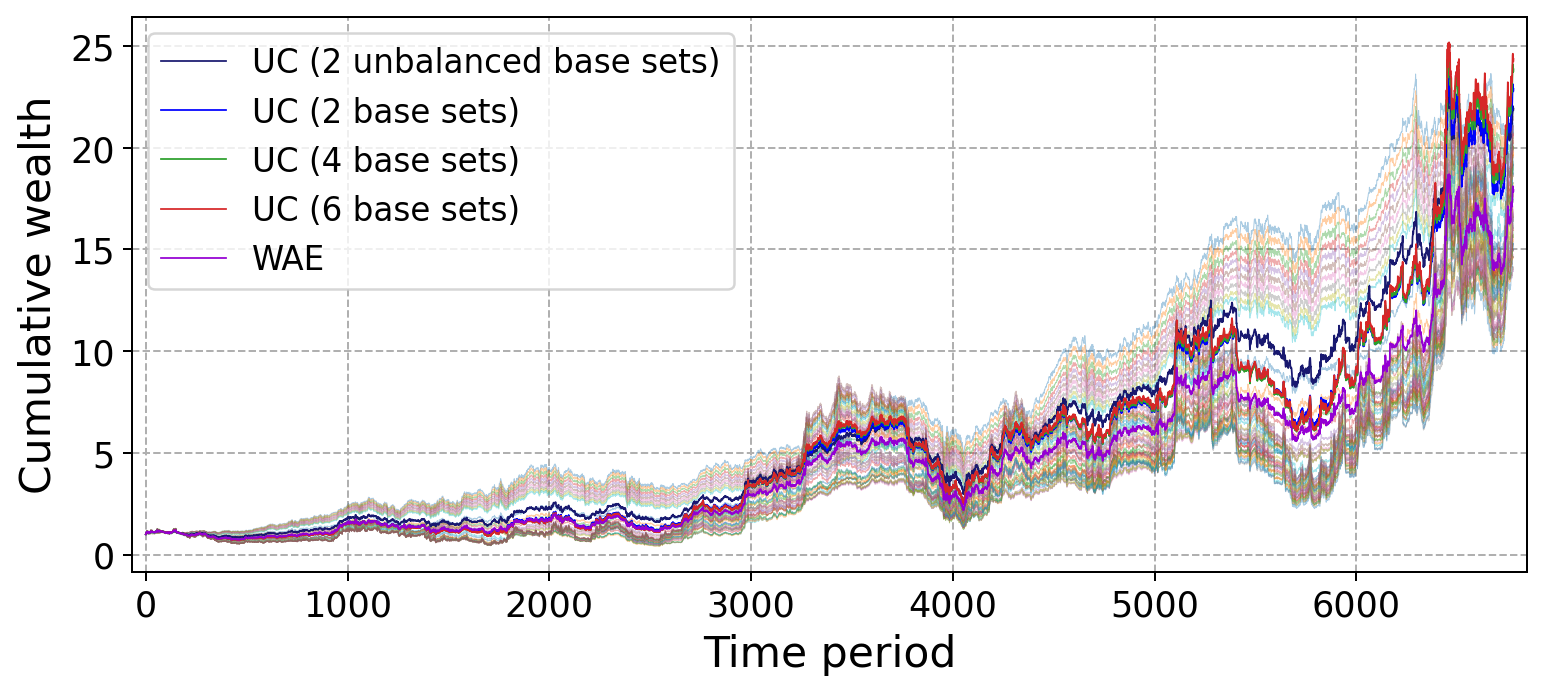} ~\includegraphics[scale=0.33]{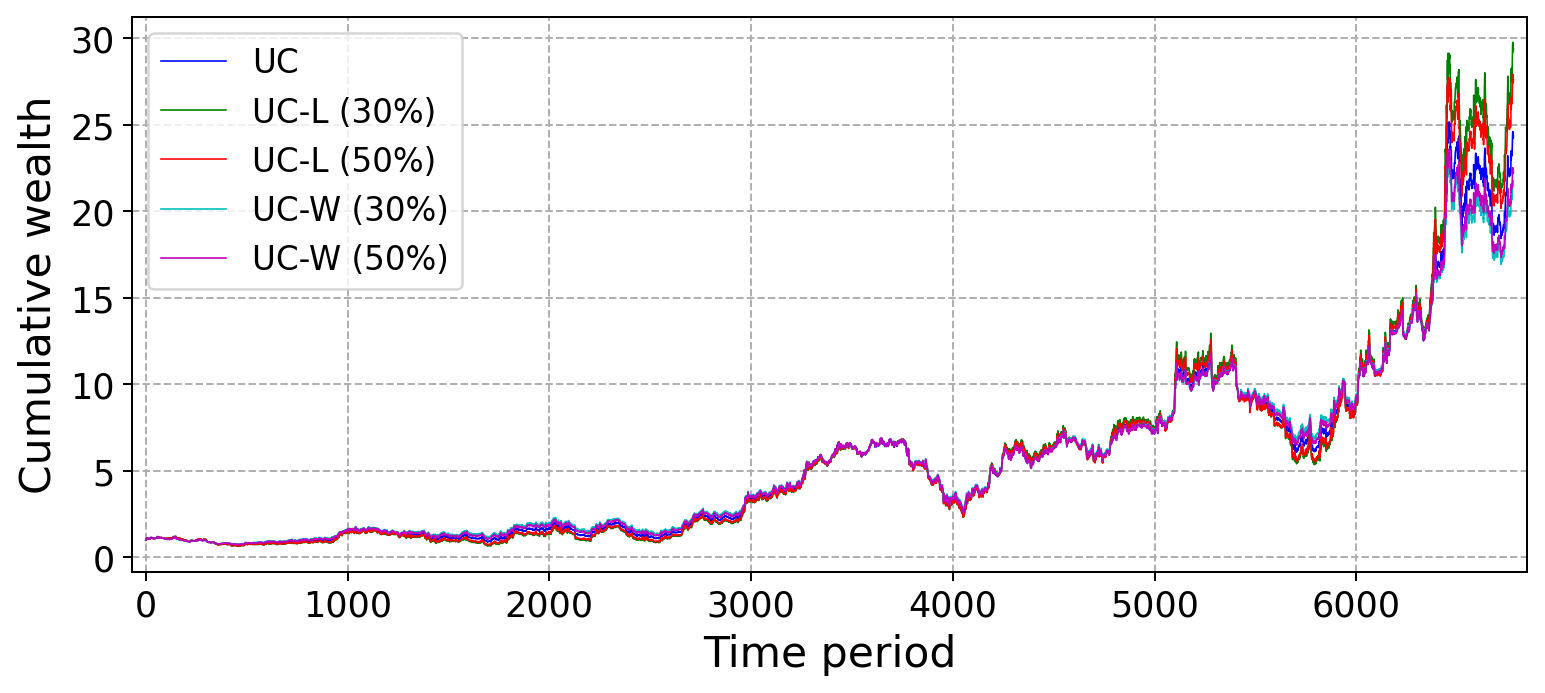}
\par\end{centering}
\caption{Left: evolutions of the combinatorial strategies over time with different
base sets in Experiment 9. Right: evolutions of the accelerated combinatorial
strategies over time with six base sets in Experiment 10. The blurred
background depicts the evolutions of the component M-V strategies.\label{Figure 10}}
\end{figure}
\renewcommand{\arraystretch}{0.3}
\begin{table}[H]
\caption{Performance measures for the strategies in Experiments 9 and 10. \label{Table 3}}

\begin{centering}
\begin{tabular*}{12cm}{@{\extracolsep{\fill}}>{\raggedright}m{2.5cm}>{\raggedright}m{2cm}>{\raggedright}m{2cm}>{\raggedright}m{2cm}>{\raggedright}m{2cm}}
\toprule 
{\tiny\textbf{Strategy}} & {\tiny\textbf{Final wealth}} & {\tiny\textbf{Average growth rate}} & {\tiny\textbf{Average return}} & {\tiny\textbf{Sharpe ratio}}\tabularnewline
\midrule
\midrule 
{\tiny M-V ($\alpha=0.001$)} & {\tiny 21.753985} & {\tiny 0.000454} & {\tiny 1.000515} & {\tiny 90.837290}\tabularnewline
{\tiny M-V ($\alpha=0.06$)} & {\tiny 13.955013} & {\tiny 0.000389} & {\tiny 1.000477} & {\tiny 75.186736}\tabularnewline
\end{tabular*}
\par\end{centering}
\begin{centering}
\begin{tabular*}{12cm}{@{\extracolsep{\fill}}>{\raggedright}p{2.5cm}>{\raggedright}p{2cm}>{\raggedright}p{2cm}>{\raggedright}p{2cm}>{\raggedright}p{2cm}}
\toprule 
\multicolumn{5}{c}{{\tiny Experiment 9 with risk aversion $\alpha\in\mathcal{A}$}}\tabularnewline
\midrule
{\tiny UC (2 unbalanced base sets)} & {\tiny 21.991619} & {\tiny 0.000456} & {\tiny 1.000552} & {\tiny 72.286540}\tabularnewline
{\tiny UC (2 base sets)} & {\tiny 22.882642} & {\tiny 0.000462} & {\tiny 1.000604} & {\tiny 59.289323}\tabularnewline
{\tiny UC (4 base sets)} & {\tiny 23.847455} & {\tiny 0.000468} & {\tiny 1.000627} & {\tiny 56.137547}\tabularnewline
{\tiny UC (6 base sets)} & {\tiny 24.360797} & {\tiny 0.000471} & {\tiny 1.000633} & {\tiny 55.534858}\tabularnewline
{\tiny WAE} & {\tiny 17.919960} & {\tiny 0.000426} & {\tiny 1.000551} & {\tiny 63.130386}\tabularnewline
\end{tabular*}
\par\end{centering}
\begin{centering}
\begin{tabular*}{12cm}{@{\extracolsep{\fill}}>{\raggedright}p{2.5cm}>{\raggedright}p{2cm}>{\raggedright}p{2cm}>{\raggedright}p{2cm}>{\raggedright}p{2cm}}
\toprule 
\multicolumn{5}{c}{{\tiny Experiment 10 with risk aversion $\alpha\in\mathcal{A}$ and
6 base sets partition}}\tabularnewline
\midrule
{\tiny UC-L (30\%)} & {\tiny 29.405695} & {\tiny 0.000499} & {\tiny 1.000725} & {\tiny 46.965682}\tabularnewline
{\tiny UC-L (50\%)} & {\tiny 27.598225} & {\tiny 0.000490} & {\tiny 1.000697} & {\tiny 49.072607}\tabularnewline
{\tiny UC-W (30\%)} & {\tiny 21.596212} & {\tiny 0.000453} & {\tiny 1.000578} & {\tiny 63.439990}\tabularnewline
{\tiny UC-W (50\%)} & {\tiny 22.317208} & {\tiny 0.000458} & {\tiny 1.000594} & {\tiny 60.699273}\tabularnewline
\bottomrule
\end{tabular*}
\par\end{centering}
\centering{}%
\begin{minipage}[t]{12cm}%
\begin{spacing}{0.5}
{\tiny\textbf{Note}}{\tiny . The UC, UC-L, and UC-W strategies are
approximated by 2001 discretization points by step 0.0005, 8605 discretization
points by step 0.0277 and 14831 discretization points by step 0.0666
for the partitions into 2, 4 and 6 base sets, correspondingly. Two
M-V strategies corresponding to the risk aversion levels 0.001 and
0.06 are the best and the worst ones in terms of the final cumulative
wealth.}
\end{spacing}
\end{minipage}
\end{table}
\vspace{-2ex}

Performance measures of the UC strategy, along with a WAE strategy
applied to all component M-V strategies, are presented in Figure \ref{Figure 10}
and Table \ref{Table 3}. The results demonstrate that the WAE strategy
performs poorly in comparison to the UC strategy across all partitions,
despite its higher Sharpe ratio. This is because the growth rate of
the WAE strategy is bounded by that of the best component M-V strategy.
Overall, all partitions lead to the outperformance of the UC strategy
over the best M-V strategy in terms of cumulative wealth. However,
the partition into unbalanced base sets shows relatively lower cumulative
wealth compared to more evenly segmented base sets. Furthermore, increasing
the number of base sets while reducing the size of each set simultaneously
improves both the average growth rate, implying higher final cumulative
wealth, and the average return, but results in a lower Sharpe ratio.
These findings align with the characteristics observed in the previous
eight experiments of the UC strategy. In more detail, regarding cumulative
wealth, the UC strategy with two unbalanced base sets merely matches
the cumulative wealth of the best M-V strategy after approximately
$6600$ time periods, while the UC strategy with two balanced base
sets eventually exceeds it. After the same amount of time, the UC
strategies with four and six base sets exceed the best M-V strategy,
with slightly more favorable cumulative wealth for the latter partitioning.
This distinction is clearly depicted in Figure \ref{Figure 11}.\smallskip

In the first row of Figure \ref{Figure 11}, it is evident that a
higher number of base sets leads to a slower convergence speed for
the growth rates of the benchmark and UC strategies. Conversely, partitioning
into smaller-sized base sets can rapidly lift the difference $W_{n}\big(b_{n}^{*}\big)-W_{n}\big(\ddot{b}_{n}\big)$
above zero over time, enabling the UC strategy to exceed the baseline
strategy faster. However, this difference may be impeded from exceeding
zero by partitioning into strongly unbalanced base sets, as observed
in the first upper graph. This phenomenon is caused by the behaviors
of the representative strategies of the respective base sets, as a
base set may contain numerous M-V strategies with markedly different
behaviors from the others in the same set. Additionally, the second
row of Figure \ref{Figure 11} illustrates the upper bound for the
difference $W_{n}\big(\ddot{b}_{n}\big)-W_{n}\big(b_{n}^{\text{UC}}\big)$,
established in Proposition \ref{infinite strategy}, with respect
to the numerically approximated UC strategy. It shows that, in general,
the average difference over time converges to zero more slowly when
the number of base sets is higher. The results of Experiment 9 emphasize
the importance of appropriately tuning the partition for the set $\mathcal{A}$,
such that the number of base sets should not be large, while each
base set should mainly contain component strategies that behave in
a similar manner.
\begin{figure}[H]
\begin{centering}
\begin{tabular}{cccc}
\includegraphics[scale=0.27]{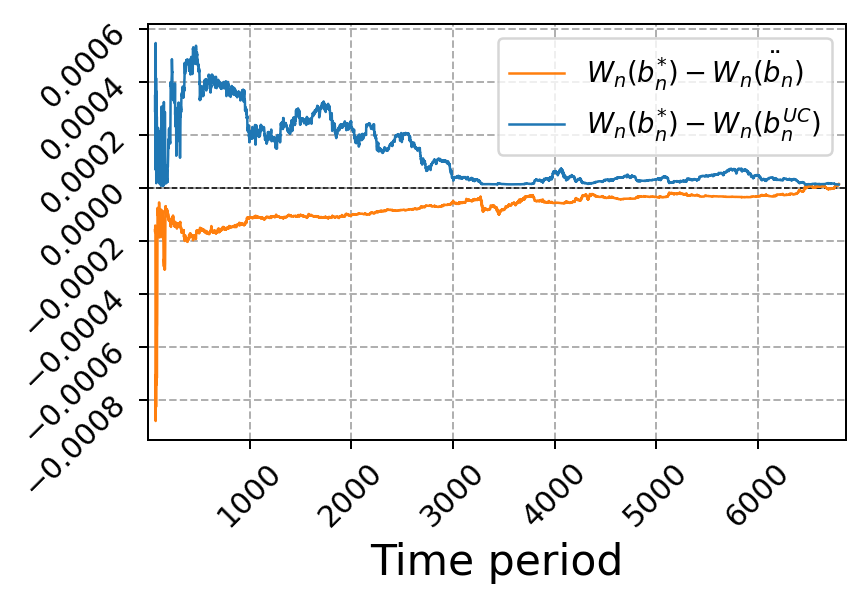} & \includegraphics[scale=0.27]{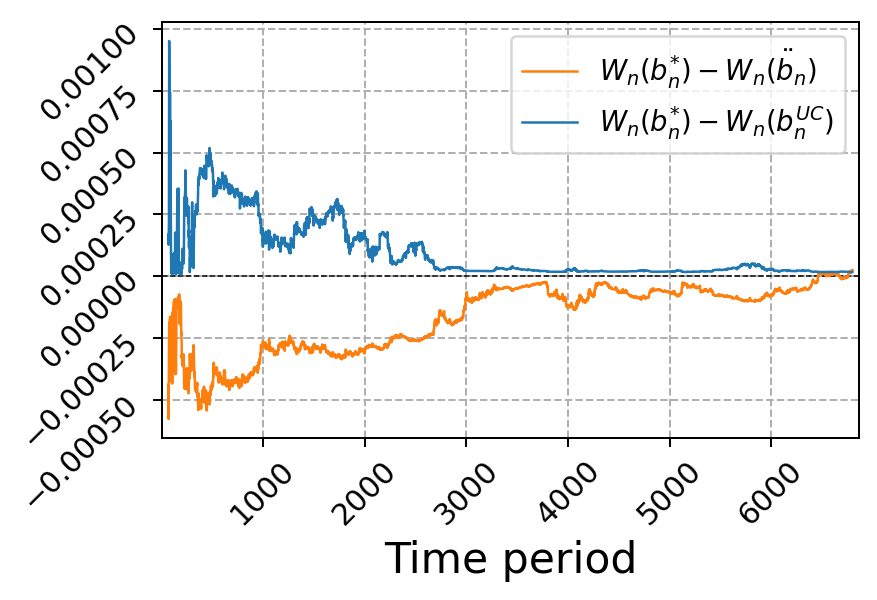} & \includegraphics[scale=0.27]{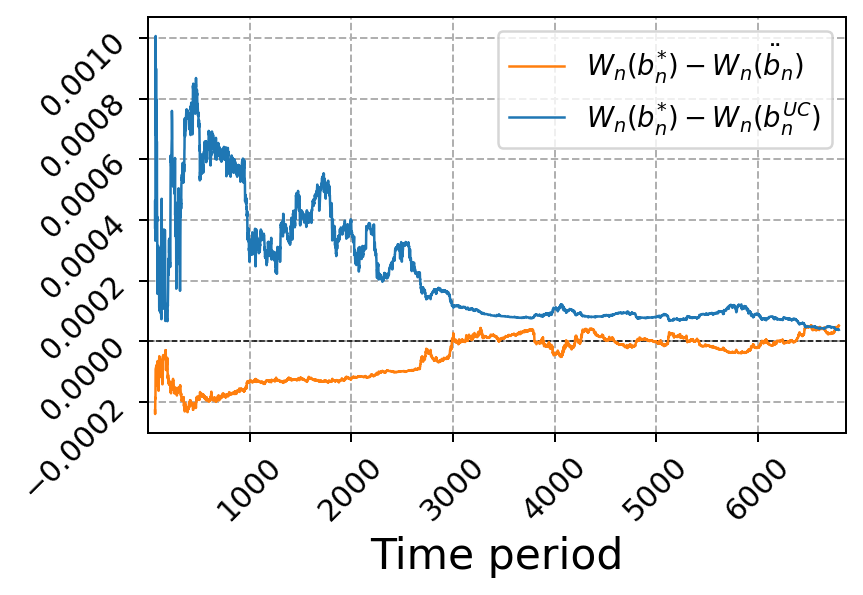} & \includegraphics[scale=0.27]{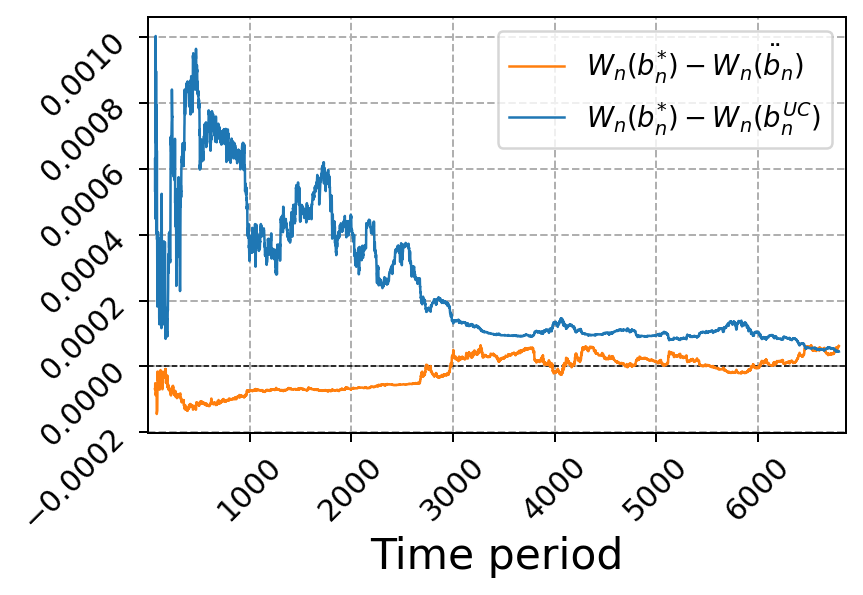}\tabularnewline
~\includegraphics[scale=0.28]{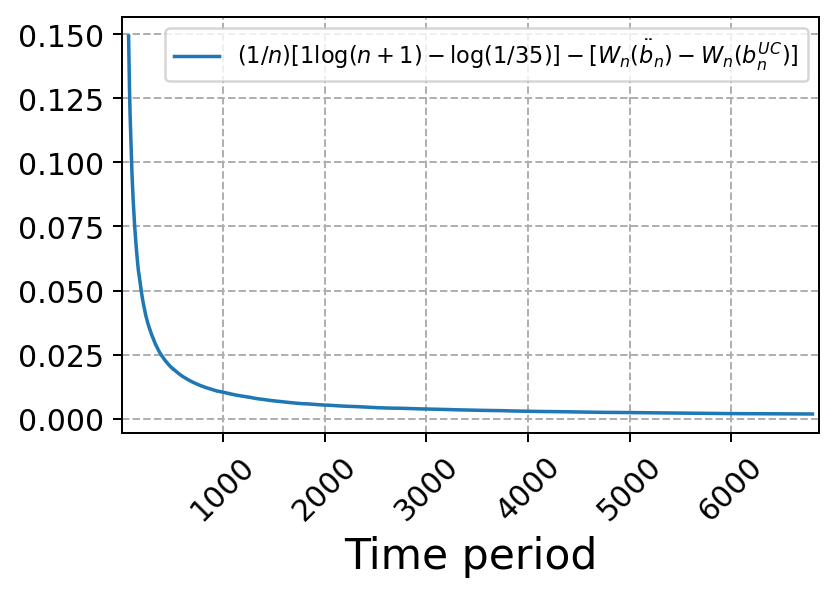} & ~\includegraphics[scale=0.28]{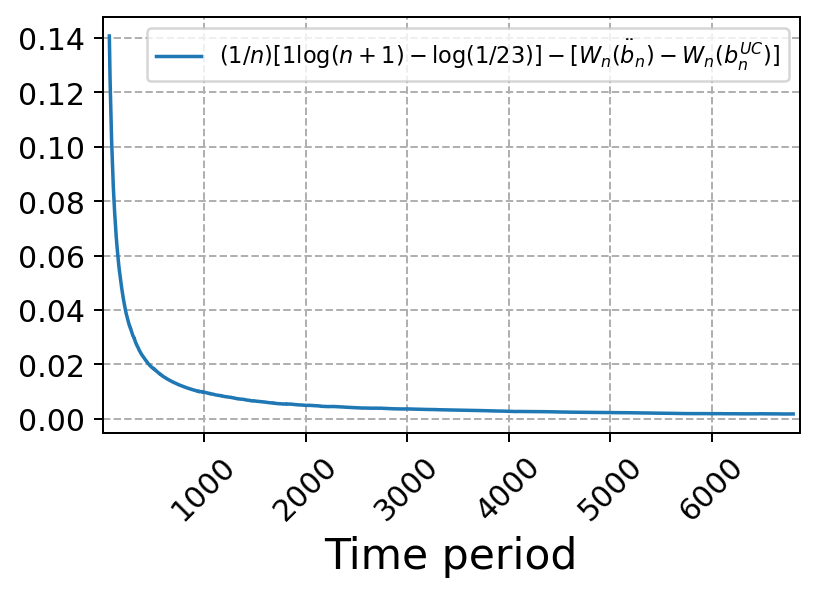} & ~\includegraphics[scale=0.28]{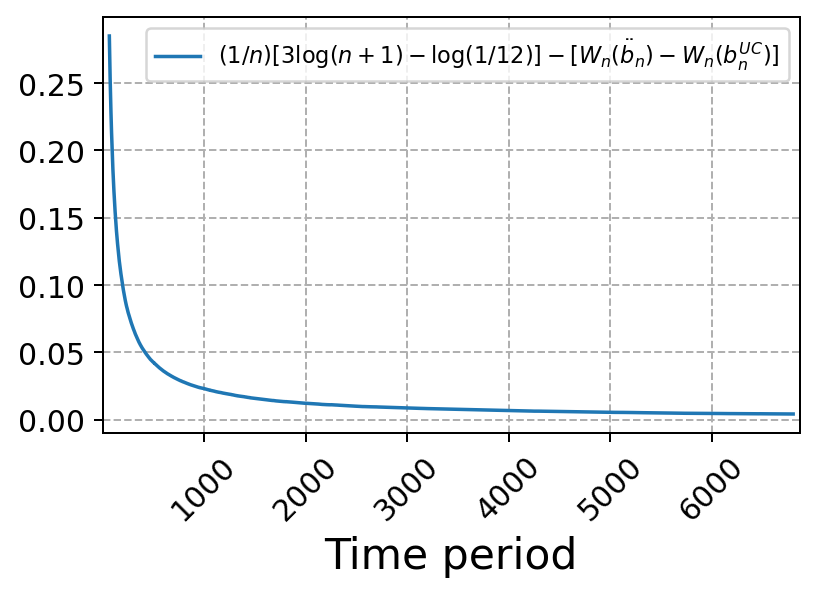} & ~\includegraphics[scale=0.28]{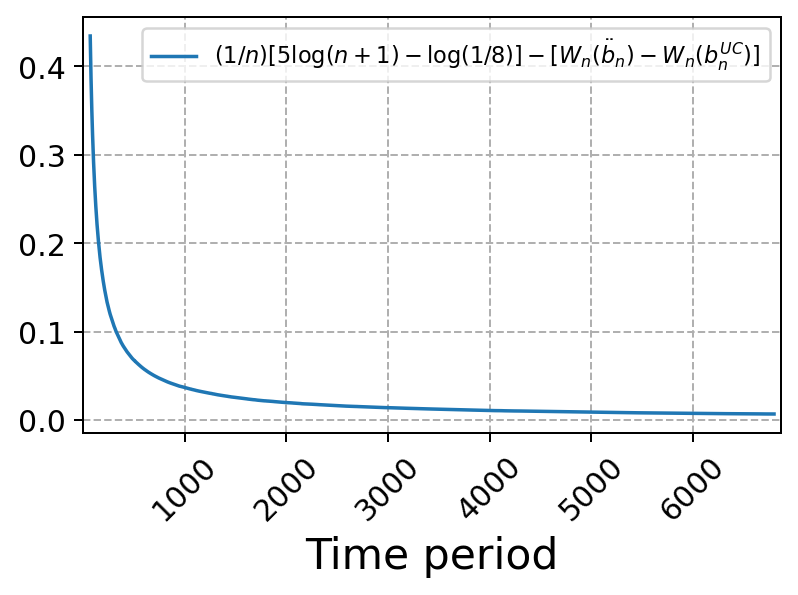}\tabularnewline
{\footnotesize ~~~2 unbalanced base sets} & {\footnotesize ~~~~~2 base sets} & {\footnotesize ~~~~~4 base sets} & {\footnotesize ~~~~~6 base sets}\tabularnewline
\end{tabular}
\par\end{centering}
{\footnotesize\caption{Upper row: differences over time between the growth rates of the strategies
$\big(\ddot{b}_{n}\big)$, $\big(b_{n}^{*}\big)$ and $\big(b_{n}^{\text{UC}}\big)$.
Lower row: upper bound over time according to Proposition \ref{infinite strategy}.
\label{Figure 11}}
}{\footnotesize\par}
\end{figure}
\begin{figure}[H]
\begin{centering}
\begin{tabular}{cccc}
\includegraphics[scale=0.235]{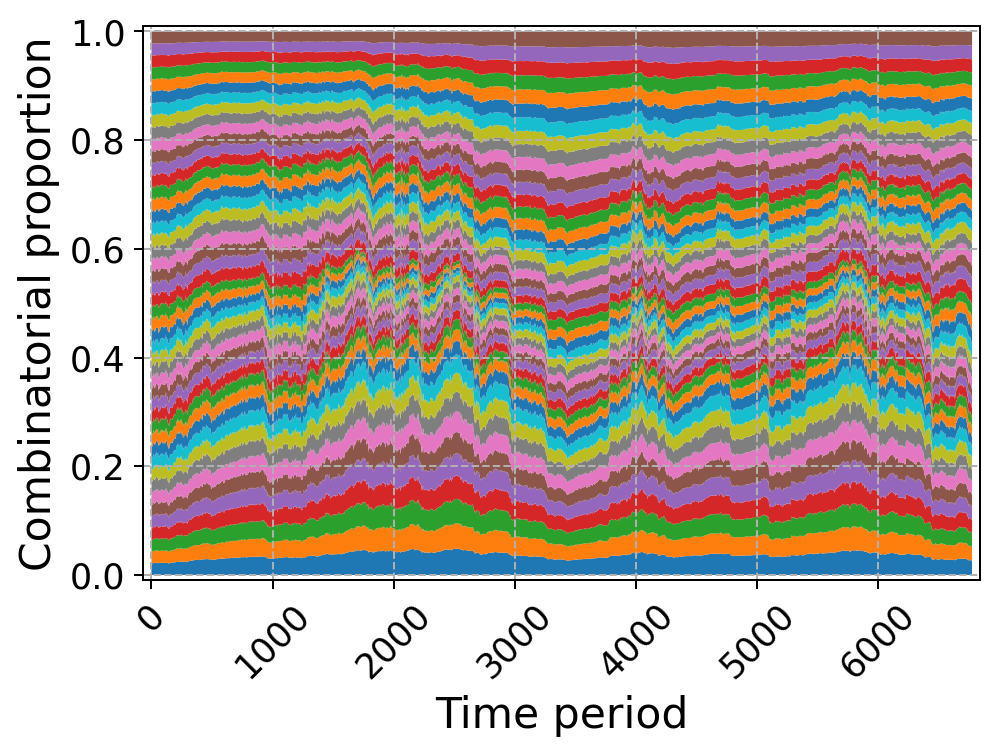} & \includegraphics[scale=0.235]{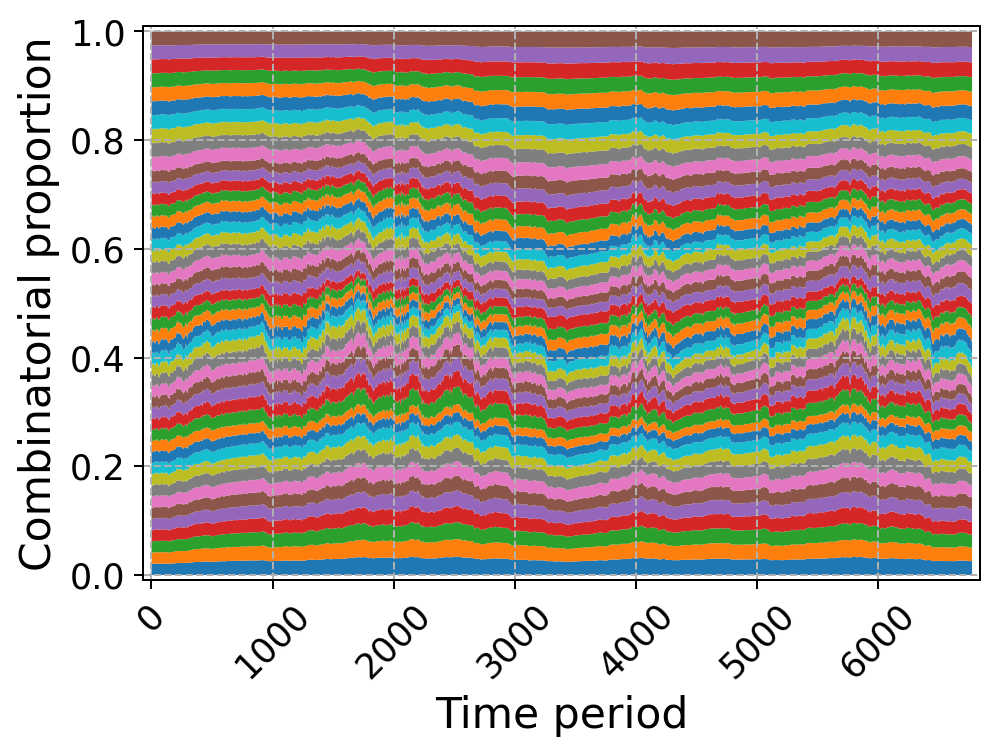} & \includegraphics[scale=0.235]{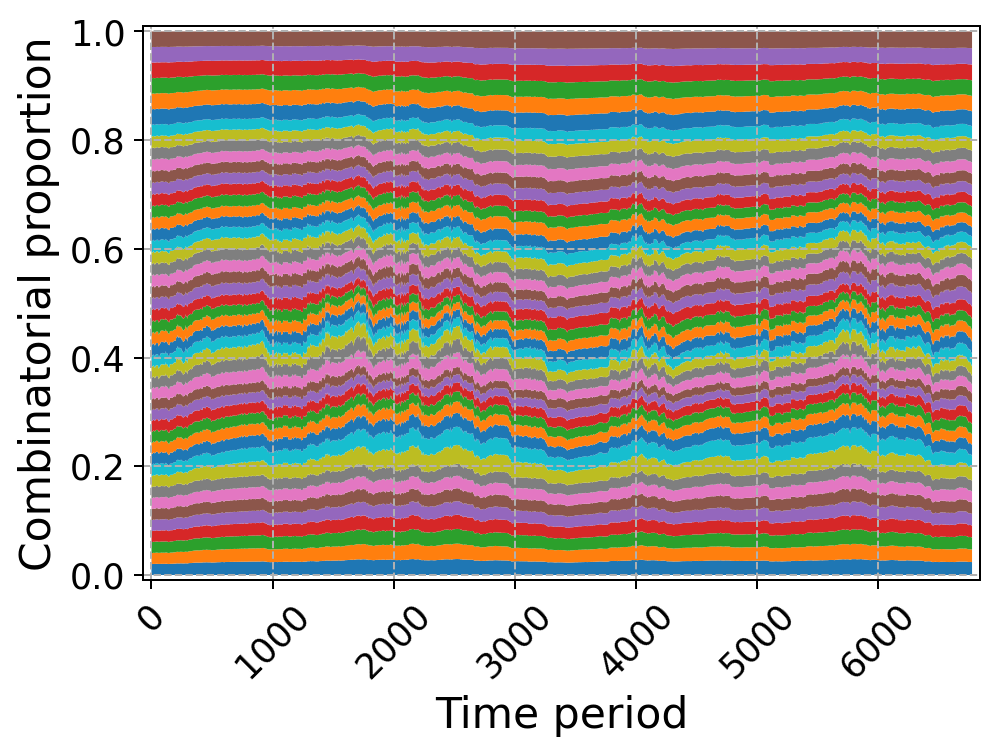} & \includegraphics[scale=0.235]{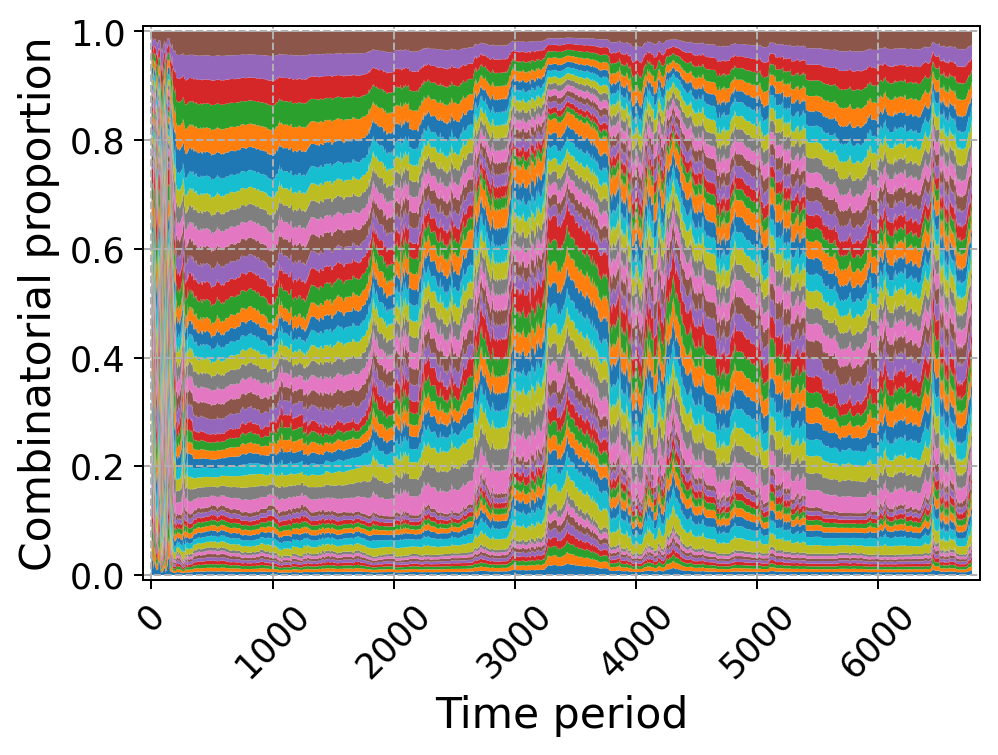}\tabularnewline
{\footnotesize ~~~~UC: 2 base sets} & {\footnotesize ~~~~UC: 4 base sets} & {\footnotesize ~~~~UC: 6 base sets} & {\footnotesize ~~~UC-L (30\%): 6 base sets}\tabularnewline
\end{tabular}
\par\end{centering}
\caption{Allocations of capital proportions into the component M-V strategies
over time (from bottom-to-top: risk aversion levels from $0.001$
to $10$).\label{Figure 12}}
\end{figure}

In Experiment 10, we conduct the UC-L and UC-W strategies alongside
the UC strategy using a partition of six base sets. The computation
procedure for the UC-L and UC-W strategies remains unchanged, except
for the replacement of component M-V strategies with representative
strategies. The results of this experiment are presented in Figure
\ref{Figure 11} and Table \ref{Table 3}. Similar to the case with
a small scale of component strategies, the UC-L strategy with a small
set $\bar{B}_{n}$ exhibits higher cumulative wealth and average return
compared to the other two, albeit with a slight decrease in the Sharpe
ratio. Figure \ref{Figure 12} illustrates the allocation proportions
of capital to each component M-V strategy over time by the UC-L ($30\%$)
strategy, in comparison to the original UC strategy with all partitions
except the unbalanced one. The graphs show that the proportions experience
lower variability as the number of base sets increases, while demonstrating
stronger fluctuations under acceleration for the UC strategy. This
phenomenon can be explained using the results from the experiments
conducted with small scales of risk aversion levels. Specifically,
with fewer base sets, a larger number of component strategies are
adjusted relative to the weights assigned to their corresponding representative
strategies, and vice versa. Moreover, changes in the weights of each
representative strategy are more gradual when their number is higher,
as depicted in the second row of Figure \ref{Figure 7}. Therefore,
the extent of fluctuation in the proportions of the UC strategy decreases
as the number of base sets increases. However, since acceleration
causes strong variations in the weights of representative strategies,
the proportions of the component strategies also exhibit significant
fluctuations, as shown in the last graph for the UC-L strategy.

\section{Summary and concluding Remarks }

In this paper, we consider an investor in the market who prioritizes
wealth accumulation and possesses the ability to learn several observable
strategies represented as sequences of no-short portfolios, which
are then ensembled into a multi-strategy. We establish a distribution-free
preference framework based on asymptotic growth rate to ensure that
the multi-strategy eventually exceeds all component strategies simultaneously
in terms of cumulative wealth, provided it is strictly preferred to
the baseline strategy as a merger of all components. Moreover, given
the unknown possibilities of the infinite sequence of market assets'
returns, the combinatorial strategy must exhibit weak preference everywhere
with strict preference at some instances of the sequence compared
to the baseline strategy. Subsequently, to satisfy the combination
criterion, we propose an online learning combinatorial strategy applicable
to any number of component strategies.\smallskip

In the experiment section, our proposed combinatorial strategy undergoes
testing alongside other online learning multi-strategies in various
designed scenarios using real data of six assets spanning three decades
from CRSP. We numerically compute the proposed combinatorial strategy
and demonstrate its ability to meet the investor's objective of exceeding
all component Mean-Variance strategies, as per the theoretical guarantees.
It exhibits a superior growth rate and cumulative wealth compared
to other online learning strategies in the experiments, albeit with
a slight loss in the Sharpe ratio. Additionally, in experiments involving
a large scale of component strategies, the time required for our combinatorial
strategy to surpass the best one is longer than that in experiments
with a small scale, consistent with theoretical analysis. Furthermore,
to expedite the speed and magnitude of surpassing the cumulative wealth
of the best component strategy, we propose an accelerated variant
of the original combination. While this accelerated combinatorial
strategy lacks a theoretical guarantee, our analysis of the experimental
results demonstrates its empirical efficiency over the original strategy
across any scale of component strategies.\smallskip

There are certain caveats to consider when applying the proposed combinatorial
strategy, particularly in tuning it for a large scale of component
strategies. If there exists a single component strategy that dominates
all others in terms of returns, our strategy will not exceed the cumulative
wealth of the baseline one; however, the accelerated variant may achieve
this. Nonetheless, such cases are often uncommon and can be addressed
by adding more component strategies so that there are at least two
competitive component strategies in terms of growth. In cases where
the majority of the components perform poorly, the time it takes for
the combinatorial strategy to exceed the best one among them will
be longer, as it may experience a significant decrease before gaining
enough momentum to recover. Similarly, when tuning the partition for
a large scale of component strategies, it is crucial to ensure that
the number of base sets is sufficient, with each base set including
a majority of component strategies exhibiting similar behavior. In
the context of component M-V strategies, appropriate tuning involves
partitioning into evenly-sized base sets that include risk aversion
levels within a small local range. In any case, acceleration can help
improve the magnitude and speed of exceeding the cumulative wealth
of the best component strategy.

\pagebreak
\pagebreak{}

\appendix
\section*{Appendix}
\setcounter{equation}{0}
\renewcommand{\theequation}{A.\arabic{equation}} 

\subsubsection*{Proof of Proposition \ref{Combination-criterion}}

Since $\big(b_{n}\big)\succ\big(\ddot{b}_{n}\big)$ implies $\big(b_{n}\big)\succsim\big(\ddot{b}_{n}\big)\land\big(b_{n}\big)\nsim\big(\ddot{b}_{n}\big)$
and by using the definition of supremum limit corresponding to the
preference $\big(b_{n}\big)\succsim\big(\ddot{b}_{n}\big)$, we have:
\[
\forall\epsilon>0,\,\exists N_{\epsilon}\text{ such that }\sup_{n\geq N}\big(W_{n}\big(\ddot{b}_{n}\big)-W_{n}\big(b_{n}\big)\big)<\epsilon,\,\forall N\geq N_{\epsilon}.
\]
Moreover, combining with the preference $\big(b_{n}\big)\nsim\big(\ddot{b}_{n}\big)$,
it must result in the following:
\[
\forall N\geq N_{\epsilon},\,\exists n_{N}\geq N\text{ such that }W_{n_{N}}\big(\ddot{b}_{n_{N}}\big)-W_{n_{N}}\big(b_{n_{N}}\big)<0,
\]
which is equivalent to $S_{n_{N}}\big(\ddot{b}_{n_{N}}\big)<S_{n_{N}}\big(b_{n_{N}}\big)$,
i.e., the combinatorial strategy $\left(b_{n}\right)$ is not always
exceeded by the baseline strategy $\big(\ddot{b}_{n}\big)$. Otherwise,
if there exists $M\geq1$ such that $W_{n}\big(\ddot{b}_{n}\big)>W_{n}\big(b_{n}\big)$
for any $n\geq M$, then it violates the assumption $\big(b_{n}\big)\nsim\big(\ddot{b}_{n}\big)$
as:
\[
\epsilon>W_{n}\big(\ddot{b}_{n}\big)-W_{n}\big(b_{n}\big)\geq0,\,\forall n\geq\max\left\{ M,N_{\epsilon}\right\} ,
\]
which implies $W_{n}\big(\ddot{b}_{n}\big)\to W_{n}\big(b_{n}\big)$.

\subsubsection*{Proof of Proposition \ref{S(lamda Y) =00003D S(lamda gamma)}}

On the one hand, $\max_{\gamma}S_{n}\big(b_{n}^{\gamma}\big)\leq\max_{\lambda}S_{n}\big(b_{n}^{\lambda}\big)$
for all $n$ since $\mathcal{H}\subseteq\mathcal{\big[}k\big]$. On
the other hand, since:
\[
\big<b_{n}^{\alpha^{j}},x_{n}\big>\geq\big<b_{n}^{\alpha},x_{n}\big>,\text{ \ensuremath{\forall\alpha\in\mathcal{A}^{j}},}\,j\in\left[h\right],
\]
then we have $S_{n}\big(b_{n}^{\lambda_{n}}\big)\leq S_{n}\big(b_{n}^{\gamma_{n}}\big)$
for all $n$ as follows:
\[
\max_{\lambda\in\mathcal{B}^{k}}{\displaystyle \prod_{t=1}^{n}\sum_{\alpha=1}^{k}\lambda_{\alpha}}\big<b_{t}^{\alpha},x_{t}\big>\leq\max_{\lambda\in\mathcal{B}^{k}}\prod_{t=1}^{n}\sum_{j=1}^{h}\Big(\sum_{\alpha\in\mathcal{A}^{j}}\lambda_{\alpha}\Big)\big<b_{t}^{\alpha^{j}},x_{t}\big>\leq\max_{\gamma\in\mathcal{B}^{h}}\prod_{t=1}^{n}\sum_{j=1}^{h}{\displaystyle \gamma}_{j}\big<b_{t}^{\alpha^{j}},x_{t}\big>,\forall n.
\]

\subsubsection*{Proof of Proposition \ref{finite strategy}}

By using the equality in (\ref{eq:3.5}), we have:
\begin{align*}
{\displaystyle \max_{x_{1}^{n}}\log\dfrac{S_{n}\big(b_{n}^{\lambda_{n}}\big)}{S_{n}\big(b_{n}\big)}} & =\max_{x_{1}^{n}}\log\dfrac{{\displaystyle \prod_{t=1}^{n}}{\displaystyle \sum_{\alpha=1}^{k}}\lambda_{n,\alpha}\big<b_{t}^{\alpha},x_{t}\big>}{{\displaystyle \int_{\mathcal{B}^{k}}}{\displaystyle \prod_{t=1}^{n}}{\displaystyle \sum_{\alpha=1}^{k}\lambda_{\alpha}}\big<b_{t}^{\alpha},x_{t}\big>\mu\left(\lambda\right)d\lambda}=\max_{x_{1}^{n}}\log\dfrac{{\displaystyle \prod_{t=1}^{n}}\big<\lambda_{n},r_{t}\big>}{{\displaystyle \int_{\mathcal{B}^{k}}}{\displaystyle \prod_{t=1}^{n}}\big<\lambda,r_{t}\big>\mu\left(\lambda\right)d\lambda},
\end{align*}
where the vectors $\big(\big<b_{t}^{1},x_{t}\big>,...,\big<b_{t}^{k},x_{t}\big>\big)\eqqcolon\big(r_{t,1},...,r_{t,k}\big)\eqqcolon r_{t}\in\mathbb{R}_{++}^{k}$
represent the respective returns of $k$ component portfolios at time
$t$. Let $r_{1}^{n}\coloneqq\big\{ r_{1},...,r_{n}\big\}\in\mathbb{R}_{++}^{k\times n}$
denote the sequence of $n$ vectors of portfolios' returns. Then,
by applying the theorem of the upper bound for the worst-case of the
Universal portfolio in \citet{Cover1991,Cover1996a}, with replacing
the sequence of assets' returns by $r_{1}^{n}$, we obtain the upper
bound with uniform density $\mu\left(\lambda\right)$ as:
\begin{align*}
\max_{x_{1}^{n}}\log\dfrac{{\displaystyle \prod_{t=1}^{n}}\big<\lambda_{n},r_{t}\big>}{{\displaystyle \int_{\mathcal{B}^{k}}}{\displaystyle \prod_{t=1}^{n}}\big<\lambda,r_{t}\big>\mu\left(\lambda\right)d\lambda} & =\max_{r_{1}^{n}}\log\dfrac{{\displaystyle \prod_{t=1}^{n}}\big<\lambda_{n},r_{t}\big>}{{\displaystyle \int_{\mathcal{B}^{k}}}{\displaystyle \prod_{t=1}^{n}}\big<\lambda,r_{t}\big>\mu\left(\lambda\right)d\lambda}\leq\left(k-1\right)\log\left(n+1\right),\forall n.
\end{align*}
Thus, ${\displaystyle \max_{x_{1}^{n}}\left(\log S_{n}\left(b_{n}^{*}\right)-\log S_{n}\left(b_{n}\right)\right)\leq\left(k-1\right)\log\left(n+1\right)}$
for all $n$, as needed to show.

\subsubsection*{Proof of Proposition \ref{infinite strategy}}

On the one hand, similar to equation in (\ref{eq:3.5}), we have the
following by telescoping:
\begin{align*}
S_{n}\big(\hat{b}_{n}^{i}\big) & ={\displaystyle {\displaystyle \prod_{t=1}^{n}\dfrac{{\displaystyle \sum_{\mathcal{A}^{i}}\big<{\displaystyle b_{n}^{\alpha}},x_{n}\big>}S_{t-1}\big(b_{t-1}^{\alpha}\big)\mu^{i}\big(\alpha\big)}{{\displaystyle \sum_{\mathcal{A}^{i}}}{\displaystyle S_{t-1}\big(b_{t-1}^{\alpha}\big)\mu^{i}\big(\alpha\big)}}}={\displaystyle \sum_{\mathcal{A}^{i}}}{\displaystyle S_{n}\left(b_{n}^{\alpha}\right)\mu^{i}\big(\alpha\big)},\,\forall i\in\left[N\right],\forall n.}
\end{align*}
Then, we have the following bound for every representative strategy:
\begin{align}
{\displaystyle \log\dfrac{S_{n}\big(\hat{b}_{n}^{i}\big)}{S_{n}\big(b_{n}^{\alpha_{n}^{i}}\big)}={\displaystyle \log\dfrac{{\displaystyle \sum_{\mathcal{A}^{i}}}{\displaystyle S_{n}\left(b_{n}^{\alpha}\right)\mu^{i}\big(\alpha\big)}}{S_{n}\big(b_{n}^{\alpha_{n}^{i}}\big)}}\geq{\displaystyle \log\dfrac{S_{n}\big(b_{n}^{\alpha_{n}^{i}}\big)\mu^{i}\big(\alpha_{n}^{i}\big)}{S_{n}\big(b_{n}^{\alpha_{n}^{i}}\big)}}} & \geq\log\mu^{i}\big(\alpha_{n}^{i}\big),\text{\ensuremath{\forall i}\ensuremath{\in}\ensuremath{\left[N\right]},\ensuremath{\forall n,}}\nonumber \\
 & \geq\log\epsilon_{n},\text{\ensuremath{\forall i}\ensuremath{\in}\ensuremath{\left[N\right]},\ensuremath{\forall n.}}\label{eq:proof Proposition 4 B}
\end{align}
On the other hand, by Proposition \ref{finite strategy} and the property
of the benchmark strategy for this context, namely $S_{n}\big(b_{n}^{*}\big)\geq\max_{i\in\left[N\right]}S_{n}\big(\hat{b}_{n}^{i}\big)$
for all $n$, we obtain the needed upper bound as follows:
\begin{align*}
\left(N-1\right)\log\left(n+1\right)\geq\max_{x_{1}^{n}}\big(\log S_{n}\big(b_{n}^{*}\big)-\log S_{n}\big(b_{n}\big)\big) & \geq\max_{x_{1}^{n}}\max_{i\in\left[N\right]}\big(\log S_{n}\big(\hat{b}_{n}^{i}\big)-\log S_{n}\big(b_{n}\big)\big)\\
 & \geq\max_{x_{1}^{n}}\big(\log S_{n}\big(\ddot{b}_{n}\big)-\log S_{n}\big(b_{n}\big)\big)+\log\epsilon_{n},\forall n.
\end{align*}
The last inequality follows from (\ref{eq:proof Proposition 4 B})
and the fact that $S_{n}\big(\ddot{b}_{n}\big)=\max_{i\in\left[N\right]}S_{n}\big(b_{n}^{\alpha_{n}^{i}}\big)$
by definition.
\end{document}